\documentclass[reqno,10pt]{amsart}
\usepackage[margin=1in]{geometry}

\usepackage[utf8]{inputenc}
\usepackage{amsmath,amssymb,amsthm, comment, mathtools}
\usepackage{amsthm}
\usepackage{mathrsfs}
\usepackage{slashed}
\usepackage{units}
\usepackage{color}

\usepackage{amsmath,amsthm}
\usepackage {latexsym}
\usepackage{amssymb,mathtools}
\usepackage{slashed}
\usepackage[linktocpage=true]{hyperref}

\newcommand{\les}{\lesssim}
\newcommand{\bea}{\begin{eqnarray}}\newcommand{\eea}{\end{eqnarray}}
\newcommand{\beq}{\begin{equation}}\newcommand{\ee}{\end{equation}}
\newcommand{\beqs}{\begin{equation*}}\newcommand{\eqs}{\end{equation*}}

\newcommand{\eps}{{\varepsilon}}

\newcommand{\BS}{{\mathbb S}}

\newcommand{\la}{\langle}

\renewcommand{\b}{\beta}
\newcommand\Ga{{\Gamma}}
\def\la{{\lambda}}
\newcommand\G{{\mathcal G}}
\newcommand\tG{\tilde{\mathcal G}}
\newcommand\g{\tilde g}
\newcommand\Gat{\tilde{\Ga}}
\newcommand\tpi{\tilde{\pi}}
\newcommand\tla{\tilde{\lambda}}

\newtheorem{theorem}{Theorem}[section]
\newtheorem{lemma}[theorem]{Lemma}

\newtheorem{prop}[theorem]{Proposition}
\newtheorem{proposition}[theorem]{Proposition}

\theoremstyle{remark}
\newtheorem{remark}[theorem]{Remark}
\newtheorem*{remark*}{Remark}
\def\a{\alpha}
\def\ga{\gamma}\def\de{\delta}

\def\bm{\left( \begin{array}{cc}}
\def\endm{\end{array}\right)}\newcommand{\eq}{\end{equation}}

\def\tr{\text{tr}}
\def\a{\alpha}\def\b{\beta}
\def\ga{\gamma}\def\de{\delta}
\def\pa{\partial}

\def \rectangle#1#2{\hbox{\vrule\vbox to #2
{\hrule\hbox to #1{\hfil}\vfil\hrule}\vrule}}

\def\Lb{\underline{L}}

\def\tr{\text{tr}}
\def\a{\alpha}\def\b{\beta}\def\ga{\gamma}
\def\de{\delta}\def\Boxr{\widetilde{\square}}
\def\pa{\partial}

\def\Lb{\underline{L}}

\def\pas{\text{$\pa\mkern -10.0mu$\slash\,}}

\def\sls#1{\text{$#1\mkern -13.0mu$\slash\,}}

\def\us{u^{\!*}}
\def\Ss{{{S\underline{}}^*}}
\def\Ls{{{L\underline{}\!}^*}}\def\Lbs{{{\underline{L\!}}^*}}

\def\Lh{{\hat{L\underline{}}}}\def\Lbh{{\hat{\underline{L\!}\,}}}
\def\Lb{\underline{L\!}}\def\uL{\underline{L\!}}
\def\utau{\underline{\tau\!}\,}\def\uchi{\underline{\chi}}

\numberwithin{equation}{section}

\begin{document}

\title {Masses at null infinity for Einstein's equations in harmonic coordinates}
\dedicatory{Dedicated to Demetrios Christodoulou on the occasion of his 70th birthday}
\author{Lili He}
\address{Department of Mathematics \\ Johns Hopkins University \\ Baltimore, MD 21218, USA}
\email{lhe31@jhu.edu}

\author{Hans Lindblad}
\address{Department of Mathematics \\ Johns Hopkins University \\ Baltimore, MD 21218, USA}
\email{lindblad@math.jhu.edu}

\maketitle
\begin{abstract}
	In this work we give a complete picture of how to in a direct simple way define the mass at null infinity in harmonic coordinates in three different ways that we show satisfy the Bondi mass loss law. The first and second way involve only the limit of metric (Trautman mass) respectively the null second fundamental forms along asymptotically characteristic surfaces (asymptotic Hawking mass) that only depend on the ADM mass. The last involves construction of special characteristic coordinates at null infinity (Bondi mass). The results here rely on asymptotics of the metric derived in \cite{Lind17}.
\end{abstract}

\section{Introduction} The first definition of mass at null infinity was given by Trautman \cite{T58}. In Trautman's definition, the mass is defined as the integral of the so called superpotential which is expressed in terms of the metric and its first order derivatives over the spheres receding to null infinity, see \S\ref{subsec:Trautman} for precise definition. We refer the readers to \cite{T62, T02, B58, G58} and \cite[\S 3.1]{BC17} and references therein for more details.

In 1960, Bondi \cite{B60} introduced a new approach which was based on the outgoing null rays to study the gravitational waves. Later, Bondi, Metzner and van der Burg \cite{BMM62} considered the axisymmetric spacetimes. Soon after, Sachs \cite{S62} generalized the formalism to non axisymmetric spacetimes. In the Bondi-Sachs formalism, the coordinates which are called Bondi-Sachs coordinates, are adapted to the null geodesics of the space time. With respect to such a coordinate system, only $6$ metric quantities are needed to describe the spacetime, and the Bondi mass and radiated energy at null infinity are defined in terms of certain lower order terms of these metric components. 
Hintz-Vasy \cite{HV20}  showed the existence of the Bondi-Sachs coordinates for a specific class of initial data and identified the Bondi mass in a generalized wave coordinates. 

Christodoulou \cite{C91} introduced an alternative approach to defining the mass at null infinity without the need to use the Bondi-Sachs coordinates. The definition was given as the limit of the Hawking mass of a family of spheres that converge to a round metric sphere along the outgoing null hypersurfaces towards null infinity. Christodoulou proved that the limit of the Hawking mass exists and satisfies a mass loss law for the initial data used in \cite{CK93} by analyzing the null structure Einstein equations. Later on, the limit of Hawking mass of suitable spheres was analyzed in the settings of the work \cite{B10, KN03}.

 These three masses are defined completely differently and each has been analyzed in several settings. As summarized in \cite{BC17}, in the setting where the Bondi-Sachs formalism can be carried out, the limit of the Hawking mass along suitable family of spheres recovers the Trautman mass and Bondi mass. However, the notion of the mass and radiated energy at null infinity in harmonic coordinates remains to be clarified. In this work we give a complete picture of how to define the mass at null infinity in harmonic coordinates in the different ways that we show satisfy the Bondi mass loss law and therefore coincide.

\subsection{Einstein vacuum equations in  harmonic coordinates}
Einstein's equations in harmonic are a system of nonlinear wave equations
\begin{equation}\label{eq:EisnteinWave}
	\Boxr_g\,  g_{\mu\nu} =F_{\mu\nu} (g) (\pa g, \pa g),
	\qquad\text{where}\quad
	\Boxr_g=\text{$\sum$} g^{\alpha\beta}\pa_\alpha\pa_\beta,
\end{equation}
for a Lorentzian metric $g_{\alpha\beta}$, that in addition satisfy the preserved wave coordinate condition
\begin{equation}\label{eq:WaveCordinateCond}
	\partial_\alpha \big(\sqrt{|g|} g^{\alpha\beta}\big)=0,\qquad \text{where} \quad  |g|=|\det{\big(g\big)}|.
\end{equation}
Choquet-Bruhat \cite{C52} proved local existence in these coordinates. Christodoulou-Klainerman \cite{CK93} proved global existence for Einstein's vacuum equations $R_{\mu\nu}=0$ for small asymptotically flat initial data:
\begin{equation}\label{eq:asymptoticallyflatdata}
	g_{ij}\big|_{t=0}=(1+M r^{-1})\, \delta_{ij} + o(r^{-1-\gamma}),\quad
	\pa_t g_{ij}\big|_{t=0}=o(r^{-2-\gamma}),\quad r=|x|,\qquad 0<\gamma<1,
\end{equation}
where $M\!>\!0$ by the positive mass theorem \cite{SY79, W81}.
The proof avoids using coordinates since it
was believed the metric in harmonic coordinates would blow up for large times.
John \cite{J81, J85} had noticed that solutions to some nonlinear wave equations blow up for small data, whereas Klainerman \cite{K83,K86}, see also Christodoulou \cite{C86}, came up with the ``null condition'', that guaranteed global existence for small data. However Einstein's equations do not satisfy the null condition.
The null condition provide a cancellation of the nonlinear terms so that
solutions decay like solutions of linear equations.
H\"ormander introduced a simplified asymptotic system, by neglecting angular derivatives which we expect decay faster due to the rotational invariance, to study blowup.
Lindblad \cite{L92} showed that the asymptotic system corresponding to the quasilinear part of Einstein's equations does not blow up and gave an example of a nonlinear equation of this form that have global solutions that do not decay as much.
Lindblad-Rodnianski \cite{LR03} introduced the {\it weak null condition}
requiring that the corresponding asymptotic system  have global solutions and  showed that Einstein's
equations in {\it wave coordinates} satisfy the weak null condition which was used in
\cite{LR05,LR10} to prove global existence.
Starting from the $L^2$ estimates in \cite{LR05,LR10}, Lindblad \cite{Lind17} derived more detailed asymptotics that we will rely on. We expect the result of this manuscript to be true in the presence of matter since these asymptotics can also be derived by directly using a change of coordinates which is equivalent to generalized wave coordinates as in Kauffmann-Lindblad \cite{KL21} and \cite{CKL19, H21}.

 \subsection{The characteristic surfaces}
 In order to unravel the effect of the quasilinear terms in \eqref{eq:EisnteinWave} one can change to
 characteristic coordinates as in \cite{CK93}, but this loses regularity and is not explicit.
 Instead Lindblad \cite{Lind17}  used the asymptotics of the metric to determine the characteristic surfaces asymptotically and used this to construct coordinates.
 Due to the wave coordinate condition \eqref{eq:WaveCordinateCond} the outgoing light cones
 of a solution with asymptotically flat data \eqref{eq:asymptoticallyflatdata} approach those of the Schwarzschild metric with the same mass $M$.
 In \cite{Lind17} it was shown that there is a solution to the eikonal equation that approaches the one for Schwarzschild
 \begin{equation}\label{eq:eikonal}
 	g^{\alpha\beta} \pa_\alpha u \,\pa_\beta u=0,\qquad  u\to  u^*\!\!=t-r^*\!\!,\quad \text{when}\qquad r>t/2\to\infty,\qquad\text{where } r^*\!\!=r+M\ln r+O(M/r).
 \end{equation}

 \subsection{The asymptotics of the metric}
 In \cite{Lind17} the precise asymptotics of the metric was given.  Asymptotically the metric is Minkowski metric $m_{\mu\nu}\!$ plus
 \begin{equation*}
 	h_{\mu\nu}(t,r\omega)\sim H_{\mu\nu}(\widetilde{r}\!-t,\omega)/(t+r)+K_{\mu\nu}\big(\tfrac{t+\widetilde{r}}{|\widetilde{r}\!-t|+1},\omega\big)/(t+r),
 	\qquad \widetilde{r}\!= r+M\ln{r},\qquad \omega\!=\!x/|x|.
 \end{equation*}
 Here $H$ is concentrated close to the outgoing light cones $\widetilde{q}=\widetilde{r}\!-t$ constant,
 $|H(\widetilde{q},\omega)|\!\leq\! \varepsilon (1\!+| \widetilde{q}|)^{-\gamma^\prime}\!\!$ where $\gamma'=\gamma-C\varepsilon$ for some constant $C$ and small constant $\varepsilon$,
 and $K$ is homogeneous of degree $0$ with a log singularity at the light cone
 $|K(s,\omega)|\!\leq \!\varepsilon \ln{|s|}$ for the nontangential components.
 $H$ is the radiation field of the free curved wave operator, the left of \eqref{eq:EisnteinWave}, and
 $K$ is the backscattering of the wave operator with a source term
 $F_{\mu\nu}\sim P_{\mathcal{S}}(\pa_\mu h,\pa_\nu h)$ in the right of \eqref{eq:EisnteinWave},
 where $P_{\mathcal{S}}$ is the norm of the components tangential to the spheres.
 In the wave zone
 \begin{equation*}
 	K_{\mu\nu}\big(\tfrac{t+\widetilde{r}}{|\widetilde{r}\!-t|+1},\omega\big)\sim L_\mu(\omega) L_\nu(\omega) \int_{\widetilde{r}-t}^\infty \frac{1}{2}\ln{\Big(\frac{t+\widetilde{r}+\widetilde{q}}{t-\widetilde{r}+\widetilde{q}}\Big)} n\big(\widetilde{q},\omega\big)\, d \widetilde{q} ,
 	\quad\text{when}\quad  |t-\widetilde{r}|<<t+\widetilde{r},
 \end{equation*}
 where $L_\mu\!=\!m_{\mu\nu} L^\nu\!$, in a null frame $L\!=\!(1,\omega)$, $\underline{L}\!=\!(1,-\omega)$ and orthonormal $S_1,S_2\!\in\! T(\mathbb{S}^2)$, we have
 \begin{equation}\label{eq:Newsfucntion}
 	n(\widetilde{q},\omega)=-P_{\mathcal{S}}\big(\pa_{\widetilde{q}} H,\pa_{\widetilde{q}} H\big)(\widetilde{q},\omega),\quad\text{where}\quad
 	P_{\mathcal{S}} (D,E)
 	= -D_{AB}
 	\, E^{AB}\!/\,2,\quad A,B\in\{S_1,S_2\}.
 \end{equation}
\begin{remark*}
	In this manuscript $\widetilde{\bullet},  \widehat{\bullet},\overline{\bullet}$ refer to the quantities expressed in the coordinates $(\widetilde{t}=t, \widetilde{r}\omega)$ where $\widetilde{r}=r+M\ln r$, modified asymptotically Schwarzschild null coordinates $\widehat{y}^p$ as defined in \ref{subsec:2.2} and the Bondi-Sachs coordinates respectively, unless otherwise specified.
\end{remark*}
\subsection{The Trautman mass and radiated energy} \label{subsec:Trautman}
We will use the surface $\widetilde{u}=\widetilde{t}-\widetilde{r}$ constant instead of the null cones to define the Trautman mass and radiated energy at null infinity in terms of the asymptotics of the metric components in wave coordinates.

We let ${\pa}_\mu\!=\!A_\mu^\nu\widetilde{\pa}_\nu$ and $\g^{\mu\nu\!}\!=\!A^{\mu}_{ \a}A^{\nu}_{\beta} g^{\a\b}\!$,
where $A_\mu^\nu\!=\pa \widetilde{x}^\nu\!/\pa x^\mu$ and $\widetilde{x}\!=\!\widetilde{r}\omega$ where $\widetilde{r}\!\!=\!r\!+\!M\!\ln{r}$. Let $S_{\widetilde{u}, \widetilde{r}}=\!\{(\widetilde{t},\widetilde{x}); \widetilde{t}\!=\!\widetilde{u}\!+\!\widetilde{r}\}$ be a sphere, following \cite{T58, T62, T02, BC17} we define the {\it Trautman four-momentum} as
\[
M^\alpha_T(\widetilde{u})=\lim_{r\to\infty}\frac{1}{4\pi}\int_{S_{\widetilde{u}, \widetilde{r}}}\!\tilde{\mathbb{U}}^{\alpha\beta\gamma}\,dS_{\beta\gamma}.
 \]
 Here $dS_{\beta\gamma}\! =\! n_{[\beta}k_{\gamma]}\widetilde{r}^2dS(\omega)$ with $n_\gamma\! =\! (d\widetilde{r})_\gamma\! =\! (0, \omega_i)$, $k_\beta\! =\! (d\widetilde{t})_\beta\! =\! (1, 0,0,0)$, and the superpotential $\tilde{\mathbb{U}}^{\alpha\beta\gamma}$ is
\[
\tilde{\mathbb{U}}^{\a\b\gamma}=\sqrt{\lvert\tilde{g}\rvert}\tilde{g}^{\a\mu}\tilde{\mathbb{U}}_\mu^{\b\gamma} \quad\text{where} \quad \tilde{\mathbb{U}}_\mu^{\beta\gamma}=\sqrt{\lvert \tilde{g}\rvert} \tilde{g}^{\alpha\mu}\tilde{g}^{\sigma[\rho}\delta_\mu^\gamma\tilde{g}^{\beta]\tau}\widetilde{\partial}_{\tau}\tilde{g}_{\rho\sigma}.
\]
Here the square brackets denote the antisymmetric part of a tensor, i.e., $T^{[a_1\cdots a_l]}=\sum_{\sigma}(-1)^\sigma T^{a_{\sigma(1)}\cdots a_{\sigma(l)}}$ where the sum is taken over all permutations $\sigma$ of $1,\ldots,l$ and $(-1)^\sigma$ is $1$ for even permutations and $-1$ for odd permutations. A direct computation implies
\[
\tilde{\mathbb{U}}^{\a\b\gamma}=-\tla^{\a\b\mu},\quad\text{where}\quad
\tilde{\lambda}^{\a\b\mu}=\widetilde{\pa}_\nu \big(|\g|(\g^{\a\b}
\g^{\mu\nu}-\g^{\a\mu}\g^{\b\nu})\big).
\]
Therefore with $L_\alpha=(-1, \omega_i)$ and $\Lb_\alpha=(-1, -\omega_i)$ we can write
\[
M^\alpha_T(\widetilde{u})\!\!=\!\!\frac{1}{4\pi}\int_{\mathbb{S}^2}\! m_T^\alpha(\widetilde{u},\omega){ dS(\omega)}, \qquad\text{where}\quad
m_T^\alpha(\widetilde{u},\omega)
= \lim_{\widetilde{r}\to\infty}  (\widetilde{r})^2 \big(\widetilde{\lambda}^{\alpha\beta\gamma}
\uL_\gamma L_\beta\big)(\widetilde{u}-\widetilde{r},\widetilde{r}\omega).
\]
The {\it Trautman radiated four-momentum}\footnote{For the original definition of the radiated four-momentum \cite{T58, T62, T02}, Trautman uses the mixed Einstein pseudotensor of energy and momentum. In \cite{G58} it is noted that the symmetric Landau-Lifshitz pseudotensor has the same total energy and momentum as the mixed Einstein pseudotensor.} is defined as
\[
	E^\alpha_T(\widetilde{u})=\lim_{\widetilde{r}\to\infty}\frac{1}{2\pi}\int_{S_{\widetilde{u}, r}}\!\lvert\tilde{g}\rvert\tilde{\pi}^{\alpha\beta}\,dS_{\beta}.
\]
Here $dS_\beta=n_\beta\widetilde{r}^2dS(\omega)$ with $n_\beta=(0, \omega_i)$ and $\tilde{\pi}^{\alpha\beta}$ is {\it Landau-Lifshitz} pseudotensor \cite[\S101]{LL62},
\[
\tilde{\pi}^{\a\b}=-2\tilde{G}^{\a\b}+\frac{1}{\lvert\tilde{g}\rvert}\widetilde{\pa}_\mu \tla^{\a\b\mu}\quad\text{where}\quad \tilde{G}^{\a\b}=\tilde{R}^{\a\b}-\frac{1}{2}\tilde{g}^{\a\b}\tilde{R}.
,
\]
which is a symmetric.
We write
\[
E_T^\alpha(\widetilde{u})
= \frac{1}{2\pi}\int_{\mathbb{S}^2} \Delta m^\a_T(\widetilde{u},\omega)\, d\omega\quad\text{where}\quad\Delta m^\a_T(\widetilde{u},\omega)=
\lim_{\widetilde{r}\to\infty}  \widetilde{r}^2  |\g| \tpi^{\a i} \frac{\widetilde{x}_i}{\widetilde{r}}.
\]

\begin{remark*}
Many known gravitational pseudotensors can be derived from the above superpotentials, including the mixed Einstein pseudotensor of energy and momentum and the symmetric Landau-Lifshitz pseudotensor. We refer the readers to \cite{B58, G58, T62} for a more detailed discussion of different pseusotensors and their relations.
\end{remark*}
We will refer to $M^0_T(\widetilde{u})$ as the {\it Trautman mass} and to $E^0_T(\widetilde{u})$ as the radiated energy at null infinity. Using the asymptotics in \cite{Lind17} we prove in subsection \ref{subsec:4.4}
\begin{theorem}\label{thm:1}
	The Trautman four-momentum $M^\alpha_T(\widetilde{u})$ and Trautman radiated four momentum $E_T^\alpha(\widetilde{u})$ are well defined and satisfy the mass loss law
	\[ M^\alpha_T(\widetilde{u}_2)-M^\alpha_T(\widetilde{u}_1)=-\int_{{\widetilde{u}_1}}^{\widetilde{u}_2} E^\alpha_T(\widetilde{u})\, d\widetilde{u}.
	\]
	Moreover the radiated energy $E^0_T$ at null infinity can be expressed in terms of $n$ in $\eqref{eq:Newsfucntion}$ as
	$$ E^0_T(\widetilde{u})\!\!=\!\!\int_{\mathbb{S}^2}\!n(-\widetilde{u},\omega){ dS(\omega)\!}/{8\pi} ,
	$$
	and the Trautman mass $M_T^0(\widetilde{u})\to M$, the ADM mass, as $\widetilde{u}\to-\infty$ and $M_T^0(\widetilde{u})\to0$ as $\widetilde{u}\to\infty$.
\end{theorem}

 \subsection{The asymptotic Hawking mass and radiated energy}
We will use the asymptotically null surfaces $u^*=t-r^*$ constant where $r^*=r+M\ln r+O(M/r)$ instead of null cones to define the asymptotic Hawking mass and the radiated energy at null infinity.

Following \cite{C91,C09} define the radius of a surface $S$ by $r(S)\!=\!\sqrt{\!\text{Area}(S)\!/4\pi}$.
Let $\Lh$ and $\Lbh$ be the outgoing respectively incoming
null normals to $S$ satisfying
$g(\Lh,\Lbh)\!\!=\!-2$.\!
$\Lh$ and $\Lbh$ are unique up to the transformation
$\Lh\!\!\to \!a \Lh$ and
$\Lbh\!\to  a^{\!-1\!}\Lbh$.
The null second fundamental form and the conjugate null second fundamental form
are defined by
$
\chi(X,Y)\!\!=\!g(\nabla_{\!\!X}\Lh,Y)$ respectively
$\underline{\chi}(X,Y)\!\!=\!g(\nabla_{\!\!X} \Lbh,Y),
$
for any vectors $X,Y$ tangent to $S$ at a point,
where $\nabla_{\!\!X}$ is covariant differentiation.
The  Hawking mass
$$
M_{\mathcal{H}}(S)\!=r(S)\big(1\!+\int_S \tr \chi \,\tr \underline{\chi} \,dS\!/16\pi
\big),
$$
is
invariant under the transformation since $\chi\!\!\to\! a\chi$ and
$\underline{\chi}\!\!\to\! \underline{\chi}/a$. If $\tr \chi \,\tr \underline{\chi}\!<\!0$ we
can fix $\Lh$ and $\Lbh$ by
$\tr\chi\!+\tr \underline{\chi}\!=\!0$. Let $\hat{\chi}$ and $\hat{\underline{\chi}}$
be the traceless parts. The incoming and outgoing energy fluxes are
$$
E(S)\!=\!\int_{S}\! \hat{\chi}^2 dS\!/16\pi,\qquad\text{and}\qquad
\underline{E}(S)\!=\!\int_{S}\! \hat{\underline{\chi}}^2 \!dS_{\!}/16\pi\!.
$$
We use the family of spheres $S_{u^*\!, r}=\!\{(t,x); t\!=\!u^*+r^*(r), \, |x|\!\!=\!r\}$ to define the {\it asymptotic Hawking mass} and the radiated energy at null infinity as follows
\begin{align*}
	M_{AH}(u^*)=\lim_{r\to \infty}M_{\mathcal{H}}(S_{u^*\!, r})\quad \text{and}\quad E_{AH}(u^*)=\lim_{r\to\infty}\underline{E}(	S_{u^*\!,r}),
\end{align*}
with $r(S)^2\slashed{g}$ converging to a round metric
where $\slashed{g}$ is the restriction of $g$ on the spheres $S_{u^*\!, r}$.
\begin{remark}
	As pointed out in \cite{Sauter08}, it is absolutely essential in the limit process that the spheres $S_{u^*\!, r} $converge to a round metric sphere. Otherwise the limit of the Hawking mass has nothing to do with the Bondi mass, in general. This is somehow related to the undesirable fact that the Hawking mass of any spherical surface in Euclidean space is negative unless it is a metric sphere where it is zero.
\end{remark}
With the asymptotics results in \cite{Lind17} we prove in subsection \ref{subsec:5.4}:
\begin{theorem}\label{thm:2}
	The asymptotic Hawking mass $M_{AH}(u^*)$ and the radiated energy $E_{AH}(u^*)$ are well defined and in fact with $n$ in \eqref{eq:Newsfucntion}
	\[
		M_{AH}(u^*)=M-\frac{1}{8\pi}\int_{-u^*}^{\infty}\int_{\mathbb{S}^2}n(\eta, \omega)\,dS(\omega)d\eta,\quad\text{and}\quad	E_{AH}(u^*)=\!\frac{1}{8\pi}\int_{\mathbb{S}^2}\!n(-u^*,\omega){ dS(\omega)\!}.
	\]
	 Therefore, they satisfy the mass loss law
	\[
	\frac{d}{du^*}M_{AH}(u^*)=-E_{AH}(u^*).
	\]
Moreover, $M_{AH}(u^*)\to M$, the ADM mass, as $u^*\to-\infty$ and $M_{AH	}(u^*)=\to0$ as $u^*\to\infty$.
	\end{theorem}

\subsection{The Bondi-Sachs coordinates}
The definition of Bondi mass introduced in 1962 in \cite{BMM62, S62} requires the existence of the so called Bondi-Sachs coordinates. In this manuscript we will construct the Bondi-Sachs coordinates $\overline{y}^p=(u, \overline{r}, \overline{y}^3, \overline{y}^4)$ under which we denote the solution to \eqref{eq:EisnteinWave} by $\overline{g}$. The Bondi-Sachs coordinates  $\overline{y}^p =(u,\overline{r},\overline{y}^3, \overline{y}^4)$ are based on a family of outgoing
null hypersurfaces $\overline{y}^1=u=const$. The two angular coordinates $\overline{y}^a$, $(a,b,c,...=3,4)$, are constant along the null rays, i.e. $g^{\alpha\beta}\pa_\beta u\pa_\alpha \overline{y}^a\!\!=\!0$.  The coordinate $\overline{y}^2 =\overline{r}$, which varies along the null rays,
is chosen to be an areal coordinate such that
$\det [\overline{g}_{ab}] = \overline{r}^4 \mathfrak{q}$, where $\mathfrak{q}(\overline{y}^a)$ is the determinant of the unit sphere metric $\overline{q}_{ab}$
associated with the angular coordinates $\overline{y}^a$. In these coordinates, the metric takes the Bondi-Sachs form (see Proposition \ref{prop:BSmetric})
\[
\overline{g}_{pq}d\overline{y}^p d\overline{y}^q=-\frac{V}{\overline{r}}e^{2\beta} du^2-2 e^{2\beta}dud\overline{r} +\overline{r}^2h_{ab}\Big(d\overline{y}^a-U^adu\Big)\Big(d\overline{y}^b-U^bdu\Big).
\]

\subsection{The Bondi mass and radiated energy}
Once we write the metric in the Bondi-Sachs form as above, following \cite{MW16} we define the mass aspect $M_A$ and news tensor $N_{ab}$ as follows
\begin{align*}
	M_A(u, \overline{y}^a)&:=-\lim_{\overline{r}\to\infty}\Bigl(V(u, \overline{r}, \overline{y}^a)-\overline{r}\Bigr),\\
	N_{ab}(u, \overline{y}^c)&:=\frac 12\pa_uC_{ab}(u, \overline{y}^c)\quad\text{where}\quad C_{ab}(u, \overline{y}^c):=\lim_{\overline{r}\to\infty}\overline{r}(h_{ab}(u, \overline{r}, \overline{y}^c)-\overline{q}_{ab}(\overline{y}^c)).
\end{align*}
The Bondi mass $M_B$\footnote{The Bondi energy-momentum vector for the outgoing null hypersurfaces $u=const$ is defined as \cite{BMM62, S62, CJM98, BB12} the average of the Bondi mass aspect $M_A$ over the unit round sphere weighted by a vector $N^\alpha=(1, N^i )$ where $N^i (1\leq i\leq3)$ are the $l=1$ spherical harmonics. That is, $N^i=(\sin\theta\cos\varphi, \sin\theta\sin\varphi, \cos\theta)$ in the natural spherical coordinates $(\theta, \varphi)$ for a unit round sphere. More specifically, the Bondi energy-momentum vector is defined by
\[
M_{B}^\alpha(u)=\frac{1}{4\pi}\int_{\mathbb{S}^2}M_A(u, \overline{y}^a)N^\alpha d\overline{S}(\overline{y}^a).
\] The time component $M_B^0$ is referred as the Bondi energy in \cite{B04, Sauter08, BB12} and the Bondi mass in \cite{BMM62, S62, CJM98} respectively. In this manuscript we adopt the latter definition.} radiated energy $E_B$ are defined by
\[
M_B(u)=\frac{1}{4\pi}\int_{\BS^2}M_A(u,\overline{y}^a)d\overline{S}(\overline{y}^a)\quad\text{and}\quad
E_B(u)=\frac{1}{4\pi}\int_{\BS^2}\!|N|^2\,d\overline{S}(\overline{y}^a).
\]
where $d\overline{S}(\overline{y}^a)\!=\!\!\sqrt{\mathfrak{q}(\overline{y}^a)}d\overline{y}^3d\overline{y}^4$ is the volume form associated to the unit sphere metric $\overline{q}_{ab}$ and $\lvert N\rvert^2\!\!=\!\overline{q}^{ac}\overline{q}^{bd}N_{ab}N_{cd}$.\!
We will prove the existence of $M_B(u)$ and $E_B(u)$ and the Bondi mass loss law in subsection \ref{subsec:7.2}
\begin{theorem}\label{thm:3}
	Let $M_A, N_{ab}, M_B, E$ be defined as above, then we have
	\[ M_B(u)=M-\frac{1}{8\pi}\int_{\BS^2}\!\int_{-u}^\infty\!n(\eta,\overline{y}^a)\,d\eta d\overline{S}(\overline{y}^a).
	\]
	The radiated energy is expressed as
	\[	E_B(u)=\frac{1}{8\pi}\int_{\BS^2}\!n(-u,\overline{y}^a)\,d\overline{S}(\overline{y}^a).
	\]
	They satisfy the Bondi mass loss law
	\[
		\frac{d}{du}M_B(u)=-E_B(u).
	\]
	Moreover, $M_B(u)\!\to \!M$ as $u\!\to\!\!-\infty$ where $M$ is the ADM mass and $M_B(u)\!\to\!0$ as $u\to\infty$.
\end{theorem}
\begin{remark}
According to Theorem \ref{thm:1}, \ref{thm:2} and \ref{thm:3}, we see that the masses and radiated energies at null infinity defined in the above three ways are equivalent. We also note that these three different ways rely on the same collections of the asymptotics in \cite{Lind17}. 
In \cite{Lind17} it was shown that $\int_{-\infty}^\infty\int_{\BS^2}n(\eta, \omega)\,dS(\omega)d\eta/8\pi=M$, and we can conclude that the total radiated energy is equal to the ADM mass. In particular, this implies that if $n=0$ then $M=0$, and then by the positive mass theorem \cite{SY79, W81} the spacetime is Minkowski space.
\end{remark}

\subsection*{Acknowledgments} We would like to thank Igor Rodnianski for many important discussions and
initial collaboration.
 We would also like to thank Mihalis Dafermos and Volker Schlue
  for useful discussions. H. L. was supported in part by Simons Foundation Collaboration Grant 638955.

\section{The metric in modified asymptotically Schwarzschild null coordinates}
In this section we introduce the modified asymptotically Schwarzschild null coordinates and review some results concerning the asymptotics of the metric established in \cite{Lind17}.
\subsection{Modified asymptotically Schwarzschild null coordinates}\label{subsec:2.2}
Suppose $g_{\alpha\beta}=m_{\alpha\beta}+h^0_{\alpha\beta}+h^1_{\alpha\beta}$ where $h^0_{\alpha\beta}=\frac{M}{r}\delta_{\alpha\beta}\chi(\frac{r}{1+t})$ and $\chi(s)=1$ when $s\geq 1/2$ and $0$ when $s\leq1/4$. Then the inverse metric $g^{\alpha\beta}=m^{\alpha\beta}+h_0^{\alpha\beta}+h_1^{\alpha\beta}$ where $h_0^{\alpha\beta}=-\frac M r\delta^{\alpha\beta}\chi(\frac{r}{1+t})$ and $h^{\alpha\beta}_1=-m^{\alpha\mu}h^1_{\mu\nu}m^{\nu\beta}+O(h^2)$.

We introduce the modified asymptotically Schwarzschild null coordinates $\widehat{y}^p=(v^*=t+r^*, u^*=t-r^*, \widehat{y}^3, \widehat{y}^4)$. Here we let $r=|x|$, $\omega=\frac xr\in\BS^2$ and $\widehat{y}^a=(\widehat{y}^3, \widehat{y}^4)$ be local coordinates on $\BS^2$ and define $r^*=r+M\ln r+O(M/r)$ which is slightly different from $\widetilde{r}=r+M\ln r$ by solving
\[
\frac{dr^*}{dr}=\rho'(r)=\Big(\frac{1+M/r}{1-M/r}\Big)^{1/2}=1+\frac{M}{r}+O(M^2/r^2).
\]
In what follows indices $\widehat{y}^p\!,\widehat{y}^q\!,\dots$ will stand for all the modified asymptotically Schwarzschild null coordinates whereas $\widehat{y}^a\!,\widehat{y}^b\!\!,\dots$ stand for the coordinates on the sphere only. We will now calculate the changes of variables
$\pa_{\widehat{y}^p}\!=\widehat{A}_p^\mu\pa_\mu$ and $\pa_\mu=\widehat{A}_\mu^p\pa_{\widehat{y}^p}$.
We define $\Ls\!\!=\!\Ls^{\mu}\pa_\mu=\pa_t+\pa_{r^*}$, $\Lbs\!\!=\!\Lbs^{\mu}\pa_\mu=\pa_t\!-\!\pa_{r^*}$, $\Ls_{\!\!\mu}=-\pa_\mu u^*$ and $\Lbs_{\!\!\mu}\!=-\pa_\mu v^*$, then we have
\[
\pa_{v^*}=\frac{1}{2}\Ls=\frac12(\pa_t+(\omega^i/\rho')\pa_i),\quad \pa_{u^*}=\frac12\Lbs=\frac12(\pa_t-(\omega^i/\rho')\pa_i),\quad \pa_{\widehat{y}^a}=\widehat{A}_a^\mu\pa_\mu,
\]
where $\omega=x/r$, $x=r\omega$,
and
\[
\pa_\mu=-\frac12\Ls_{\!\!\mu}\pa_{\Lbs}-\frac12\Lbs_{\!\!\mu}\pa_{\Ls}+\widehat{X}_\mu^a\pa_{\widehat{x}^a}=-\Ls_{\!\!\mu}\pa_{u^*}-\Lbs_{\!\!\mu}\pa_{v^*}+\widehat{A}_\mu^a\pa_{\widehat{y}^a}.
\]
Here we have $\widehat{A}_a^\mu \widehat{A}_\mu^b=\delta_a^b$ and $\widehat{A}_a^\mu=O(r), \widehat{A}_\mu^a=O(1/r)$.

Under the coordinates $(v^*, u^*, \widehat{y}^3, \widehat{y}^4)$, we denote the corresponding metric by $\widehat{g}$. The inverse of the metric then satisfies
$\widehat{g}^{pq}=g^{\alpha\beta}\widehat{A}^p_\alpha \widehat{A}^q_\beta$. We have in the region $r>t/2$
\begin{gather*}
	\widehat{g}^{v^*\! v^*}\!\!\! =\widehat{h}_1^{v^*\! v^*}\!\!\!, \quad\widehat{g}^{v^*\! u^*}\!\!\! =-2(1\!+\!\frac{M\!}{r})+\widehat{h}_1^{v^*\! u^*}\!\!\!,\quad \widehat{g}^{v^*\! a}\!=\widehat{h}_1^{v^*\! a}\!\!,\quad \widehat{g}^{u^*\! u^*}\!\!\! =\widehat{h}_1^{u^*\! u^*}\!\!\!,\quad
	\widehat{g}^{u^*\! a}\!=\widehat{h}_1^{u^*\! a}\!\!,\quad
	\widehat{g}^{ab}\!=(1\!-\!\frac {M\!}{r})\frac{1}{r^2}\widehat{q}^{\,ab}\!+\widehat{h}_1^{ab}\!.
\end{gather*}
where $\widehat{q}_{ab}$ is the unit sphere metric on $\BS^2$ associated with the angular coordinates $(\widehat{y}^3\! , \widehat{y}^4)$ and $\widehat{q}^{\,ac}\widehat{q}_{cb}\! =\! \delta^a_b$. In fact this follows from decomposing the leading part $m\! +\! h_0$ into a time part, a radial part and a sphere part
\begin{equation}
	\big(m^{\alpha\beta}+h_0^{\alpha\beta}\big)\xi_\alpha\eta_\beta
	=-(1+\frac Mr)\xi_0 \eta_0+(1-\frac Mr)\omega^i\omega^j\xi_i\eta_j
	+(1-\frac Mr)\big(\delta^{ij}-\omega^i\omega^j\big)\xi_i\eta_j ,
\end{equation}
and composing with the change of variables in the radial part.

\subsection{Asymptotics result}

Let $(\widetilde{t},\widetilde{x})$ be the asymptotically Schwarzschild coordinates with
$$
\widetilde{t}=t,\qquad \widetilde{x}^i=\widetilde{r} \omega^i,\qquad
\text{where}\quad \omega^i= {x^i}/{r},\qquad \widetilde{r}=r+M\ln{r},\quad
r=|\,x|.
$$ We write $\widetilde{h}^0_{\alpha\beta}+\widetilde{h}^1_{\alpha\beta}=h^0_{\alpha\beta}+h^1_{\alpha\beta}$ where $\widetilde{h}^0_{\alpha\beta}=\chi(\frac{\widetilde{r}}{1+t})\frac{M}{\widetilde{r}}\delta_{\alpha\beta}$ and $\widetilde{q}=\widetilde{r}-t$. We now restate several propositions from \cite{Lind17}, which will be used frequently in this manuscript. The first proposition is the sharp decay estimates in asymtotically Schwarzschild coordinates
\begin{prop}[{\cite[Proposition 17]{Lind17}}]\label{prop:sharpmetricdecay}
	For $|I|\!\leq\! N\!-6$ with
	$\gamma^\prime\!=\gamma-C\varepsilon$, $\widetilde{q}\!=\!\widetilde{r}\!-\!t$ and
$\langle \widetilde{q}\rangle\!=\!\sqrt{1\!+\widetilde{q}^2}$ we have
	\beqs\label{eq:metricdecaysharp}
	|{\widetilde{Z}}{}^I \widetilde{h}^{1}|\les \frac{\varepsilon^2 S^0(t,\widetilde{r})}{(1\!+\!t\!+\!\widetilde{r})(1\!+\!\widetilde{q}_+)^{1-C\varepsilon}}
	+\frac{\varepsilon}{1\!+\!t\!+\!\widetilde{r}}   \frac{1}{(1\!+\!|\,\widetilde{q}|)^{\gamma^\prime}},\quad\text{where}\quad
S^0(t,\widetilde{r})
	=\frac{t}{\widetilde{r}}\,\ln{\!\Big(\frac{\langle \,t\!+\widetilde{r}\,\rangle}{\langle \,t\!-\widetilde{r}\rangle}\Big)}
	\les
	\frac{1}{\varepsilon}
	\Big(\frac{\langle\, t\!+\widetilde{r} \rangle}{\langle \,t\!-\widetilde{r}\rangle}\Big)^{\varepsilon}\!\!.
	\eqs
	For $\widetilde{r}\geq \!t_{\!}/2 \!$ we have
	\begin{align}
		\label{eq:metricdecaysharptan}
		|\widetilde{Z}^{I} \widetilde{h}^{1}_{TU}|
		&\les \frac{\varepsilon}{(1+t+\widetilde{r})(1+\widetilde{q}_+)^{\gamma^\prime}},\\
		|\pa \widetilde{Z}^{I} \widetilde{h}^{1}_{LT}|\!+|\pa \widetilde{Z}^{I\!} \delta^{AB}\widetilde{h}^{1}_{AB}|
		&\les \frac{\varepsilon}{(1+t+\widetilde{r})^{2-\varepsilon}
			(1+|\,\widetilde{q}|)^{\varepsilon}(1+\widetilde{q}_+)^{\gamma^\prime}},\label{eq:sharpwsvecoordder}\\
		|\widetilde{Z}^{I} \widetilde{h}^{1}_{LT}|\!+|\widetilde{Z}^{I} \delta^{AB} \widetilde{h}^{1}_{AB}|
		&\les \frac{\varepsilon}{(1_{\!}\!+t\!+\widetilde{r})^{1+\gamma^\prime}\!\!}
		+\!\frac{\varepsilon}{1_{\!}\!+t\!+\widetilde{r}\!\!}\,
		\Big(\frac{1\!+\widetilde{q}_-\!}{1_{\!}\!+t\!+\widetilde{r}\!}\Big)^{\!\!1-\varepsilon}\!\!\!,\!\!\!
		\label{eq:sharpwsvecoordfunc}
	\end{align}
where $\widetilde{q}_+\! =\! \max\{0, \widetilde{q}\}$, $\widetilde{q}_-\! =\! \max\{0, -\widetilde{q}\}$. Here $\widetilde{h}^1_{UV}\! =\! \widetilde{h}^1_{\alpha\beta}U_{\!\alpha} V_\beta$ where $U_{\!\alpha} \! =\! m_{\alpha\beta}U^\beta\!\!$ and $U\!, \!V\!\!  \in\!\! \{L, \Lb\,, A, B\}$, the null frame associated to the coordinates $(t, x)$, and $\widetilde{Z}^{I}\!$ stands for a product of $|I|\!$ of the vector fields
\begin{equation}\label{eq:vectorfields}
	\{\pa_{\widetilde{x}^{\alpha}}, ~\widetilde{x}^{i}\pa_{\widetilde{x}^{j}}-\widetilde{x}^{i}\pa_{\widetilde{x}^{j}}, ~\widetilde{x}^{i}\pa_{\widetilde{t}}+\widetilde{t}\pa_{\widetilde{x}^{i}}, ~\widetilde{S}
=\widetilde{t}\pa_{\widetilde{t}}+\widetilde{x}^{i}\pa_{\widetilde{x}^{i}}\},
\quad \text{where}\quad \widetilde{x}^{\alpha}=(\widetilde{t}=t, \widetilde{x}^{i})=(t, \widetilde{r}\omega^i).
\end{equation}
\end{prop}
Here and in what follows we focus on the region $\{\widetilde{r}>t/2\}\cap\{r^*>t/2\}$.
\begin{remark}\label{rem:sharpmetricdecay}
In view of Proposition \ref{prop:sharpmetricdecay} we have the sharp decay estimates for the metric components in modified asymptotically Schwarzschild null coordinates. For $r^*>t/2$, with $q^*=-u^*, 0<\gamma<1$ and $\gamma'=\gamma-C\varepsilon$ we see that for $|I|\leq N-6$
\begin{align}
	|Z^{*I}\hat{h}^{v^*v^*}|&\les \frac{\varepsilon^2}
	{(1+t+r^*)(1+q_+^*)^{1-C\varepsilon}}\ln{\!\Big(\frac{\langle \,t\!+r^*\,\rangle}{\langle \,t\!-r^*\rangle}\Bigr)},\\
	|Z^{*I}\hat{h}_1^{v^*u^*}|+|Z^{*I}(r\hat{h}_1^{v^*a})|&\les \frac{\varepsilon}{(1+t+r^*)(1+q_+^*)^{\gamma^\prime}},\\
	|Z^{*I}\hat{h}_1^{u^*u^*}|+|Z^{*I}(r\hat{h}_1^{u^*a})|+|Z^{*I}(r^2\hat{q}_{ab}\hat{h}_1^{ab})|&\les \frac{\varepsilon}{(1_{\!}\!+t\!+r^*)^{1+\gamma^\prime}\!\!}
	+\!\frac{\varepsilon}{1_{\!}\!+t\!+r^*\!\!}\,
	\Big(\frac{1\!+q_-^*\!}{1_{\!}\!+t\!+r^*\!}\Big)^{\!\!1-\varepsilon}.
\end{align}
Here $Z^{*I}$ stands for a product of $|I|$ of the vector fields
\begin{equation}\label{eq:modifiedvectorfields}
	\{\pa_{x^{*\alpha}}, ~x^{*i}\pa_{x^{*j}}-x^{*i}\pa_{x^{*j}}, ~x^{*i}\pa_{t^*}+t^*\pa_{x^{*i}}, ~t^*\pa_{t^*}+x^{*i}\pa_{x^{*i}}\},
\quad\text{where} \quad x^{*\alpha}=(t^*, x^{*i})=(t, r^*\omega^i).
\end{equation}
\end{remark}
 The second one establishes the estimates for $L$ derivatives
 \begin{prop}[{\cite[Proposition 18]{Lind17}}]\label{prop:specialsharpmetricdecay}
 	With  $\delta_{UV}^{\Lb\Lb}\!\!=\!1$ if $U\!=\!V\!\!=\!\Lb$ and $0$ otherwise and $|I|\leq N-6$ we have
 	\begin{equation}\label{eq:extraLderest1}
 		\frac{1}{r\!} \,|(\partial_{t}+\pa_{\widetilde{r}})(\widetilde{r} \widetilde{Z}{}^I \widetilde{h}^{1}_{UV})|\les
 		\frac{\varepsilon(1+\widetilde{q}_-)^{\gamma-C\varepsilon}}{(1\!+t\!+\widetilde{r})^{2+\gamma-C\varepsilon}}	+ \delta_{UV}^{\Lb\Lb}\,\frac{\varepsilon(1+\widetilde{q}_+)^{-\gamma}}{(1\!+t\!+\widetilde{r})^2}.
 	\end{equation}
 \end{prop}

\begin{remark}\label{rem:specialsharpmetricdecay}
	Correspondingly in modified asymptotically Schwarzschild null coordinates we have
	\begin{equation}
		\frac{1}{r\!} \,|\partial_{\Ls}(r^* {Z^*}{}^I \widehat{h}_1^{v^*u^*})|\les
		\frac{\varepsilon(1+q^*_-)^{\gamma-C\varepsilon}}{(1\!+t\!+r^*)^{2+\gamma-C\varepsilon}}\quad\text{for}\quad |I|\leq N-6.
		\label{eq:extraLderest2}
		\end{equation}
\end{remark}

The third one provides us with the asymtotics for the metric in asymptotically Scwarzschild coordinates
\begin{prop}[{\cite[Proposition 20, 22]{Lind17}}]	\label{prop:lim}
Let $H^1_{TU}(\widetilde{q},\omega,\widetilde{r})=\widetilde{r} h^1_{TU}(\widetilde{r}-\widetilde{q}, \widetilde{r}\omega)$,then the limit
\beqs
H^{1\infty}_{TU}(\widetilde{q},\omega)
=\lim_{\widetilde{r}\to\infty}H^{1}_{TU}(\widetilde{q},\omega,\widetilde{r}),
\eqs
exists and satisfies $H_{TU}^{1\infty}\!=\!H_{UT}^{1\infty}$,
and
$
H^{1\infty}_{LT}(\widetilde{q},\omega)\!=\!\delta^{AB}H^{1\infty}_{AB}
\textbf{}(\widetilde{q},\omega)\!=\!0.
$
Moreover, for  $|\alpha|+k\leq N-6$ and $|J|+|K|=k$ and $\widetilde{r}>t/2$
\begin{align*}
	\big|\,\pa_\omega^\alpha \big((1+|\,\widetilde{q}|)\pa_{\widetilde{q}}\big)^k
	H^{1\infty}_{TU}(\widetilde{q},\omega)\big| &\les \varepsilon
	(1+\widetilde{q}_+)^{-\gamma^\prime},\\
	\big|\pa_\omega^\alpha \widetilde{S}^J\pa_t^KH^{1}_{TU}(\widetilde{q},\omega,\widetilde{r})
	-\pa_\omega^\alpha(\widetilde{q}\,\pa_{\widetilde{q}})^J(-\pa_{\widetilde{q}})^KH^{1\infty}_{TU}(\widetilde{q},\omega)\big|
	&\les\varepsilon \Big(\frac{1+\widetilde{q}_-}{1+t+\widetilde{r}}\Big)^{\gamma^\prime}.
\end{align*}
Let
\begin{equation}\label{eq:NewsfunctionA}
 n(\widetilde{q},\omega)=\tfrac{1}{2}\delta^{CD}\delta^{C^\prime
	D^\prime} V^\infty_{CC^\prime}(\widetilde{q},\omega)
V^\infty_{DD^\prime}(\widetilde{q},\omega)\quad\text{where}\quad V^\infty_{TU}(\widetilde{q}, \omega)=\pa_{\widetilde{q}}H^{1\infty}_{TU}(\widetilde{q}, \omega),
\end{equation}
 for the component $h^1_{\Lb\Lb}(t, \widetilde{r}\omega)$ we have when $\widetilde{r}\gg1$
\[
h^1_{\Lb\Lb}(t, \widetilde{r}\omega)=2\frac{M}{\widetilde{r}}(\chi^e(\widetilde{q})-1)+\int_{\widetilde{r}-t}^\infty
\frac{2}{\widetilde{r}}\ln{\Big(\frac{t+\widetilde{r}+\eta}{t-\widetilde{r}+\eta}\Big)}
n\big(\eta,\omega\big)\, d \eta+\frac{H^1_{\Lb\Lb}(\widetilde{q} ,\omega)}{\widetilde{r}}+\widetilde{\mathcal{R}}.
\]
Here $\chi^e(s)=1$ when $s\geq2$ and $\chi^e(s)=0$ when $s\leq1$, and  for $|\alpha|+k=|I|\leq N-7$ we see that
 \[	\big|\,\pa_\omega^\alpha \big((1+|\,\widetilde{q}|)\pa_{\widetilde{q}}\big)^k
H^{1\infty}_{\Lb\Lb}(\widetilde{q},\omega)\big| \les \varepsilon
(1+\widetilde{q}_+)^{-\gamma^\prime},\qquad|\widetilde{Z}^I\widetilde{\mathcal{R}}|\lesssim \varepsilon \frac{(1+\widetilde{q}_-)^{\gamma^\prime}}{(1+t+\widetilde{r})^{1+\gamma^\prime}}.
\]
\end{prop}

\begin{remark}\label{rem:limit}
In modified asymptotically Schwarzschild null coordinates, since $|q^*-\widetilde{q}|\les M/r$ it follows from Proposition \ref{prop:lim} that the following limits
\begin{gather*}
\widehat{H}_{1\infty}^{v^*u^*}(q^*, \widehat{y}^a)\!=\!\lim_{r^*\to\infty}\!\! r^*\widehat{h}_1^{v^*u^*}(v^*, -q^*, \widehat{y}^a),\quad \widehat{H}_{1\infty}^{u^*u^*}(q^*, \widehat{y}^a)\!=\!\lim_{r^*\to\infty}\!\! r^*\widehat{h}_1^{u^*u^*}(v^*, -q^*,\widehat{y}^a),\\
\widehat{H}_{1\infty}^{u^*a}(q^*, \widehat{y}^a)\!=\!\lim_{r^*\to\infty} {r^*}^2\widehat{h}_1^{u^*a}(v^*, -q^*,\widehat{y}^a),\quad
\widehat{H}_{1\infty}^{v^*a}(q^*, \widehat{y}^a)\!=\!\lim_{r^*\to\infty} {r^*}^2\widehat{h}_1^{v^*a}(v^*, -q^*,\widehat{y}^a),\\
\widehat{H}_{1\infty}^{ab}(q^*, \widehat{y}^a)\!=\!\lim_{r^*\to\infty} {r^*}^3\widehat{h}_1^{ab}(v^*, -q^*,\widehat{y}^a)
\end{gather*}
exist and satisfy $\widehat{H}^{u^*u^*}_{1\infty}(q^*, \widehat{y}^a)=\widehat{H}^{u^*a}_{1\infty}(q^*, \widehat{y}^a)\!=\widehat{q}_{ab}\widehat{H}^{ab}_{1\infty}(q^*,\widehat{y}^a)=0$. Moreover for  $|\alpha|+k\leq N-6$  we have
\[
	\big|\,\pa_{\widehat{y}^a}^\alpha \big((1+|\,q^*)\pa_{q^*}\big)^k
\widehat{H}^{pq}_{1\infty}(q^*,\widehat{y}^a)\big| \les \varepsilon
(1+q^*_+)^{-\gamma^\prime},\qquad (p,q)\neq v^*,v^*),
\]
and when $r^*\gg1$
\begin{gather*}
	\widehat{h}_1^{v^*u^*}(v^*, -q^*,\widehat{y}^a)\!=\!\frac{\widehat{H}_{1\infty}^{v^*u^*}(q^*, \widehat{y}^a)}{r^*}\!+\!\widehat{\mathcal{R}}^{u^*v^*},\qquad	\widehat{h}_1^{v^*u^*}(v^*, -q^*,\widehat{y}^a)\!=\!\frac{\widehat{H}_{1\infty}^{u^*u^*}(q^*, \widehat{y}^a)}{r^*}\!+\!\widehat{\mathcal{R}}^{u^*u^*}\\
\widehat{h}_1^{v^*u^*}(v^*, -q^*,\widehat{y}^a)\!=\!\frac{\widehat{H}_{1\infty}^{u^*a}(q^*, \widehat{y}^a)}{{r^*}^2}\!+\!\widehat{\mathcal{R}}^{u^*a},\qquad
	\widehat{h}_1^{v^*a}(v^*, -q^*,\widehat{y}^a)\!=\!\frac{\widehat{H}_{1\infty}^{v^*a}(q^*, \widehat{y}^a)}{{r^*}^2}\!+\!\widehat{\mathcal{R}}^{v^*a},\\
\widehat{h}_1^{ab}(v^*, -q^*,\widehat{y}^a)\!=\!\frac{\widehat{H}_{1\infty}^{ab}(q^*, \widehat{y}^a)}{{r^*}^3}\!+\!\widehat{\mathcal{R}}^{ab},
\end{gather*}
where the remainders $\widehat{\mathcal{R}}$ satisfy
\begin{gather*}
|Z^{*I}\widehat{\mathcal{R}}^{v^*\!u^*}\!|
+|Z^{*I}\widehat{\mathcal{R}}^{u^*\!u^*}\!|
\les\frac{\eps(1\!+\!q^*_-)^{\gamma^\prime}}{(1\!+\!t\!+\!r^*)^{1+\gamma^\prime}},\qquad
|Z^{*I}\widehat{\mathcal{R}}^{ab}|
\les\frac{\eps(1\!+\!q^*_-)^{\gamma^\prime}}{(1\!+\!t\!+\!r^*)^{3+\gamma^\prime}},\\
|Z^{*I}\widehat{\mathcal{R}}^{v^*\!a}|+|Z^{*I}\widehat{\mathcal{R}}^{u^*\!a}|
\les\frac{\eps(1\!+\!q^*_-)^{\gamma^\prime}}{(1\!+\!t\!+\!r^*)^{2+\gamma^\prime}}.
\end{gather*}
Let $n(q^*,\widehat{y}^a)=n(q^*, \omega(\widehat{y}^a))$, then we have
\begin{equation}\label{eq:NewsfunctionM}
n(q^*,\widehat{y}^a)=\tfrac{1}{2}\widehat{q}_{ab}\widehat{q}_{a'b'}\widehat{V}^{aa'}(q^*,\widehat{y}^a)\widehat{V}^{bb'}(q^*,\widehat{y}^a)\quad\text{with}\quad\widehat{V}^{ab}=\pa_{q^*}\widehat{H}^{ab}(q^*,\widehat{y}^a).
\end{equation}
As for the component $\widehat{h}^{v^*v^*}_1$, we have when $r^*\gg1$
\[
\widehat{h}_1^{v^*v^*}(v^*, u^*, \widehat{y}^a)=-\frac{2M}{r^*}(\chi^e(q^*)-1)-\int_{q^*}^\infty
\frac{2}{r^*}\ln{\Big(\frac{v^*+\eta}{u^*+\eta}\Big)}
n\big(q^*,\omega\big)\, d \eta+\frac{\widehat{H}^{v^*v^*}_{1\infty}(q^* ,\omega)}{r^*}+\mathcal{R}^{v^*v^*}.
\]
Here for $|\alpha|+k=|I|\leq N-7$, $\widehat{H}^{v^*v^*}_{1\infty}$ and the remainder $\mathcal{R}^{v^*v^*}$ satisfy \[	\big|\,\pa_\omega^\alpha \big((1+|\,q^*)\pa_{q^*}\big)^k
\widehat{H}^{v^*v^*}_{1\infty}(q^*,\omega)\big| \les \varepsilon
(1+q^*_+)^{-\gamma^\prime}, \qquad|Z^{*I}\mathcal{R}^{v^*v^*}|\lesssim \varepsilon \frac{(1+q^*_-)^{\gamma^\prime}}{(1+t+r^*)^{1+\gamma^\prime}}.
\]
\end{remark}
The last proposition gives a relation between $M$ and $n$.
\begin{prop}[{\cite[Proposition 28]{Lind17}}]\label{prop:totalmassloss} We have
	\beq\label{eq:totalmassloss}
	\frac{1}{2}\int_{-\infty}^{+\infty} \int_{\mathbb{S}^2}n(\widetilde{q},\omega)\frac{ dS(\omega)\!\!}{4\pi}\,
	d\widetilde{q}\!=M.
	\eq
\end{prop}
In what follows we write $A=O(B)$ if $A\leq CB$  and $A=O_k(B)$ if $\sum_{|I|\leq k}|\widetilde{Z}^{I}A|+|Z^{*I}A|\leq CB$  with $\widetilde{Z}\!$ and $Z^{*}\!$ defined in \eqref{eq:vectorfields} and \eqref{eq:modifiedvectorfields} respectively for some universal constant $C$. We define $\sigma\!=\!\min\{\gamma'\!, 1\!-3\eps\}\!>\!0$.

\subsection{Wave coordinate condition in modified asymptotically Schwarzschild null coordinates}

Let $N$ be some fixed large integer ($N=9$ works). We express the wave coordinate condition in modified asymptotically Schwarzschild null coordinates.
\begin{equation}\label{eq:wcc}
\pa_\alpha\Big(g^{\alpha\beta}\sqrt{|g|}\Big)=\Big(-\frac12\Ls_{\!\!\alpha} \pa_{\Lbs}-\frac12\Lbs_{\!\!\alpha}\pa_{\Ls}+\widehat{A}_\alpha^a\pa_{\widehat{y}^a}\Big)\Big(g^{\alpha\beta}\sqrt{|g|}\Big)=0.
\end{equation}

\subsubsection{Contraction with $\Ls$}
\begin{proposition}
	We have
	\begin{equation}\label{eq:wcc_L}
		\frac12\pa_{\Lbs}(\widehat{h}_1^{u^*u^*})+ \frac{\widehat{h}^{v^*u^*}_1}{r}=O_2(\frac{\varepsilon}{r^{2+\sigma}}).
	\end{equation}
\end{proposition}
\begin{proof}
	Contracting \eqref{eq:wcc} with $\Ls_{\!\!\beta}$ we obtain
\begin{equation*} \Ls_{\!\!\beta}\pa_\alpha\big(g^{\alpha\beta}\sqrt{|g|}\big)=-\frac12\Ls_{\!\!\beta}\Ls_{\!\!\alpha} \pa_{\Lbs}\big(g^{\alpha\beta}\sqrt{|g|}\big)-\frac12\Ls_{\!\!\beta}\Lbs_{\!\!\alpha}\pa_{\Ls}\big(g^{\alpha\beta}\sqrt{|g|}\big)
+\Ls_{\!\!\beta} \widehat{A}_\alpha^c\pa_{\widehat{y}^c}\big(g^{\alpha\beta}\sqrt{|g|}\big)=0.
\end{equation*}
Here and in what follows the repeated indices $c,d$ are summed over $a,b$. We first analyze the first term
\begin{multline*}
	-\frac12\Ls_{\!\!\beta}\Ls_{\!\!\alpha} \pa_{\Lbs}\left(g^{\alpha\beta}\sqrt{|g|}\right)=-\frac12\pa_{\Lbs}\left(\Ls_{\!\!\beta}\Ls_{\!\!\alpha} g^{\alpha\beta}\sqrt{|g|}\right)+\pa_{\Lbs}(\Ls_{\!\!\alpha})\Ls_{\!\!\beta}g^{\alpha\beta}\sqrt{|g|}\\
	=-\frac12\pa_{\Lbs}(\widehat{h}_1^{u^*u^*})+\omega^i\delta_{i\alpha}\frac {M}{r^2}\frac{1}{\rho'}\Ls_{\!\!\beta}g^{\alpha\beta}\sqrt{|g|}+O_2(\frac{\varepsilon}{r^{2+\sigma}}) =-\frac12\pa_{\Lbs}(\widehat{h}_1^{u^*u^*})+\frac{M}{r^2}+O_2(\frac{\varepsilon}{r^{2+\sigma}}) \end{multline*}
where we used the estimate $\pa_{\Lbs}(\sqrt{|g|})=O_2(\varepsilon r^{-1})$. Here and in what follows the repeated indices $i, j$ are summed over $1,2,3$. We note that the error term is of order $O_2(\varepsilon r^{-2-\sigma})$ because it only depends on the metric. For the second term, we have
\begin{multline*}	-\frac12\Ls_{\!\!\beta}\Ls_{\!\!\alpha}\pa_{\Ls}\left(g^{\alpha\beta}\sqrt{|g|}\right)=-\frac12\pa_{\Ls}\left(\Ls_{\!\!\beta}\Lbs_{\!\!\alpha} g^{\alpha\beta}\sqrt{|g|}\right)+\frac12\pa_{\Ls}(\Lbs_{\!\!\alpha})\Ls_{\!\!\beta}g^{\alpha\beta}\sqrt{|g|}+\frac12\pa_{\Ls}(\Ls_{\!\!\beta})\Lbs_{\!\!\alpha}g^{\alpha\beta}\sqrt{|g|}\\
	=-\frac{M}{r^2}-\frac12\pa_{\Ls}(\widehat{h}_1^{v^*u^*}) -\frac{M}{r^2}+\frac12\pa_{\Ls}(\widehat{h}_1^{v^*u^*})+\frac{M}{2r^2}+\frac{M}{2r^2}+O_2(\frac{\varepsilon}{r^{2+\sigma}})
	=-\frac{M}{r^2}+O_2(\frac{\varepsilon}{r^{2+\sigma}}).
\end{multline*}
Before analyzing the last term, we calculate the Christoffle symbols  $\widehat{{\Gamma}}^c_{ab}$ on sphere under the coordinates $\widehat{y}^a$:
\[
\widehat{\Gamma}^c_{ab}=-\frac{\pa x^\alpha}{\pa \widehat{y}^a}\frac{\pa x^\beta}{\pa \widehat{y}^b}\frac{\pa^2\widehat{y}^c}{\pa x^\alpha\pa x^\beta}=-\widehat{A}^\alpha_a\pa_{\widehat{y}^b}\widehat{A}^c_\alpha.
\]
As we have $\widehat{A}_\alpha^c\widehat{A}_c^\beta=(\delta_i^j-\omega_i\omega^j)\delta_{j\alpha}\delta^{i\beta}$, we see that
\begin{multline*}
	\widehat{A}^c_\beta\widehat{\Gamma}^a_{cb}=-\widehat{A}^c_\beta \widehat{A}^\alpha_c\pa_{\widehat{y}^b}\widehat{A}^a_\alpha=-\pa_{\widehat{y}^b}\widehat{A}^a_\beta
-\widehat{A}^a_\alpha\omega^j\delta_{j\beta}\pa_{\widehat{y}^b}(\omega_i\delta^{i\alpha}) =-\pa_{\widehat{y}^b}\widehat{A}^a_\beta-\widehat{A}^a_\alpha\omega^j\delta_{j\beta}\widehat{A}_b^\mu\pa_\mu(\omega_i\delta^{i\alpha})\\	=-\pa_{\widehat{y}^b}\widehat{A}^a_\beta
-\widehat{A}^a_\alpha\omega^j\delta_{j\beta}\widehat{A}_b^\mu\frac{\delta_k^l-\omega_k\omega^l}{r}\delta_{l\mu}\delta^{k\alpha}
	=-\pa_{\widehat{y}^b}\widehat{A}^a_\beta-\frac{\delta^a_b\omega^j\delta_{j\beta}}{r}.
\end{multline*}
Finally, we compute
\begin{multline*}
	\Ls_{\!\!\beta} \widehat{A}_\alpha^c\pa_{\widehat{y}^c}\left(g^{\alpha\beta}\sqrt{|g|}\right)=\pa_{\widehat{y}^c}\left(\Ls_{\!\!\beta} \widehat{A}_\alpha^c g^{\alpha\beta}\sqrt{|g|}\right)-\pa_{\widehat{y}^c}(\Ls_{\!\!\beta})\widehat{A}^c_\alpha g^{\alpha\beta}\sqrt{|g|}-\pa_{\widehat{y}^c}(\widehat{A}^c_\alpha)\Ls_{\!\!\beta} g^{\alpha\beta}\sqrt{|g|}\\ =-\rho'\frac{\delta^{ij}-\omega^i\omega^j}{r}\delta_{i\alpha}\delta_{j\beta}g^{\alpha\beta}\sqrt{|g|}+(\widehat{A}^d_\alpha\widehat{\Gamma}_{dc}^c+\frac{\omega^i\delta_{i\alpha}\delta^c_c}{r})\Ls_{\!\!\beta} g^{\alpha\beta}\sqrt{|g|}\\
	=-\frac{2}{r}(1+\!\frac Mr)(1-\!\frac Mr)\sqrt{|g|}+\frac{\rho'}{r}\slashed{\tr} h^1+\frac 2r\frac{\Ls_{\!\!\alpha}-\Lbs_{\!\!\alpha}}{2\rho'}L^*_\beta g^{\alpha\beta}\sqrt{|g|}+O_2(\frac{\varepsilon}{r^{2+\sigma}})
	=-\frac{\widehat{h}^{v^*u^*}_1}{r}+O_2(\frac{\varepsilon}{r^{2+\sigma}}).
\end{multline*}
Gathering our estimates yield the lemma.
\end{proof}

\subsubsection{Contraction with $\widehat{A}^a$}
\begin{proposition}
	We have
	\begin{equation}\label{eq:wcc:Y}
		\frac 12\pa_{\Lbs}(\widehat{h}_1^{u^*a})+\frac 12\pa_{\Ls}(\widehat{h}_1^{v^*a})+\frac{2\widehat{h}^{v^*a}_1}{r}+\frac{1}{2r^2}\widehat{q}^{\,ac}\pa_{\widehat{y}^c}(\widehat{h}_1^{v^*u^*})+\widehat{\slashed{\nabla}}_c\widehat{h}_1^{ac}=O_2(\frac{\varepsilon}{r^{3+\sigma}}).
	\end{equation}
\end{proposition}
\begin{proof}
Contracting \eqref{eq:wcc} with $\widehat{A}^a_\beta$ yields
\begin{align*}
	\widehat{A}^a_\beta\pa_\alpha\left(g^{\alpha\beta}\sqrt{|g|}\right)=-\frac12\widehat{A}^a_\beta \Ls_{\!\!\alpha} \pa_{\Lbs}\left(g^{\alpha\beta}\sqrt{|g|}\right)-\frac12\widehat{A}^a_\beta\Lbs_{\!\!\alpha}\pa_{\Ls}\left(g^{\alpha\beta}\sqrt{|g|}\right)+\widehat{A}^a_\beta \widehat{A}_\alpha^c\pa_{\widehat{y}^c}\left(g^{\alpha\beta}\sqrt{|g|}\right)=0.
\end{align*}
As for the first term, since we have $\pa_r\widehat{A}^a_\alpha=-\widehat{A}^a_\alpha/r$, we see that
\begin{align*}
	-\frac12\widehat{A}^a_\beta \Ls_{\!\!\alpha} \pa_{\Lbs}\!\!\left(g^{\alpha\beta}\sqrt{|g|}\right)&=-\frac12\pa_{\Lbs}\!\!\left(\widehat{A}^a_\beta \Ls_{\!\!\alpha} g^{\alpha\beta}\sqrt{|g|}\right)+\frac 12\pa_{\Lbs}(\Ls_{\!\!\alpha})	\widehat{A}^a_\beta g^{\alpha\beta}\sqrt{|g|}+\frac 12\pa_{\Lbs}(\Ls_{\!\!\alpha})	\widehat{A}^a_\beta g^{\alpha\beta}\sqrt{|g|}\\
	&=\frac 12\pa_{\Lbs}(\widehat{h}_1^{u^*a})+O_2(\frac{\varepsilon}{r^{3+\sigma}}).
\end{align*}
Next we calculate the second term
\begin{align*}
	-\frac12\widehat{A}^a_\beta\Lbs_{\!\!\alpha}\pa_{\Ls}\!\!\left(g^{\alpha\beta}\sqrt{|g|}\right)&=-\frac12\pa_{\Ls}\!\!\left(\widehat{A}^a_\beta \Lbs_{\!\!\alpha} g^{\alpha\beta}\sqrt{|g|}\right)+\frac 12\pa_{\Ls}(\Lbs_{\!\!\alpha})	\widehat{A}^a_\beta g^{\alpha\beta}\sqrt{|g|}+\frac 12\pa_{\Ls}(\Lbs_{\!\!\alpha})	\widehat{A}^a_\beta g^{\alpha\beta}\sqrt{|g|}\\
	&=\frac 12\pa_{\Ls}(\widehat{h}_1^{v^*a})+\frac{\widehat{h}^{v^*a}_1}{2r}+O_2(\frac{\varepsilon}{r^{3+\sigma}}).
\end{align*}
Finally,
\begin{align*}
	\widehat{A}^a_\beta \widehat{A}_\alpha^c\pa_{\widehat{y}^c}\left(g^{\alpha\beta}\sqrt{|g|}\right)&=\widehat{A}^a_\beta \widehat{A}_\alpha^cg^{\alpha\beta}_0\pa_{\widehat{y}^c}\left(\sqrt{|g|}\right)+\widehat{A}^a_\beta \widehat{A}_\alpha^c\pa_{\widehat{y}^c}\left(h_1^{\alpha\beta}\sqrt{|g|}\right)\\
	&=\frac{1}{2r^2}\widehat{q}^{\,ac}\pa_{\widehat{y}^c}(\widehat{h}_1^{v^*u^*})+\widehat{A}^a_\beta \widehat{A}_\alpha^c\pa_{\widehat{y}^c}\left(h_1^{\alpha\beta}\sqrt{|g|}\right)+O_2(\frac{\varepsilon}{r^{3+\sigma}}).
\end{align*}
We write
\begin{multline*}
	\widehat{A}^a_\beta \widehat{A}_\alpha^c\pa_{\widehat{y}^c}\left(h_1^{\alpha\beta}\sqrt{|g|}\right)
	=\pa_{\widehat{y}^c}\left(\widehat{A}^a_\beta \widehat{A}_\alpha^ch_1^{\alpha\beta}\sqrt{|g|}\right)
-\widehat{A}_\alpha^c\pa_{\widehat{y}^c}(\widehat{A}^a_\beta)h_1^{\alpha\beta}\sqrt{|g|}-\widehat{A}^a_\beta \pa_{\widehat{y}^c}(\widehat{A}^c_\alpha)h_1^{\alpha\beta}\sqrt{|g|}\\
	=\pa_{\widehat{y}^c}(\widehat{h}_1^{ac})+\Big(\widehat{A}^c_\alpha \widehat{A}^d_\beta\widehat{\Gamma}_{dc}^a+\frac{\delta_c^a\omega^j\delta_{j\beta}}{r}+\widehat{A}^a_\beta \widehat{A}^d_\alpha\widehat{\Gamma}_{dc}^c
+\frac{\delta_c^c\omega^j\delta_{j\alpha}}{r}\Big)h_1^{\alpha\beta}\sqrt{|g|}+O_2(\frac{1}{r^{3+\sigma}})\\ =\widehat{\slashed{\nabla}}_c\widehat{h}_1^{ac}+\frac{\widehat{h}_1^{v^*a}}{2r}+\frac{\widehat{h}_1^{v^*a}}{r}+O_2(\frac{\varepsilon}{r^{3+\sigma}})=\widehat{\slashed{\nabla}}_c\widehat{h}_1^{ac}+\frac{3\widehat{h}_1^{v^*a}}{2r}+O_2(\frac{\varepsilon}{r^{3+\sigma}})
\end{multline*}
where $\widehat{\slashed{\nabla}}$ is the covariant derivative on sphere.
Putting all together yields the conclusion.
\end{proof}

\subsubsection{Contraction with $\Lbs$}
\begin{proposition}\label{prop:wccLb}
	We have
	\begin{equation}\label{eq:wccLb}
		\frac12\pa_{\Lbs}(\slashed{\tr}h^1)+\frac{\widehat{h}_1^{v^*u^*}}{r}-\frac12\pa_{\Ls}(\widehat{h}_1^{v^*v^*})-\frac{\widehat{h}_1^{v^*v^*}}{r}-\widehat{\slashed{\nabla}}_c(\widehat{h}_1^{v^*c})-\frac 14\pa_{\Lbs}(\widehat{h}_1^{ac}r^2\widehat{q}_{cb}\widehat{h}_1^{bd}r^2\widehat{q}_{da})=O_2(\frac{\varepsilon}{r^{2+\sigma}}).
	\end{equation}
\end{proposition}
 In order to prove this proposition, we now need the following lemmas.
\begin{lemma}\label{lem:deriofdeter}
	We have
	\[\pa_{\Lbs}(\sqrt{|g|})=\frac{M}{r^2}-\frac 12\pa_{\Lbs}(h^1_{\Ls\Lbs})+\frac{M}{r}\pa_{\Lbs}(\widehat{h}_1^{v^*u^*})+\frac12\pa_{\Lbs}(\slashed{\tr}h^1)-\frac 14\pa_{\Lbs}(\widehat{h}_1^{ac}r^2\widehat{q}_{cb}\widehat{h}_1^{bd}r^2\widehat{q}_{da})+O_2(\frac{\varepsilon}{r^{2+\sigma}}).\]
\end{lemma}
\begin{proof}
	We first notice that $h^1_{UV}=h^1_{\alpha\beta}U^\alpha V^\beta$. Given the facts that $2\pa_\alpha(\sqrt{|g|})=\sqrt{|g|}g^{\mu\nu}\pa_{\alpha}g_{\mu\nu}$,  $\sqrt{|g|}=1+\frac{M}{r}-\frac 12h^1_{L\Lb}+O_2(\varepsilon r^{-2-\sigma})$ and $h^1_{L\Lb}=h^1_{\Ls\Lbs}+O_3(\varepsilon r^{-1-\sigma})=-\widehat{h}_1^{v^*u^*}+O_3(\varepsilon r^{-1-\sigma})$, we obtain
	\begin{multline*}	
		\pa_{\Lbs}(\sqrt{|g|})=\sqrt{|g|}\Big(\frac{M}{r^2}+\frac 12g^{\alpha\beta}\pa_{\Lbs}(h^1_{\alpha\beta})\Big)
		=\sqrt{|g|}\Big(\frac{M}{r^2}-\frac 12\pa_{\Lbs}(h^1_{\Ls\Lbs})-\frac{M}{2r}\pa_{\Lbs}(h^1_{\Ls\Lbs})
+\frac12\pa_{\Lbs}(\slashed{\tr}h^1)\Big)\\
+\sqrt{|g|}\Big(-\frac 14h^1_{\Ls\Lbs}\pa_{\Lbs}(h^1_{\Ls\Lbs})-\frac 14\pa_{\Lbs}(\widehat{h}_1^{ac}r^2\widehat{q}_{cb}\widehat{h}_1^{bd}r^2
\widehat{q}_{da})\Big)+O_2(\frac{\varepsilon}{r^{2+\sigma}})\\
		=\frac{M}{r^2}-\frac 12\pa_{\Lbs}\!(h^1_{\Ls\Lbs}\!)+\frac{M}{r}\pa_{\Lbs}\!(\widehat{h}_1^{v^*u^*}\!)
+\frac12\pa_{\Lbs}(\slashed{\tr}h^1)-\frac 14\pa_{\Lbs}(\widehat{h}_1^{ac}r^2\widehat{q}_{cb}\widehat{h}_1^{bd}r^2
\widehat{q}_{da})+O_2(\frac{\varepsilon}{r^{2+\sigma}}).
 \tag*{\qedhere}
	\end{multline*}
\end{proof}
\begin{lemma}\label{lem:deriofLLb}
	We have 	
	\[
	\pa_{\Lbs}(h^1_{\Ls\Lbs})+\pa_{\Lbs}(\widehat{h}_1^{v^*u^*})=-\widehat{h}_1^{v^*u^*}\pa_{\Lbs}(\widehat{h}_1^{v^*u^*})+\frac{2M}{r}\pa_{\Lbs}(\widehat{h}_1^{v^*u^*})+O_2(\frac{\varepsilon}{r^{2+\sigma}}).
	\]
\end{lemma}
\begin{proof}
	Since we express $g_{\alpha\beta}=g^0_{\alpha\beta}+h^1_{\alpha\beta}$ and $g^{\alpha\beta}=g_0^{\alpha\beta}+h_1^{\alpha\beta}$, we see that
	\[
	h_1^{\alpha\beta}=-m^{\alpha\mu}h^1_{\mu\nu}m^{\nu\beta}+m^{\alpha\alpha'}(\frac Mr\delta_{\alpha'\mu}+h^1_{\alpha'\mu})m^{\mu\nu}(\frac Mr\delta_{\nu\beta'}+h^1_{\nu\beta'})m^{\beta'\beta}
	+O_3(\frac{\varepsilon}{r^{2+\sigma}}).\]
	Therefore
	\begin{equation*}
		\widehat{h}_1^{v^*u^*}=h_1^{\alpha\beta}\Lbs_{\!\!\alpha}\Ls_{\!\!\beta}
		=-m ^{\alpha\mu}h^1_{\mu\nu}m^{\nu\beta}\Lbs_{\!\!\alpha}\Ls_{\!\!\beta}+\Lb^{\alpha'}L^{\beta'}m^{\mu\nu}(\frac Mr\delta_{\alpha'\mu}+h^1_{\alpha'\mu})(\frac Mr\delta_{\nu\beta'}+h^1_{\nu\beta'})+O_3(\frac{\varepsilon}{r^{2+\sigma}}).
	\end{equation*}
	We analyze the first term
	\begin{align*}
		-m ^{\alpha\mu}h^1_{\mu\nu}m^{\nu\beta}\Lbs_{\!\!\alpha}\Ls_{\!\!\beta}&=-h^1_{\mu\nu}(\Lbs^{\mu}+(\frac{1}{\rho'}-\rho')\omega_j\delta^{j\mu})(\Ls^{\nu}+(\rho'-\frac{1}{\rho'})\omega_j\delta^{j\nu})\\
		&=-h^1_{\Ls\Lbs}+\frac Mr(h^1_{\Ls\Ls}+h^1_{\Lbs\Lbs}-2h^1_{\Ls\Lbs})+
		O_3(\frac{\varepsilon}{r^{2+\sigma}}).
	\end{align*}
	Now we turn to the second term. Using the null frame $(L=\pa_t+\pa_r, \Lb=\pa_t-\pa_r, A, B)$, we rewrite $\Lb^{\alpha'}\delta_{\alpha'\mu}$, $L^{\beta'}\delta_{\nu\beta'}$, $\Lb^{\alpha'}h^1_{\alpha'\mu}$, $L^{\beta'}h^1_{\nu\beta'}$  as follows
	\begin{align*}
		\Lb^{\alpha'}h^1_{\alpha'\mu}&=-\frac 12\Lb_\mu h^1_{L\Lb}-\frac 12L_\mu h^1_{\Lb\Lb}+A_\mu h^1_{A\Lb},\qquad\Lb^{\alpha'}\delta_{\alpha'\mu}=-\frac 12L_\mu\delta_{\Lb\Lb}=-L_\mu,\\
		L^{\beta'}h^1_{\nu\beta'}&=-\frac 12\Lb_\nu h^1_{LL}-\frac 12L_\nu h^1_{\Lb L}+A_\nu h^1_{AL},\qquad L^{\beta'}\delta_{\nu\beta'}=-\frac 12\Lb_\nu\delta_{LL}=-\Lb_\nu.
	\end{align*}
	We see that
	\begin{align*}
		\Lb^{\alpha'}L^{\beta'}m^{\mu\nu}(\frac Mr\delta_{\alpha'\mu}+h^1_{\alpha'\mu})(\frac Mr\delta_{\nu\beta'}+h^1_{\nu\beta'})=-\frac 12h^1_{L\Lb}h^1_{L\Lb}-\frac Mrh^1_{LL}-\frac Mrh^1_{\Lb\Lb}-\frac{2M}{r^2}+O_3(\frac{\varepsilon}{r^{2+\sigma}}).
	\end{align*}
	Thus we find
	\begin{align*}
		\widehat{h}_1^{v^*u^*}+h^1_{\Ls\Lbs}=-\frac 12h^1_{L\Lb}h^1_{L\Lb}-\frac {2M}{r}h^1_{L\Lb}-\frac{2M}{r^2}+O_3(\frac{\varepsilon}{r^{2+\sigma}})
	\end{align*}
	and
	\begin{equation*}		\pa_{\Lbs}(\widehat{h}_1^{v^*\!u^*}\!\!\!+h^1_{\Ls\Lbs}\!)\!=-h^1_{L\Lb}
\pa_{\Lbs}(h^1_{L\Lb})-\frac {2M\!}{r}\pa_{\Lbs}(h^1_{L\Lb})+O_2(\frac{\varepsilon}{r^{2+\sigma}}) =-\widehat{h}_1^{v^*\!u^*}\!\pa_{\Lbs}(\widehat{h}_1^{v^*\!u^*}\!)
+\frac{2M\!}{r}\pa_{\Lbs}(\hat{h}_1^{v^*\!u^*}\!)
+O_2(\frac{\varepsilon}{r^{2+\sigma}}).
	\end{equation*}
\end{proof}

\begin{proof}[Proof of Proposition \ref{prop:wccLb}]
	Contracting with $\Lbs_{\!\!\beta}$ we see that
	\begin{align*}
		\Lbs_{\!\!\beta}\pa_\alpha\left(g^{\alpha\beta}\sqrt{|g|}\right)=-\frac12\Lbs_{\!\!\beta} \Ls_{\!\!\alpha} \pa_{\Lbs}\left(g^{\alpha\beta}\sqrt{|g|}\right)-\frac12\Lbs_{\!\!\beta}\Lbs_{\!\!\alpha}\pa_{\Ls}\left(g^{\alpha\beta}\sqrt{|g|}\right)+\Lbs_{\!\!\beta} \widehat{A}_\alpha^c\pa_{\widehat{y}^c}\left(g^{\alpha\beta}\sqrt{|g|}\right)=0.
	\end{align*}
	We first analyze the last two expressions
	\begin{equation*} -\frac12\Lbs_{\!\!\beta}\Lbs_{\!\!\alpha}\pa_{\Ls}\left(g^{\alpha\beta}\sqrt{|g|}\right)=-\frac12\pa_{\Ls}\left(\Lbs_{\!\!\beta}\Lbs_{\!\!\alpha} g^{\alpha\beta}\sqrt{|g|}\right)+\pa_{\Ls}(\Lbs_{\!\!\alpha})\Lbs_{\!\!\beta}g^{\alpha\beta}\sqrt{|g|}=-\frac12\pa_{\Ls}(\widehat{h}_1^{v^*v^*})-\frac{M}{r^2}+O_2(\frac{\varepsilon}{r^{2+\sigma}})
	\end{equation*}
	and
	\begin{multline*}
		\Lbs_{\!\!\beta} \widehat{A}_\alpha^c\pa_{\widehat{y}^c}\left(g^{\alpha\beta}\sqrt{|g|}\right)=\pa_{\widehat{y}^c}\!\left(\Lbs_{\!\!\beta} \widehat{A}_\alpha^cg^{\alpha\beta}\sqrt{|g|}\right)-\pa_{\widehat{y}^c}(\widehat{A}^c_\alpha)\Lbs_{\!\!\beta} g^{\alpha\beta}\sqrt{|g|}-\pa_{\widehat{y}^c}(\Lbs_{\!\!\beta})\widehat{A}^c_\alpha g^{\alpha\beta}\sqrt{|g|}\\ =-\pa_{\widehat{y}^c}(\widehat{h}_1^{v^*c})+(\widehat{A}^d_\alpha\widehat{\Gamma}^c_{dc}
		+\frac{\delta_c^c\omega^j\delta_{j\beta}}{r})\Lbs_{\!\!\beta} g^{\alpha\beta}\sqrt{|g|}+\rho'\frac{\delta^{ij}-\omega^i\omega^j}{r}\delta_{i\alpha}\delta_{j\beta}g^{\alpha\beta}
		+O_2(\frac{\varepsilon}{r^{2+\sigma}})\\
		=-\widehat{\slashed{\nabla}}_c(\widehat{h}_1^{v^*c})
		+\frac{\widehat{h}_1^{v^*u^*}}{r}-\frac{\widehat{h}_1^{v^*v^*}}{r}+O_2(\frac{\varepsilon}{r^{2+\sigma}}).
	\end{multline*}
	Then it remains to compute the term $-\frac12\Lbs_{\!\!\beta} \Ls_{\!\!\alpha} \pa_{\Lbs}\left(g^{\alpha\beta}\sqrt{|g|}\right)$.
Now we compute
\begin{align*}
	-\frac12\Lbs_{\!\!\beta}\Ls_{\!\!\alpha} \pa_{\Lbs}\left(g^{\alpha\beta}\sqrt{|g|}\right)=-\frac12 \pa_{\Lbs}\left(\Lbs_{\!\!\beta}\Ls_{\!\!\alpha} g^{\alpha\beta}\sqrt{|g|}\right)-\frac{M}{r^2}+O_2(\frac{\varepsilon}{r^{2+\sigma}}).
\end{align*}
Using Lemma ~\ref{lem:deriofdeter} and Lemma ~\ref{lem:deriofLLb}, we further calculate
\begin{multline*}
	-\frac12 \pa_{\Lbs}\!\big(\Lbs_{\!\!\beta}\Ls_{\!\!\alpha} g^{\alpha\beta}\!\sqrt{|g|}\big)
	=\frac{M}{r^2}-\frac 12\pa_{\Lbs}(\widehat{h}_1^{v^*\!u^*})-\frac{M}{2r}\pa_{\Lbs}(\widehat{h}_1^{v^*\!u^*})-\frac 14\widehat{h}_1^{v^*\!u^*}\pa_{\Lbs}(\widehat{h}_1^{v^*u^*})-\frac 12\widehat{g}^{v^*\!u^*}\pa_{\Lbs}(\sqrt{|g|})+O_2(\frac{\varepsilon}{r^{2+\sigma}})\\
	=\frac{2M}{r^2}+\frac12\pa_{\Lbs}(\slashed{\tr}h^1)-\frac 14\pa_{\Lbs}(\widehat{h}_1^{ac}r^2\widehat{q}_{cb}\widehat{h}_1^{bd}r^2\widehat{q}_{da})+O_2(\frac{\varepsilon}{r^{2+\sigma}}).
\end{multline*}
Putting all together finishes the proof.
\end{proof}

\section{Construction of outgoing characteristic surfaces}\label{sec:3}
In \cite{Lind17} we show that the eikonal equation
\beq \label{eq:eikonalintro} g^{\alpha\beta}\pa_\alpha u\, \pa_\beta
u=0,\qquad \text{in}\quad r>|t|/2 ,
\eq
has a unique solution with asymptotic data at infinity $u\sim \us\!=t-r^*$, as $t\to\infty$.

\begin{remark}
		In the construction of $u$ coordinate,  we may change the asymptotic data imposed at infinity, i.e., we may require $u\sim u^*+f(\widehat{y}^3, \widehat{y}^4)$. This leads to the transformation at future null infinity $\mathscr{I}^+: u\to u+f(\overline{y}^3, \overline{y}^4)$ where $(\overline{y}^3, \overline{y}^4)$ are the angular coordinates in the Bondi-Sachs coordinate system associated to $u$ whose construction will be given in Subsection \ref{subsec:6.1}. So the changes of asymptotic data for $u$ of this type generate the \textit{supertranslations} which is an infinite dimensional subgroup of the asymptotic symmetry group at null infinity---Bondi-Metzner-Sachs group \cite{S62a, W84, MW16}.
\end{remark}
\begin{prop}[{\cite[Proposition 26]{Lind17}}]\label{prop:eikonalintro} The eikonal equation \eqref{eq:eikonalintro} has a solution
	$u\!=\mathring{u}+\us$ satisfying
	\beq\label{eq:estimateofu}
	{\sum}_{|I|\leq 2}|{Z}^{* I}\mathring{u}|
	\leq C_1\varepsilon\Big(\frac{1+(r^*\!-|t|)_-}{1+t+|\,q^*|}\Big)^{\gamma\prime}, \quad  r>|t|/2.
	\eq
\end{prop}
\begin{remark}\label{rem:Ithree}
	Following the proof of Proposition 26 in \cite{Lind17}, we can prove Proposition \ref{prop:eikonalintro} for $|I|\leq 3$. We commutate the vector fields $X\in\mathcal{X}=\{S^*=t\pa_t+x^{*i}\pa_{x^{*i}}, ~\Omega_{ij}, ~\pa_t\}$ through the equation \eqref{eq:eikonalintro}. Let
	$\widetilde{X}\!\!=\!X\!-\!\delta_{X \!\Ss}\!$ and
	$\widetilde{\mathcal{L}}_{\!X}\!=\!\mathcal{L}_{\!X}\!+\!2\delta_{X\!\Ss}$,
	where $\delta_{X\! \Ss}\!\!=\!1$ if $X\!\!=\Ss\!\!$, and $=\!0$ otherwise. In fact, when $|I|=3$ we consider the equation $\pa_{\widetilde{L}}\widetilde{X}\widetilde{Y}\widetilde{Z}\tilde{u}=-H(g, u)/2$ where $\widetilde{L}^\alpha=g^{\alpha\beta}\pa_\beta u$ and
	\begin{align*}
		H(g,u)&=\widetilde{\mathcal{L}}_{\!X}\widetilde{\mathcal{L}}_{\!Y}\widetilde{\mathcal{L}}_{\!Z} g(\pa u, \pa u)\\
		&+2\widetilde{\mathcal{L}}_{\!X}\widetilde{\mathcal{L}}_{\!Y}g(\pa\widetilde{Z}u, \pa u)+2\widetilde{\mathcal{L}}_{\!X}\widetilde{\mathcal{L}}_{\!Z} g(\pa\widetilde{Y}u, \pa u)+2\widetilde{\mathcal{L}}_{\!Y}\widetilde{\mathcal{L}}_{\!Z} g(\pa\widetilde{X}u, \pa u)+2\widetilde{\mathcal{L}}_{\!X}g(\pa\widetilde{Y}\widetilde{Z}u, \pa u)\\
		&+2\widetilde{\mathcal{L}}_{\!Y}g(\pa\widetilde{X}\widetilde{Z}u, \pa u)+2\widetilde{\mathcal{L}}_{\!Z}g(\pa\widetilde{X}\widetilde{Y}u, \pa u)+2\widetilde{\mathcal{L}}_{\!X}g(\pa\widetilde{Z}u, \pa \widetilde{Y}u)+2\widetilde{\mathcal{L}}_{\!Y}g(\pa\widetilde{Z}u, \pa \widetilde{X}u)\\
		&+2\widetilde{\mathcal{L}}_{\!Z}g(\pa\widetilde{X}u, \pa \widetilde{Y}u)+2g(\pa\widetilde{X}\widetilde{Y}u, \pa\widetilde{Z}u)+2g(\pa\widetilde{X}\widetilde{Z}u, \pa\widetilde{Y}u)+2g(\pa\widetilde{Y}\widetilde{Z}u, \pa\widetilde{X}u)
	\end{align*}
	We notice that only the new term  $\widetilde{\mathcal{L}}_X\widetilde{\mathcal{L}}_Y\widetilde{\mathcal{L}}_Zg_0(\pa u, \pa u)$ needs additional analysis. Since we already know that $\pa\widetilde{X}u^*=0$ and $\widetilde{\mathcal{L}}_X\widetilde{\mathcal{L}}_Yg_0(\pa u^*, \pa u^*)=0$, we obtain
	\[
	\widetilde{\mathcal{L}}_X\widetilde{\mathcal{L}}_Y\widetilde{\mathcal{L}}_Zg_0(\pa u^*, \pa u^*)=X\Bigl(\widetilde{\mathcal{L}}_Y\widetilde{\mathcal{L}}_Zg_0(\pa u^*, \pa u^*)\Bigr)-2\widetilde{\mathcal{L}}_Y\widetilde{\mathcal{L}}_Zg_0(\pa \widetilde{X}u^*, \pa u^*)=0.
	\]
	Moreover, we have $\widetilde{\mathcal{L}}_{\pa_t} g_0\!=\!\widetilde{\mathcal{L}}_\Omega g_0\!=\!0$ and
	$\widetilde{\mathcal{L}}_{\Ss} g_0\!=\!\kappa_3 g_0\!-2(\kappa_1\!-\kappa_2)\overline{g}_0$ where $\kappa_1\!\sim \! M\! \ln{r}\!/r$, $\kappa_2\!\sim \! \kappa_3\!\sim \!M\!/r$ and $\overline{g\overline{}^{}}_0(\pa u,\pa v)\!={g}_0^{ij}\pas_{\!i } u\,\pas_{\!\!j}v$. Using $\overline{g\overline{}^{}}_0(\pa \widetilde{X}^Iu^*,\pa w)=0$ for $|I|\leq2$ we obtain $\widetilde{\mathcal{L}}_Y\overline{g\overline{}^{}}_0(\pa \widetilde{X}^Iu^*,\pa w)=-\overline{g\overline{}^{}}_0(\pa \widetilde{Y}\widetilde{X}^Iu^*,\pa w)-\overline{g\overline{}^{}}_0(\pa \widetilde{X}^Iu^*,\pa \widetilde{Y}w)=0$ for $|I|\leq1$ and then $\widetilde{\mathcal{L}}_X\widetilde{\mathcal{L}}_Y\overline{g\overline{}^{}}_0(\pa u^*,\pa w)=-\widetilde{\mathcal{L}}_Y\overline{g\overline{}^{}}_0(\pa \widetilde{X}u^*,\pa w)-\widetilde{\mathcal{L}}_Y\overline{g\overline{}^{}}_0(\pa u^*,\pa \widetilde{X}w)=0$. Hence
	\[
	|\widetilde{\mathcal{L}}_X\widetilde{\mathcal{L}}_Y\widetilde{\mathcal{L}}_Zg_0(\pa u^*, \pa \mathring{u})|\les|\pa_{\Ls}\mathring{u}|.
	\]
	Using $\widetilde{\mathcal{L}}_X\overline{g\overline{}^{}}_0(\pa v,\pa w)=X(\overline{g\overline{}^{}}_0(\pa v,\pa w))-\overline{g\overline{}^{}}_0(\pa \widetilde{X}v,\pa w)-\overline{g\overline{}^{}}_0(\pa v,\pa \widetilde{X}w)$ and the expression for $\widetilde{\mathcal{L}}_{\Ss} g_0$ we have
	\[
	|\widetilde{\mathcal{L}}_X\widetilde{\mathcal{L}}_Y\widetilde{\mathcal{L}}_Zg_0(\pa \mathring{u}, \pa \mathring{u})|\les |\pa_{\Ls}\mathring{u}||\pa_{\Lbs}\mathring{u}|+(|\slashed{\pa}\mathring{u}|+|\slashed{\pa}\widetilde{X}\mathring{u}|)(|\slashed{\pa}\mathring{u}|+|\slashed{\pa}\widetilde{Y}\mathring{u}|)
	+|\slashed{\pa}\mathring{u}||\slashed{\pa}\widetilde{X}\widetilde{Y}\mathring{u}|.
	\]
	Finally the estimates for the remaining terms in $H(g, u)$ and thus the bounds of $Z^{*I}\mathring{u}$ with $|I|=3$ follow as in Proposition 26 in \cite{Lind17}. Then we have
	\[
	|\pa_{\Ls}Z^{*I}\mathring{u}|=O(\frac{\varepsilon}{r^{1+\sigma}}),
\quad\text{and}\quad |\pa_{\Lbs}Z^{*I}\mathring{u}|+|\pa_{\widehat{y}^a}Z^{*I}\mathring{u}|
=O(\frac{\varepsilon}{r^{\sigma}}),\quad\text{for}\quad|I|\leq2.
	\]
\end{remark}

The estimate for $\pa_{\Lbs}Z^{*I}\mathring{u}$ is not precise enough and we need to sharpen it. We first record three lemmas in \cite{Lind17}, which will be of use when refining the expressions for $\pa_{\Lbs}Z^{*I}\mathring{u}$.
\begin{lemma} [{\cite[Lemma 21]{Lind17}}]If $Z=\partial_t$
	then with $h_1^{\alpha\beta}\!\!=\!g^{\alpha\beta}\!-g_0^{\alpha\beta}\!\!$ and
	$\widetilde{L}^\alpha=g^{\alpha\beta}\pa_\beta u$
	we have
	\begin{equation}\label{eq:eikonalsystemone}
		\partial_{\widetilde{L}} Z \mathring{u}
		=-\tfrac{1}{2}h_{1Z}(\partial u,\partial u)
	\end{equation}
	with the notation $h_{1Z}(U,V)=h_{1Z}^{\alpha\beta}U_\alpha V_\beta$ where the Lie derivative $h_{1Z}^{\alpha\beta}=\mathcal{L}_Z h_1^{\alpha\beta}$ is given by
	\beq\label{eq:liederivative}
	h_{1Z}^{\alpha\beta}\partial_\alpha u\,\partial_\beta w=
	(Zh_1^{\alpha\beta}) \partial_\alpha u\,\partial_\beta w
	+h_1^{\alpha\beta} \partial_\alpha u\, [Z,\partial_\beta] w+
	h_1^{\alpha\beta} [Z,\partial_\alpha] u\,\partial_\beta w.
	\eq
\end{lemma}

\begin{lemma} [{\cite[Lemma 25]{Lind17}}]We  have
	\begin{equation}\label{eq:htestone}
		h_{1\partial_t}(\partial u,\partial u)=\pa_t(\widehat{h}_1^{u^*u^*})+2\pa_t(\widehat{h}_1^{u^*u^*})\pa_t\mathring{u}+\pa_t(\widehat{h}_1^{u^*u^*})(\pa_t\mathring{u})^2+O(\frac{\varepsilon}{r^{2+\sigma}}).	
	\end{equation}
\end{lemma}

\begin{lemma} [{\cite[Lemma 24]{Lind17}}]\label{lem:quadraticterms}
	If $\Omega\!=\!x^i\pa_j-x^j\pa_i$ then with
	$k^{\alpha \Omega/r}\!\!=k^{\alpha i} \omega_{j}-k^{\alpha j}
	\omega_i$ we have
	\begin{align}\label{eq:LieOmega}
		(\mathcal{L}_\Omega k)(\pa u,\pa v)&=(\Omega k)(\pa u,\pa v)+k([\Omega,\pa]u,\pa v)
		+k(\pa u,[\Omega,\pa]v),\\
		\label{eq:commutarorrotationquadraticform}
		k^{\alpha\beta} [\pa_\beta, \Omega] u
		&=k^{\alpha \Omega/r}\pa_r u
		+\big(k^{\alpha i}
		\overline{\pa}_{j}-k^{\alpha j} \overline{\pa}_i\big)u.
	\end{align}
\end{lemma}
Now we are ready to refine $\pa_{\Lbs}Z^{*I}\mathring{u}$.
\begin{prop}\label{prop:badderiofu} The eikonal equation \eqref{eq:eikonalintro} has a solution
	$u\!=\mathring{u}+\us$ satisfying
\begin{equation}\label{eq:highbadderi}
	\pa_{\Lb^*} Z^I \mathring{u}= Z^I(\widehat{h}_1^{v^*u^*})+O(\frac{\varepsilon}{r^{1+\sigma}})\quad \text{for}\quad Z\in\{\pa_t, \Omega_{ij}\}\quad\text{and}\quad|I|\leq 2.
\end{equation}
\end{prop}

\begin{proof}
	Putting \eqref{eq:wcc_L}, \eqref{eq:eikonalsystemone} and \eqref{eq:htestone} together, we obtain
	\[
	\partial_{\widetilde{L}} \pa_t \mathring{u}=\frac 12\frac{\widehat{h}_1^{v^*u^*}}{r}+O(\frac{\varepsilon}{r^{2+\sigma}}).
	\]
It follows from Remark \ref{rem:specialsharpmetricdecay} that
	\[
	\partial_{\widetilde{L}} \pa_t \mathring{u}=-\frac 12\pa_{\Ls}(\widehat{h}_1^{v^*u^*})+O(\frac{\varepsilon}{r^{2+\sigma}}).	
	\]
	Since
	\begin{equation}\label{eq:tildaL}
		\widetilde{L}=g^{\alpha\beta}\pa_{\alpha}u\pa_\beta=(-2+O(\frac{\varepsilon}{r^\sigma}))\pa_{v^*}+O(\frac{\varepsilon}{r^{1+\sigma}})\pa_{u^*}+O(\frac{\varepsilon}{r^{2+\sigma}})\pa_{\widehat{x}^a},
	\end{equation}
	we find
	\[
	\partial_{\widetilde{L}} \pa_t \mathring{u}=\frac 12\pa_{\widetilde{L}}(\widehat{h}_1^{v^*u^*})+O(\frac{\varepsilon}{r^{2+\sigma}}).	
	\]
	Therefore we can conclude that
	\[
	\pa_t\mathring{u}=\frac{1}{2}\widehat{h}_1^{v^*u^*}+O(\frac{\varepsilon}{r^{1+\sigma}}).
	\]
If we commutate the vector fields $Z\in\{\pa_t, \Omega_{ij}=x^i\pa_j-x^j\pa_i\}$ through the equation $\pa_{\widetilde{L}}\pa_t \mathring{u}=-\frac 12h_{1\pa_t}(\pa u, \pa u)$, using proposition \ref{prop:eikonalintro}, Remark \ref{rem:Ithree} and Lemma  \ref{lem:quadraticterms} and then integrating along the integral curves of $\widetilde{L}$ yield
\[
	\pa_t Z^I \mathring{u}=\frac12 Z^I(\widehat{h}_1^{v^*u^*})+O(\frac{\varepsilon}{r^{1+\sigma}}).
	\]
Then	\eqref{eq:highbadderi} follows from the fact that $\pa_{\Lbs}=2\pa_t+\pa_{\Ls}$.
\end{proof}

\section{The Trautman mass}\label{sec:4}

\subsection{The asymtotically Schwarzschild coordinates}
In this section we will use the asymptotically Schwarzschild coordinates
$(\widetilde{t},\widetilde{x})$ with
$$
\widetilde{t}=t,\qquad \widetilde{x}^i=\widetilde{r} \omega^i,\qquad
\text{where}\quad \omega^i={x^i}/{r},\qquad \widetilde{r}=r+M\ln{r},\quad
r=|\,x|.
$$
Then
$$
\pa_t=\pa_{\widetilde{t}},\qquad \pa_{x^i} = \Big (\omega^i\omega^j (1+\frac
Mr)+\frac {\widetilde{r}}{r} (\de^{ij}- \omega^i\omega^j)\Big) \pa_{\widetilde{x}^j}
=\Big(\delta^{ij}+\frac{\widetilde{r}-r}{r}(\delta^{ij}-\omega^i\omega^j)
+\frac{M}{r}\omega^i \omega^j\Big)\pa_{\widetilde{x}^j}.
$$
In particular,
$$
\frac {x^i}{|x|} \pa_{x^i} = (1+\frac Mr)\frac {\widetilde{x}^i}{|\widetilde{x}|}
\pa_{\widetilde{x}^i},\qquad \frac {x^i}{|x|} \pa_{x^k} - \frac {x^k}{|x|}
\pa_{x^i}=\frac {\widetilde{r}}{r} \Big (\frac {\widetilde{x}^i}{|\widetilde{x}|} \pa_{\widetilde{x}^k} - \frac
{\widetilde{x}^k}{|\widetilde{x}|} \pa_{\widetilde{x}^i}\Big).
$$
Denote the corresponding metric components by $\g_{\a\b}$. We see
that
$$
\pa_{x^\a} = A_{\a\b} \pa_{\widetilde{x}^\b},
$$
where the matrix $A$ has the form
$$
A_{\a\b}=\de_{\a\b}
+ \frac {M \ln r}{r} (\de^{ij}- \omega^i\omega^j)
\de_{\a i} \de_{\b j} + \frac{M}{r}
\omega^i\omega^j \de_{\a i} \de_{\b j},
$$
where the sums are over $i,j=1,2,3$ only.
As a consequence, we have that
$$
\g^{\a\b} = A_{\a\mu} A_{\b\nu} g^{\mu\nu}.
$$
Expanding the metric
$$
g^{\mu\nu}=m^{\mu\nu} -  M \de^{\mu\nu}/r + h_1^{\mu\nu}
$$
and  using that $|h_1|\le r^{-1} \ln r$
we obtain
\begin{align*}
	\g^{\a\b}&=m^{\a\b} - \frac Mr \de^{\a\b} + h_1^{\a\b}+\frac
	{2M \ln r}{r} (\de^{ij}- \omega^i\omega^j)\de_{\a i} \de_{\b j} +
	\frac{2M}{r} \omega^i\omega^j\de_{\a i} \de_{\b j} + O(\frac
	{\ln^2 r}{r^2}) + O(\frac {\ln r}{r} h_1)\\ &= (1+\frac Mr)
	m^{\a\b} + h_1^{\a\b}+\frac {2M (\ln r-1)}{r}
	(\de^{ij}- \omega^i\omega^j)\de_{\a i} \de_{\b j}
	+ O(\frac {\ln^2 r}{r^2}).
\end{align*}
\begin{lemma}\label{lem:newmetric}
	Relative to the asymptotically Schwarzschild coordinates the metric components
	$\g^{\a\b}\!\!$  verify
	$$
	\g^{\a\b}=(1+\frac Mr) m^{\a\b} +\frac {2M (\ln r-1)}{r}
	(\de^{ij}-\omega^i\omega^j)\de_{\a i} \de_{\b j} + h_1^{\a\b}
	+ O(\frac
	{\ln^2 r}{r^2}),
	$$
	in the region $r>t/2$. Here the sum is over $i,j=1,2,3$ only.
\end{lemma}
Next we examine the wave coordinate expression $\pa_{\widetilde{x}^\a} \big
(\g^{\a\b} \sqrt{|\g|}\big)$ evaluated in $\widetilde{x}$-coordinates.
\begin{lemma}\label{lem:wcoordtilde} Relative to
	the asymptotically Schwarzschild coordinates
	the wave coordinate expression satisfy
	$$
	\pa_{\widetilde{x}^\a} \left (\g^{\a\b} \sqrt{|\g|}\right)
	= - 2\frac {M\ln r}{r^2}
	\omega^j\de_{j \b} + O(\frac 1{r^2}).
	$$
\end{lemma}
\begin{proof}
We have
$$
\pa_{\widetilde{x}^\a} \big (\g^{\a\b} \sqrt{|\g|}\big)= (A^{-1})^{\a\mu}
\pa_{x^\mu} \big (A_{\a\nu} A_{\b\de} |A|^{-1} g^{\nu\de}
\sqrt{|g|}\big)= (A^{-1})^{\a\mu} \pa_{x^\mu} \big(A_{\a\nu}
A_{\b\de} |A|^{-1} \big)  g^{\nu\de} \sqrt{|g|},
$$
where we used that the wave coordinate condition is satisfied in
the $x$
coordinates. Here the expression
$
\pa_{x^\mu} \left (A_{\a\nu} A_{\b\de} |A|^{-1} \right)
$
is already at most of the order of $M{\ln r}/{r^2}$. Therefore,
ignoring the terms of the order of  $M^2{\ln^2 r}/{r^3}$ allows
us to replace the above expression by
$$
\pa_{x^\a} \left (m^{\nu\de} A_{\a\nu} A_{\b\de} |A|^{-1} \right).
$$
Replacing the matrix $A$ by its expansion
$$
A_{\a\b}=\de_{\a\b} +
\frac {M \ln r}{r} (\de^{ij}- \omega^i\omega^j)\de_{\a i} \de_{\b j}
+ O(\frac 1r), \qquad |A|=1+ 2
\frac {M \ln r}{r}+\frac{M}{r} + O(\frac{\ln^2 r}{r^2}),
$$
and ignoring the terms of the order of $M /{r^2}$ we obtain
$$
2\pa_{x^i}
\Big(  \frac {M \ln r}{r}
(\de^{ij}-\omega^i\omega^j)\de_{j \b} \Big)
- 2 m^{\nu\b} \pa_{\nu} \Big
(\frac {M \ln r}{r}\Big).
$$
The proof follows since again up to the terms of order $M/{r^2}$ we have
\begin{equation}
-2 \Big(  \frac {M \omega^i \ln r}{r^2}
(\de^{ij}-\omega^i\omega^j)\de_{j \b} \Big)
- 2\frac {M\ln r}{r^2} 2\omega^j \de_{j \b}
+ 2 \de_{j\b} \frac {M \omega^j \ln
	r}{r^2} = - 2\frac {M\ln r}{r^2} \omega^j \de_{j \b}.\tag*{\qedhere}
\end{equation}
\end{proof}
The same calculation also gives
\begin{lemma}\label{lem:gfun} We have
	\begin{align*}
		&\pa_{\widetilde{x}^\a} \left (\g^{\a\b} {|\g|}\right)=A_{\b\de} |A|^{-2}
		\pa_{x^\nu} \left (g^{\nu\de} |g|\right) + O(\frac 1{r^2})
		= \frac{\g^{\a\b}}{\sqrt {|g|}} \pa_{\widetilde{x}^\a}
		\left (\sqrt{|g|}\right) + O(\frac 1{r^2}),\\
		&\pa_{\widetilde{x}^\a} \left (\g^{\a\b} \sqrt{|g|}\right)
		= - 4\frac {M\ln r}{r^2}
		\omega^j\de_{j \b} + O(\frac 1{r^2}),\qquad
		\pa_{\widetilde{x}^\a} \left (\g^{\a\b} {|\g|}/\sqrt{|g|}\right)
		= O(\frac 1{r^2}).
	\end{align*}
\end{lemma}

\subsection{The Landau-Lifshitz pseudotensor}\label{sec:LL}
In view of Proposition \ref{prop:eikonalintro} the characteristic hypersurfaces of the metric $g$ become asymptotic to the null cones of
the Schwarzschild metric, we recast the Einstein equations in the
form explicitly involving the asymptotically Schwarzschild coordinates $(\widetilde{t}=t,\widetilde{x})$
as opposed to the original Minkowski $(t,x)$ harmonic coordinates.
\vskip 1pc
\!\!Let $S_{\widetilde{u}, \widetilde{r}}\!=\!\{(\widetilde{t},\widetilde{x}); \widetilde{t}\!=\!\widetilde{u}\!+ \!\widetilde{r}\}$ be a sphere, following \cite{T58, T62, T02, BC17} we define the {\it Trautman four-momentum}
\[
M^\alpha_T(\widetilde{u})=\lim_{r\to\infty}\frac{1}{4\pi}\int_{S_{\widetilde{u}, \widetilde{r}}}\!\tilde{\mathbb{U}}^{\alpha\beta\gamma}\,dS_{\beta\gamma}.
\]
Here $dS_{\beta\gamma}\!=\!n_{[\beta}k_{\gamma]}\widetilde{r}^2dS(\omega)$ with $n_\gamma\!=\!(d\widetilde{r})_\gamma\!=\!(0, \omega_i)$, $k_\beta\!=\!(d\widetilde{t})_\beta\!=\!(1, 0,0,0)$, and the superpotential $\tilde{\mathbb{U}}^{\alpha\beta\gamma}\!$ is
\[
\tilde{\mathbb{U}}^{\a\b\gamma}=\sqrt{\lvert\tilde{g}\rvert}\tilde{g}^{\a\mu}\tilde{\mathbb{U}}_\mu^{\b\gamma} \quad\text{where} \quad \tilde{\mathbb{U}}_\mu^{\beta\gamma}=\sqrt{\lvert \tilde{g}\rvert} \tilde{g}^{\alpha\mu}\tilde{g}^{\sigma[\rho}\delta_\mu^\gamma\tilde{g}^{\beta]\tau}\widetilde{\partial}_{\tau}\tilde{g}_{\rho\sigma}.
\]
Here the square brackets denote the antisymmetric part of a tensor, i.e., $T^{[a_1\cdots a_l]}=\sum_{\sigma}(-1)^\sigma T^{a_{\sigma(1)}\cdots a_{\sigma(l)}}$ where the sum is taken over all permutations $\sigma$ of $1,\ldots,l$ and $(-1)^\sigma$ is $1$ for even permutations and $-1$ for odd permutations. A direct computation implies
\[
\tilde{\mathbb{U}}^{\a\b\gamma}=-\tla^{\a\b\mu},\quad\text{where}\quad
\tilde{\lambda}^{\a\b\mu}=\widetilde{\pa}_\nu \big(|\g|(\g^{\a\b}
\g^{\mu\nu}-\g^{\a\mu}\g^{\b\nu})\big).
\]
Therefore we can write
\[
M^\alpha_T(\widetilde{u})\!\!=\!\!\frac{1}{4\pi}\int_{\mathbb{S}^2}\! m_T^\alpha(\widetilde{u},\omega){ dS(\omega)\!},
\]
where with $L_\alpha=(-1, \omega_i$ and $\Lb_\alpha=(-1, -\omega_i)$
\[
m_T^\alpha(\widetilde{u},\omega)
= \lim_{\widetilde{r}\to\infty}  (\widetilde{r})^2 \big(\widetilde{\lambda}^{\alpha\beta\gamma}
\uL_\gamma L_\beta\big)(\widetilde{u}-\widetilde{r},\widetilde{r}\omega).
\]
The {\it Trautman radiated four-momentum} is defined as
\[
E^\alpha_T(\widetilde{u})=\lim_{\widetilde{r}\to\infty}\frac{1}{2\pi}\int_{S_{\widetilde{u}, r}}\!\lvert\tilde{g}\rvert\tilde{\pi}^{\alpha\beta}\,dS_{\beta}.
\]
Here $dS_\beta=n_\beta\widetilde{r}^2dS(\omega)$ with $n_\beta=(0, \omega_i)$ and $\tilde{\pi}^{\alpha\beta}$ is {\it Landau-Lifshitz} pseudotensor \cite[\S101]{LL62}, which is a symmetric pseudotensor satisfying
\[
\tilde{\pi}^{\a\b}=-2\tilde{G}^{\a\b}+\frac{1}{\lvert\tilde{g}\rvert}\widetilde{\pa}_\mu \tla^{\a\b\mu}\quad\text{where}\quad \tilde{G}^{\a\b}=\tilde{R}^{\a\b}-\frac{1}{2}\tilde{g}^{\a\b}\tilde{R}.
\]
We write
\[
E_T^\alpha(\widetilde{u})
= \frac{1}{2\pi}\int_{\mathbb{S}^2} \Delta m^\a_T(\widetilde{u},\omega)\, dS(\omega)\quad\text{where}\quad\Delta m^\a_T(\widetilde{u},\omega)=
\lim_{\widetilde{r}\to\infty}  \widetilde{r}^2  |\g| \tpi^{\a i} \frac{\widetilde{x}_i}{\widetilde{r}}.
\]
It can be shown that the Einstein-vacuum equations
$R_{\a\b}(g)=0$ can be written in the form
$$
|\g| \tpi^{\a\b} = \frac{\pa \tla^{\a\b\mu}}{\pa \widetilde{x}^\mu},
$$
where
$$
\tla^{\a\b\mu}=\frac {\pa}{\pa \widetilde{x}^\nu} \left (|\g|(\g^{\a\b}
\g^{\mu\nu}-\g^{\a\mu}\g^{\b\nu})\right),\qquad
\tla^{\a\b\mu}=-\tla^{\a\mu\b},
$$
and $\tpi^{\a\b}$ is the Landau-Lifshitz pseudo tensor
\begin{multline*}
	\tpi^{\a\b}=(2\Gat^\gamma_{\mu\nu} \Gat^\delta_{\gamma\delta} -
	\Gat^\gamma_{\mu\delta}\Gat^\delta_{\nu\gamma}-\Gat^\gamma_{\mu\gamma}
	\Gat^\delta_{\nu\delta})(\g^{\a\nu} \g^{\b\mu} - \g^{\a\b}
	\g^{\mu\nu})
	+\g^{\a\gamma} \g^{\mu\nu} (\Gat^\b_{\gamma\delta}
	\Gat^\delta_{\mu\nu} + \Gat^\b_{\mu\nu} \Gat^\delta_{\gamma\delta} -
	\Gat^\b_{\mu\delta} \Gat^\delta_{\gamma\nu} -\Gat^\b_{\gamma\nu}
	\Gat^\delta_{\mu\delta})
	\\
	+\g^{\b\gamma} \g^{\mu\nu} (\Gat^\a_{\gamma\delta}
	\Gat^\delta_{\mu\nu}
	+ \Gat^\a_{\mu\nu} \Gat^\delta_{\gamma\delta} - \Gat^\a_{\mu\delta}
	\Gat^\delta_{\gamma\nu} -\Gat^\a_{\gamma\nu}
	\Gat^\delta_{\mu\delta}) +\g^{\mu\gamma}
	\g^{\nu\delta}(\Gat^\a_{\mu\nu}\Gat^\b_{\gamma\delta}
	-\Gat^\a_{\mu\gamma}
	\Gat^\b_{\nu\delta}),
\end{multline*}
where $\Gat^\a_{\mu\nu}$ are the Christoffel symbols of $\g$.
Alternatively, with $ \tG^{\a\b}=\sqrt{|\g|} \g^{\a\b}$,
\begin{multline*}
\!\!\!\!\!\!	|\g| \tpi^{\a\b}\!\!=\!\pa_{\widetilde{x}^\mu} \tG^{\a\b\!} \pa_{\widetilde{x}^\nu} \tG^{\mu\nu}\!-\pa_{\widetilde{x}^\mu}
	\tG^{\a\mu\!} \pa_{\widetilde{x}^\nu} \tG^{\b\nu}\!
	+\tfrac 12 \g^{\a\b} \g_{\mu\nu}\pa_{\widetilde{x}^\delta} \tG^{\gamma\mu}
	\pa_{\widetilde{x}^\gamma}\tG^{\delta\nu}\!
	-\big(\g^{\a\gamma\!} \g_{\mu\nu}
	\pa_{\widetilde{x}^\delta} \tG^{\b\mu} \pa_{\widetilde{x}^\gamma} \tG^{\delta\nu}\!
	+ \g^{\b\gamma\!}\g_{\mu\nu}
	\pa_{\widetilde{x}^\delta} \tG^{\a\mu} \pa_{\widetilde{x}^\gamma}\tG^{\delta\nu}\big)\\
	+ \g_{\mu\nu}\g^{\gamma\delta}
	\pa_{\widetilde{x}^\gamma} \tG^{\a\mu} \pa_{\widetilde{x}^\delta} \tG^{\b\nu}
	+\frac 18 (2\g^{\a\mu} \g^{\b\nu} - \g^{\a\b} \g^{\mu\nu})
	(2 \g_{\gamma\delta}
	\g_{\rho\sigma} - \g_{\gamma\rho} \g_{\delta\sigma}) \pa_{\widetilde{x}^\mu}
	\tG^{\gamma\rho} \pa_{\widetilde{x}^\nu}\tG^{\delta\sigma}.
\end{multline*}

The tensor $\tpi^{\a\b}=\tpi^{\b\a}$ is symmetric and due to the
anti-symmetry of $\tla^{\a\b\ga}$ is divergence free
$$
\pa_{\widetilde{x}^\b} \left (|\g| \tpi^{\a\b}\right)=0.
$$
Integrating the above identity in the region $\{(\widetilde{t},\widetilde{x}):\, \widetilde{q}_1\le
\widetilde{q}=\widetilde{r}-\widetilde{t}\le \widetilde{q}_2,\,\, |\widetilde{x}|\le R\}$ we obtain
\begin{equation}\label{eq:massid}
	\int_{|\widetilde{x}|\le R,\, \widetilde{q}=\widetilde{q}_1} |\g| \tpi^{\a\b} L_\b = \int_{|\widetilde{x}|\le
		R,\, \widetilde{q}=\widetilde{q}_2} |\g| \tpi^{\a\b} L_\b + \frac 1R\int_{|\widetilde{x}|=R, \,
		\widetilde{q}_1\le\,
		\widetilde{q}\le \widetilde{q}_2} |\g| \tpi^{\a i} {\widetilde{x}_i} .
\end{equation}
Using that $|\,\g| \tpi^{\a\b} = {\pa \tla^{\a\b\ga}}/{\pa
	\widetilde{x}^\ga}$ we have
$$
\int_{|\widetilde{x}|\le R,\, \widetilde{q}=\widetilde{q}_1} |\g| \tpi^{\a\b} L_\b = \int_{|\widetilde{x}|\le
	R,\, \widetilde{q}=\widetilde{q}_1} \pa_{\widetilde{x}^\ga} \tla^{\a\b\ga} L_\b .
$$
As usual we use the decomposition
$$
\pa_{\widetilde{x}^\ga} \tla^{\a\b\ga}=-\frac 12 L_\ga \pa_{\Lb}
\tla^{\a\b\ga}-\frac 12 \Lb_\ga \pa_{L}
\tla^{\a\b\ga} + A_\ga \pa_A \tla^{\a\b\ga}
+ B_\ga \pa_B \tla^{\a\b\ga} .
$$
where $(L, \Lb\,, A, B)$ is the null frame associated to the asymptotically Schwarzschild coordinates $(\widetilde{t}, \widetilde{x})$. By anti-symmetry of $\tla^{\a\b\ga}$
$$
L_\ga \pa_{\Lb} \tla^{\a\b\ga} L_\b=0.
$$
Therefore,
$$
\int_{|\widetilde{x}|\le R,\, \widetilde{q}=\widetilde{q}_1} \pa_{\widetilde{x}^\ga} \tla^{\a\b\ga} L_\b
=-\frac 12 \int_{|\widetilde{x}|\le R,\, \widetilde{q}=\widetilde{q}_1} \pa_{L}  \left
(\tla^{\a\b\ga}  \Lb_\ga L_\b\right) + \int_{|\widetilde{x}|\le R,\,
	\widetilde{q}=\widetilde{q}_1} A_\ga L_\b \pa_A \tla^{\a\b\ga}
+B_\ga L_\b \pa_B \tla^{\a\b\ga} .
$$
On the surface $\widetilde{q}=\widetilde{q}_1$ we introduce coordinates $(s,\omega)$
with $s=\frac 12(t+\widetilde{r}+\widetilde{q}_1)$ so that $\pa_L=\pa_s$ and the volume
form is $s^2 d\omega$. Then
\begin{align}
	\frac 12\int_{|\widetilde{x}|\le R,\, \widetilde{q}=\widetilde{q}_1} \pa_{L}
	\left (\tla^{\a\b\ga}  \Lb_\ga L_\b\right) &=
	\frac 12\int_0^R\int_{{\Bbb S}^2} \frac d{ds}
	\left (\tla^{\a\b\ga}  \Lb_\ga L_\b\right) s^2 ds
	d\omega\notag\\ &=
	\frac 12\int_{{\Bbb S}^2} \left (\tla^{\a\b\ga}
	\Lb_\ga L_\b\right)(R,\omega) R^2 d\omega
	- \int_0^R\int_{{\Bbb S}^2} \left (\tla^{\a\b\ga}  \Lb_\ga
	L_\b\right) s ds d\omega . \label{eq:Lder}
\end{align}
On the other hand, using that
$\partial_A L_\beta=A^k\partial_{\widetilde{x}^k} \widetilde{x}_\beta/\widetilde{r}=A^\beta/\widetilde{r}$
and that $\tla^{\alpha\beta\gamma}$ is anti symmetric,
\newcommand{\divv}{\mbox{${\text{div}} \mkern-13mu$/\,}}
\begin{equation}
	A_\ga L_\b \pa_A \tla^{\a\b\ga}= A_\ga \pa_A
	\big (L_\b  \tla^{\a\b\ga}\big)-
	\frac 1s A_\ga  A_\b  \tla^{\a\b\ga}= A_\ga \pa_A
	\big (L_\b  \tla^{\a\b\ga}\big) .
\end{equation}
We now need the following lemma
\begin{lemma}\label{lem:AB}
	Let  $A,B$ be orthonormal vector fields on ${\Bbb S}^2$,
	independent of $\widetilde{r}$; $\partial_{\widetilde{r}} A=\partial_{\widetilde{r}} B=0$.  Then
	\begin{equation}
		|\partial A|+|\partial B|\les {1}/{\widetilde{r}},
	\end{equation}
	and with $\partial_A=A^k\partial_{\widetilde{x}^k}=A^k\overline{\partial}_{\widetilde{x}^k} $ we have
	\begin{equation}
		\pa_A (A^\ell)= -{\omega^\ell}/{\widetilde{r}}
		+ \langle \pa_A A,B\rangle B^\ell,\qquad
		\pa_A (B^\ell)= \langle \pa_A B,A\rangle A^\ell  .
	\end{equation}
	Moreover, if $\overline{F}^k$ is tangential to ${\Bbb S}^2$ then
	$$
	\divv \,\overline{F} = \overline{\pa}_k \overline{F}^k
	=\partial_k \overline{F}^k=A^k\pa_A \overline{F}_k
	+B^k\pa_B \overline{F}_k
	$$
	satisfies
	\begin{equation}
		\int_{{\Bbb S}^2} \divv \, \overline{F}\,  d\omega=0.
	\end{equation}
	On the other hand if $F$ is not tangential then
	\begin{equation}
		\int_{{\Bbb S}^2} (A^\ell \partial_A +B^\ell \pa_B)
		F_\ell\,  d\omega=
		\frac{2}{\widetilde{r}} \int_{{\Bbb S}^2} \omega^\ell F_\ell \, d\omega .
	\end{equation}
\end{lemma}
Using that $\omega_\gamma=\frac 12 (L_\gamma-\Lb_\ga)$ and by
anti-symmetry of $\tla^{\alpha\beta\gamma}L_\beta L_\gamma =0$
we obtain
$$
\int_{{\Bbb S}^2} A_\ga L_\b \pa_A \tla^{\a\b\ga}
+ B_\ga L_\b \pa_B \tla^{\a\b\ga}\, d\omega
= -\int_{{\Bbb S}^2}
\frac 1s \Lb_\ga L_\b  \tla^{\a\b\ga}\, d\omega .
$$
Combining this with \eqref{eq:Lder} we finally obtain
$$
\int_{|\widetilde{x}|\le R,\, \widetilde{q}=\widetilde{q}_1} \pa_{\widetilde{x}^\ga} \tla^{\a\b\ga} L_\b=
-\frac 12\int_{{\Bbb S}^2} \left (\tla^{\a\b\ga}
\Lb_\ga L_\b\right)(R,\widetilde{q}_1,\omega) R^2 d\omega.
$$
Substituting this into \eqref{eq:massid} we obtain
$$
\int_{{\Bbb S}^2} \left (\tla^{\a\b\ga}  \Lb_\ga
L_\b\right)(R,\widetilde{q}_1,\omega) R^2 d\omega=\int_{{\Bbb S}^2} \left
(\tla^{\a\b\ga}  \Lb_\ga L_\b\right)(R,\widetilde{q}_2,\omega) R^2 d\omega
- 2\int_{|\widetilde{x}|=R,\,  \widetilde{q}_1\le \,\widetilde{q}\le \widetilde{q}_2}
|\g| \tpi^{\a i} \frac{\widetilde{x}_i}{R} .
$$
Assume for a moment that the following limits exist
$$
m^\a_T(\widetilde{q},\omega)= \lim_{R\to\infty}  R^2 \left (\tla^{\a\b\ga}
\Lb_\ga L_\b\right)(R,\widetilde{q},\omega), \qquad \Delta m^\a_T(\widetilde{q},\omega)=
\lim_{R\to\infty}  R^2  |\g| \tpi^{\a i} {\widetilde{x}_i}/{R},
$$
then we have the following analog of the
{\it Bondi mass loss formula}
\begin{equation}\label{eq:Bondi}
	M^\a_T(\widetilde{q}_1)=M^a_T(\widetilde{q}_2)
	- \int_{\widetilde{q}_1}^{\widetilde{q}_2}  E^\a_T(\widetilde{q}) d\widetilde{q} .
\end{equation}
In what follows we will establish existence of the above limits
together with non-positivity
of $\Delta m_T$.

\subsection{Existence of the Trautman mass}\label{sec:mass}
Here we are concerned with establishing existence of the limit
$$
m_T^\alpha(\widetilde{q},\omega)
= \lim_{\widetilde{r}\to\infty}  (\widetilde{r})^2 \left (\tla^{\a\b\ga}
\Lb_\ga L_\b\right)(\widetilde{r},\widetilde{q},\omega) $$ and thus the Trautman mass $M^0_T$.
\begin{proposition}\label{prop:trautmanmass}
	The Trautman four-momentum
	\[
	M^\alpha_T(\widetilde{u})\!\!=\!\!\frac{1}{4\pi}\int_{\mathbb{S}^2}\! m_T^\alpha(\widetilde{u},\omega){ dS(\omega)\!}
	\] is well defined.
\end{proposition}
\begin{proof}
We consider the quantity $\tla^{\a\b\ga}  \Lb_\ga L_\b$ and show
that for sufficiently large
$r>t/2$
$$
|\tla^{\a\b\ga}  \Lb_\ga L_\b|\le C r^{-2}.
$$
Moreover, from our discussion and analysis in \cite{Lind17} it will be
clear that the
quantities defining the $r^{-2}$ behavior $\tla$ all have well
defined limits as
$r\to \infty$.
To estimate $\tla^{\a\b\ga}  \Lb_\ga L_\b$ we first note that
by Lemma \ref{lem:gfun}
we have
\begin{equation*}
	\tla^{\a\b\ga}  \Lb_\ga L_\b =  \Lb_\ga L_\b
	\frac {\pa}{\pa \widetilde{x}^\mu} \left (|\g|(\g^{\a\b} \g^{\ga\mu}
	-\g^{\a\ga}\g^{\b\mu})\right)
	=\frac{|\g|}{\sqrt{|g|}}\,  \Lb_\ga L_\b
	\left (\g^{\ga\mu} \pa_{\widetilde{x}^\mu} (\sqrt {|g|}\g^{\a\b})-
	\g^{\b\mu}\pa_{\widetilde{x}^\mu}(\sqrt{|g|} \g^{\a\ga})\right)
	+ O(\frac 1{r^2}).
\end{equation*}
According to Lemma \ref{lem:newmetric}
\begin{equation}\label{eq:tilg}
	\g^{\a\b}=(1+\frac Mr) m^{\mu\nu} +\frac {2M (\ln r-1)}{r}
	(\de^{ij}-\omega^i\omega^j)\de_{\a i} \de_{\b j} + h_1^{\mu\nu}
	+ O(\frac
	{\ln^2 r}{r^2}).
\end{equation}
Using that
$$
(\de^{ij}- \omega^i\omega^j)\de_{\a i} \de_{\b j} L^\b=(\de^{ij}-
\omega^i\omega^j)\de_{\a i} \de_{\b j} \Lb^\b=0
$$
and the crude estimates $
|\pa g|+|\pa \g|\le \varepsilon r^{-1+\eps}
$, we obtain
\begin{align*}
	\tla^{\a\b\ga}  \Lb_\ga L_\b &= (1+\frac Mr)
	\frac{|\g|}{\sqrt{|g|}}\,
	\Lb_\ga L_\b  \left (m^{\ga\mu} \pa_{\widetilde{x}^\mu} (\sqrt {|g|}\g^{\a\b})-
	m^{\b\mu}\pa_{\widetilde{x}^\mu}(\sqrt{|g|} \g^{\a\ga})\right) \\ &+
	\frac{|\g|}{\sqrt{|g|}}\,  \Lb_\ga L_\b
	\left (h_1^{\ga\mu} \pa_{\widetilde{x}^\mu} (\sqrt {|g|}\g^{\a\b})-
	h_1^{\b\mu}\pa_{\widetilde{x}^\mu}(\sqrt{|g|} \g^{\a\ga})\right)
	+ O(\frac 1{r^2}).
\end{align*}
We first analyze the expression
$$
\Lb_\ga L_\b  \left (m^{\ga\mu} \pa_{\widetilde{x}^\mu} (\sqrt {|g|}\g^{\a\b})-
m^{\b\mu}\pa_{\widetilde{x}^\mu}(\sqrt{|g|} \g^{\a\ga})\right) =
L_\b\, \pa_{\Lb }  (\sqrt {|g|}\g^{\a\b})
- \Lb_\b \pa_{L}  (\sqrt {|g|}\g^{\a\b}).
$$
We write
$$
L_\b\, \pa_{\Lb }  (\sqrt {|g|}\g^{\a\b})=- \Lb_\b\, \pa_{L }
(\sqrt {|g|}\g^{\a\b}) +
2 C_\b \pa_C (\sqrt {|g|}\g^{\a\b}) - 2 \pa_{\widetilde{x}^\b} (\sqrt {|g|}\g^{\a\b}).
$$
Here and in what follows repeated index C is summed over $C=A,B$.
Therefore
\begin{equation*}
	\Lb_\ga L_\b  \left (m^{\ga\mu} \pa_{\widetilde{x}^\mu} (\sqrt {|g|}\g^{\a\b})-
	m^{\b\mu}\pa_{\widetilde{x}^\mu}(\sqrt{|g|} \g^{\a\ga})\right)
	=
	- 2\Lb_\b\, \pa_{L }  (\sqrt {|g|}\g^{\a\b}) +
	2 C_\b \pa_C (\sqrt {|g|}\g^{\a\b}) - 2 \pa_{\widetilde{x}^\b} (\sqrt {|g|}\g^{\a\b}).
\end{equation*}
We analyze the expression
\begin{align*}
	2 C_\b \,\pa_C (\sqrt {|g|}\g^{\a\b})
	&=  2 (1+\frac Mr) C_\b \pa_C (\sqrt {|g|}\, m^{\a\b})+
	4  \frac {M(\ln r-1)}{r} C_\b \pa_C (\sqrt {|g|}\,
	(\de_{ij}-\omega^i\omega^j)\de^{\a j} \de^{\b i})\\
	&\qquad+2 C_\b  \pa_C (\sqrt {|g|}\,  h_1^{\a\b}) + O(\frac 1{r^{3-\eps}})\\
&	=  2 (1+\frac Mr) C^\a \pa_C (\sqrt {|g|})
	+ 4  \frac {M(\ln r-1)}{r} \big( \pa_C (\sqrt {|g|} \, C^\a)
	-  \sqrt {|g|}\,(\de_{ij}-\omega^i\omega^j)\de^{\a j} \de^{\b i}
	\pa_C  C_\b \big)\\
&\qquad	+ 2 \pa_C (\sqrt {|g|}\, h_1^{\a\b} C_\b)
	-\,  2 \sqrt{|g|} \,h_1^{\a\b} \pa_{C} C_\b + O(\frac 1{r^{3-\eps}}).
\end{align*}
Here we used that $(\de_{ij}-\omega^i\omega^j)\de^{\a j} \de^{\b
	i}$ is the orthogonal projection on ${\Bbb S}^2$. In particular, $
(\de_{ij}-\omega^i\omega^j)\de^{\a j} \de^{\b i} C_\b
=C^\a
$
and by Lemma \ref{lem:AB}
\beq
\pa_A (A^k)= -{\omega^k}/{\widetilde{r}}+ \langle \pa_A A,B\rangle B^k.
\eq
Thus,
\begin{equation}
	\begin{split}
	&\pa_A (\sqrt {|g|} A^k) - \sqrt {|g|}\,
	(\de_{ij}-\omega^i\omega^j)\de^{k j} \de^{\b i}\pa_A  A_\b\\
	&\qquad=  A^k \pa_A (\sqrt {|g|})
	+\sqrt{|g|} \big( \pa_A A^k-(\de_{ij}-\omega^i\omega^j)
	\de^{k j} \de^{\b i}\pa_A  A_\b\big)=A^k \pa_A (\sqrt {|g|})
	- \frac {\omega^k}{\widetilde{r}} .
	\end{split}
\end{equation}
Furthermore by Proposition \ref{prop:sharpmetricdecay} we have the following estimates
\begin{equation*}
	|\, h_1^{\a\b} A_\b|\le  \frac{C\varepsilon}{\widetilde{r}},\qquad
	|\,\pa_A \sqrt{|g|}|\le \frac {C\varepsilon}{{\widetilde{r}}^2},
	\qquad |\,\pa_A (h_1^{\a\b}
	A_\b)|\le \frac {C\varepsilon}{{\widetilde{r}}^2} .
\end{equation*}
We remark that the last estimate above is very sensitive since
its not true for each term
in $(\pa_A h_1^{\a\b})A_\b+ h_1^{\a\b} \pa_A A_\b$.
As a result,
\begin{equation*}
2 \sum_{C=A,B}C_\b \pa_C (\sqrt {|g|}\g^{\a\b})
= -  8 \frac {M\ln r}{r\widetilde{r}} \omega^j \de^{j\a}+\frac 4{\widetilde{r}}
h_1^{\a\b} \omega^j \de_{j\b} +O(\frac 1{r^2})
=-  8 \frac {M\ln r}{r^2} \omega^j \de^{j\a}+\frac 4{r}
h_1^{\a\b} \omega^j \de_{j\b} +O(\frac 1{r^2}).
\end{equation*}
Using Lemma \ref{lem:gfun} we further compute
$$
-2\pa_{\widetilde{x}^\a} \left (\g^{\a\b} \sqrt{|g|}\right)
= 8\frac {M\ln r}{r^2}
\omega^j\de_{j \b} + O(\frac 1{r^2}).
$$
Finally, with the help of the estimate $
|\pa_L (\sqrt{|g|})|\le \frac{C\varepsilon}{r^{-2}}
$
we obtain
\begin{align*}
	- 2\Lb_\b\, \pa_{L }  (\sqrt {|g|}\g^{\a\b})
	&= - 2\Lb_\b\, \pa_{L }
	(\g^{\a\b}) +O(\frac 1{r^2}) \\
	&= - 2\Lb_\b\, \pa_{L }  ((1\!+\!\frac
	Mr) m^{\a\b}) - 2 \pa_{L }
	\Big(\frac {2M\!\ln r\!}{r}(\de^{ij}\!-\omega^i\omega^j)
	\de^{\a i}\de^{\b j} \Lb_\b\!\Big)\! - 2\,
	\pa_{L } (h_1^{\a\b}\! \Lb_\b)+O(\frac 1{r^2})\\
	&= - 2\, \pa_{L }
	(h_1^{\a\b} \Lb_\b)+O(\frac 1{r^2}).
\end{align*}
Gathering our estimates we obtain
\begin{align*}
	\Lb_\ga L_\b  \left (m^{\ga\mu} \pa_{\widetilde{x}^\mu} (\sqrt {|g|}\g^{\a\b})-
	m^{\b\mu}\pa_{\widetilde{x}^\mu}(\sqrt{|g|} \g^{\a\ga})\right) &= \frac 4{r}
	h_1^{\a\b} \omega^j \de_{j\b} - 2\, \pa_{L } (h_1^{\a\b} \Lb_\b)
	+O(\frac 1{r^2})\\ &=
	-\frac 2{r}
	h_1^{\a\b} \Lb_\b - 2\, \pa_{L } (h_1^{\a\b} \Lb_\b)
	+O(\frac 1{r^2}),
\end{align*}
where in the last line we used that
$
\omega^j \de_{j\b} =(L_\b-\Lb_\b)/2$ and $ |h_1^{\a\b} L_\b|\le C \varepsilon r^{-1}
$. We now note that
$$
\frac 2{r}
h_1^{\a\b} \Lb_\b +2\, \pa_{L } (h_1^{\a\b} \Lb_\b)
=\frac 2{r} \pa_L (r h_1^{\a\b}) \Lb_\b=O\big(\frac{1}{r^2}\big),
$$
by the results in Proposition \ref{prop:specialsharpmetricdecay}. Therefore,
$$
\Lb_\ga L_\b
\left (m^{\ga\mu} \pa_{\widetilde{x}^\mu} (\sqrt {|\,g|}\,\g^{\a\b})-
m^{\b\mu}\pa_{\widetilde{x}^\mu}(\sqrt{|\,g|}\, \g^{\a\ga})\right)
= O(\frac 1{r^2}) .
$$
To achieve the desired result for $\tla^{\a\b\ga} \Lb_\ga L_\b$
it remains to show that the expression
$$
\Lb_\ga L_\b  \left
(h_1^{\ga\mu} \pa_{\widetilde{x}^\mu} (\sqrt {|\,g|}\,\g^{\a\b})-
h_1^{\b\mu}\pa_{\widetilde{x}^\mu}(\sqrt{|\,g|}\, \g^{\a\ga})\right)
= O(\frac 1{r^2}).
$$
Given the fact that
$
|(h_1)_{LT}|\le C\varepsilon r^{-1-\sigma}
$, the term
$$
h_1^{\b\mu}\pa_{\widetilde{x}^\mu}(\sqrt{|\,g|} \,\g^{\a\ga} ) \Lb_\ga L_\b
= O(\frac
1{r^{2+\sigma}}).
$$
It is clear that we only need to analyze the term
$
h^1_{L\Lb} \pa_{\Lb} \left (\sqrt{|\,g|} \,\g^{\a\b}  L_\b\right)
$.
Since
$$
|h^1_{L\Lb}|\le \frac{C\varepsilon}r,\qquad |\pa \sqrt{|\,g|}|
\le  \frac{C\varepsilon}r,\qquad
|\pa \tilde g^{\a\b} L_\b|\le \frac{C\varepsilon}{r},
$$
we see that
\begin{equation*}
h^1_{L\Lb} \pa_{\Lb} \left (\sqrt{|\,g|}\, \g^{\a\b}
L_\b\right)=O(\frac 1{r^2}).\tag*{\qedhere}
\end{equation*}
\end{proof}

\subsection{Existence of the news function $\Delta m_T$}\label{subsec:4.4}
We now establish existence of the limit
$$
\Delta m^\a_T(\widetilde{q},\omega)
=2\lim_{\widetilde{r}\to\infty} \widetilde{r}^2 |\g|\,\tpi^{\a i}\omega^i.
$$
\begin{proposition}\label{thm:existenceofnews}
	$$
	\Delta m^\a_T(\widetilde{q},\omega)=\frac 14 L^\a \lim_{\widetilde{r}\to\infty} \widetilde{r}^2
	|\pa_{\widetilde{q}} \hat\ga|^2(\widetilde{r},\widetilde{q},\omega).
	$$
	Here $\hat\ga_{CD}=h^1_{CD} -\frac 12 \de_{CD} (h^1_{AA}+h^1_{BB})$
	is
	the traceless part of of the angular part of the metric $g$. Using
	Proposition \ref{prop:lim} we can identify $\Delta m^\a_T$ with the
	expression.
	$$
	\Delta m^\a_T(\widetilde{q},\omega)=\frac 14 L^\a\delta^{CC'}\delta^{DD'} \pa_{\widetilde{q}}
	H_{CD}^{1\infty}\pa_{\widetilde{q}}
	H_{C'D'}^{1\infty}=\frac 12 L^\a n(\widetilde{q}, \omega).
	$$
\end{proposition}
Combining Propositions \ref{prop:trautmanmass}, \ref{thm:existenceofnews} with the analog of the Bondi  mass
formula stated in \eqref{eq:Bondi} we obtain
\begin{theorem}\label{thm:masslosslawT}
	$$
	M^\a_T(\widetilde{q}_1)=M^\a_T(\widetilde{q}_2) - \frac {1}{8\pi}
	\int_{\widetilde{q}_1}^{\widetilde{q}_2}  \int_{{\Bbb S}^2}  L^\a n(\widetilde{q}, \omega)\,
	d\omega\, d\widetilde{q} .
	$$
\end{theorem}
\begin{remark}
	Note that since
	$$
	|\,\pa_{\widetilde{q}}  H_{AB}^{1\infty}|\le \frac {C\varepsilon}{1+|\,\widetilde{q}|},
	$$
	the News function is easily integrable with respect to the variable
	$\widetilde{q}$.
\end{remark}
\begin{proof}[Proof of Proposition \ref{thm:existenceofnews}]
	Recall that
\begin{multline*}
	|\,\g| \tpi^{\a\b}=\pa_{\widetilde{x}^\nu} \tG^{\a\b} \pa_{\widetilde{x}^\la} \tG^{\nu\la}-\pa_{\widetilde{x}^\nu}
	\tG^{\a\nu} \pa_{\widetilde{x}^\la} \tG^{\b\la} +\frac 12 \g^{\a\b} \g_{\nu\la}
	\pa_{\widetilde{x}^\de} \tG^{\nu\mu} \pa_{\widetilde{x}^\mu}\tG^{\de\la}
-\g^{\a\nu} \g_{\la\mu}
	\pa_{\widetilde{x}^\de} \tG^{\b\mu} \pa_{\widetilde{x}^\nu} \tG^{\la\de} \\ - \g^{\b\nu}
	\g_{\la\mu}\pa_{\widetilde{x}^\de} \tG^{\a\mu} \pa_{\widetilde{x}^\nu} \tG^{\la\de}
	+ \g_{\nu\la}
	\g^{\mu\de} \pa_{\widetilde{x}^\mu} \tG^{\a\nu} \pa_{\widetilde{x}^\de}\tG^{\b\la} +\frac 18 (2
	\g^{\a\nu} \g^{\b\la} - \g^{\a\b} \g^{\nu\la}) (2 \g_{\mu\de}
	\g_{\ga\rho} - \g_{\de\ga} \g_{\mu\rho}) \pa_{\widetilde{x}^\nu} \tG^{\mu\rho}
	\pa_{\widetilde{x}^\la} \tG^{\de\ga},
\end{multline*}
with $\tG^{\a\b}=\sqrt{|\g|} \g^{\a\b}$. Recall that by Lemma
\ref{lem:wcoordtilde}
\begin{equation}\label{eq:wg}
	\pa_{\widetilde{x}^\a} \tG^{\a\b}=O(\frac {\ln r}{r^2}).
\end{equation}
We also use the crude estimate
$$
|\pa_{\widetilde{x}} \tG|\le \frac {C\varepsilon}{r^{1-\eps} (1+|\widetilde{q}|)^{1+\eps}}.
$$
Based on this we can replace
\begin{multline*}
	|\,\g| \,\tpi^{\a\b}
	= \frac 12 m^{\a\b} m_{\nu\la} \pa_{\widetilde{x}^\de} \tG^{\nu\mu}
	\pa_{\widetilde{x}^\mu}\tG^{\de\la}-(m^{\a\nu} m_{\la\mu} \pa_{\widetilde{x}^\de} \tG^{\b\mu} \pa_{\widetilde{x}^\nu}
	\tG^{\la\de} + m^{\b\nu} m_{\la\mu}\pa_{\widetilde{x}^\de} \tG^{\a\mu} \pa_{\widetilde{x}^\nu}
	\tG^{\la\de})\\+ m_{\nu\la} m^{\mu\de} \pa_{\widetilde{x}^\mu} \tG^{\a\nu} \pa_{\widetilde{x}^\de}
	\tG^{\b\la}
	+\frac 18 (2 m^{\a\nu} m^{\b\la} - m^{\a\b}
	m^{\nu\la}) (2 m_{\mu\de} m_{\ga\rho} - m_{\de\ga} m_{\mu\rho})
	\pa_{\widetilde{x}^\nu} \tG^{\mu\rho} \pa_{\widetilde{x}^\la} \tG^{\de\ga} \! + O(\frac 1{r^{3-2\eps}
		(1\!+\!|\widetilde{q}|)}).
\end{multline*}
Taking into account that
\begin{equation}\label{eq:tang}
	|\pa_A \tG|+|\pa_L \tG|\le \frac C{r^{2-\eps}(1+|\,\widetilde{q}|)^\eps},
\end{equation}
and using that modulo tangential derivatives
$\partial_{\widetilde{x}^\alpha}$ is $L_\alpha\partial_{\widetilde{q}}$
we can further write
\begin{align*}
	|\,\g| \,\tpi^{\a\b}&= \frac{1}{2} m^{\a\b} m_{\nu\la} L_\de L_\mu \pa_{\widetilde{q}}
	\tG^{\nu\mu} \pa_{\widetilde{q}}\tG^{\de\la}-(m_{\la\mu} L^\a L_\de \pa_{\widetilde{q}}
	\tG^{\b\mu} \pa_{\widetilde{q}} \tG^{\la\de} +  m_{\la\mu} L^\b L_\de \pa_{\widetilde{q}}
	\tG^{\a\mu} \pa_{\widetilde{q}} \tG^{\la\de})
	\\
	&\qquad+\frac{1}{4} L^\alpha L^\beta
	(2 m_{\mu\de} m_{\ga\rho} - m_{\de\ga} m_{\mu\rho})
	\pa_{\widetilde{q}} \tG^{\mu\rho} \pa_{\widetilde{q}} \tG^{\de\ga} + O(\frac 1{r^{3-2\eps}
		(1+|\,\widetilde{q}|)}).
\end{align*}
Using the expression \eqref{eq:tilg} for the metric $g$
$$
\g^{\a\b}=(1+\frac Mr) m^{\mu\nu} +\frac {2M (\ln r-1)}{r} (\de^{ij}-
\omega^i\omega^j)\de^{\a i} \de^{\b j} + h_1^{\mu\nu}+ O(\frac
{\ln^2 r}{r^2}).
$$
We easily see that\footnote{The $\epsilon$ loss occurs only due to
	the presence of the logarithmic terms.}
$$
|\,\pa_{\widetilde{q}} \tG^{\la\mu} T_\mu|
\le \frac {C\varepsilon}{r(1+|\,\widetilde{q}|)^{1-\eps}}.
$$
Moreover, using the wave coordinate condition
$$
|\,\pa_{\widetilde{q}} \g_{LT}|
\le \frac {C\varepsilon}{r^{2-\eps} (1+|\,\widetilde{q}|)^\eps},
$$
we conclude that
$$
|\,\pa_{\widetilde{q}} \tG_{LT}|
\le \frac {C\varepsilon}{r^{2-\eps} (1+|\,\widetilde{q}|)^\eps}.
$$
In the above we used the fact that $\g_{LT}=O(1/r)$.
Using that $m_{\mu\nu}
=-(L_\mu\underline{L}_\nu+\underline{L}_\mu L_\nu)/2
+A_\mu A_\nu+B_\mu B_\nu$ this allows us to conclude that
\begin{align*}
	|\,\g| \,\tpi^{\a\b}
	&= -\frac{1}{2} \big(-L_{\mu} L^\a \pa_{\widetilde{q}} \tG^{\b\mu}
	\pa_{\widetilde{q}} \tG_{\Lb L} -  L_{\mu} L^\b L_\de \pa_{\widetilde{q}} \tG^{\a\mu}
	\pa_{\widetilde{q}} \tG_{\Lb L}\big)+
	\\ &\qquad+\frac 1{4} L^\a L^\b
	(2 m_{\mu\de} m_{\ga\rho} - m_{\de\ga} m_{\mu\rho}) \pa_{\widetilde{q}}
	\tG^{\mu\rho} \pa_{\widetilde{q}}
	\tG^{\de\ga} + O(\frac 1{r^{3-2\eps}(1+|\,\widetilde{q}|)})\\
	&=-\frac 12 L^\a L^\b (\pa_{\widetilde{q}} \tG_{L\Lb})^2+\frac 1{4} L^\a L^\b
	\left(2\pa_{\widetilde{q}} \tG_{\mu\nu} \pa_{\widetilde{q}} \tG^{\mu\nu} - (m^{\mu\nu}
	\pa_{\widetilde{q}}\tG_{\mu\nu})^2\right) + O(\frac 1{r^{3-2\eps}(1+|\,\widetilde{q}|)}).
\end{align*}
Here we used that we can expand any vector in the null frame
$V^\alpha
=-(L^\alpha F_{\underline{L}}+\underline{L}^\alpha F_{L})/2
+A^\alpha F_A+B^\alpha F_B$.
We can write
$$
-2\pa_{\widetilde{q}} \tG_{L\Lb}=-\pa_L \tG_{\Lb\Lb} + 2 \pa_C \tG_{\mu\nu} C^\mu
\Lb^\nu- 2m^{\mu\de}\pa_{\widetilde{x}^\de}\tG_{\mu\nu} \Lb^\nu.
$$
Using \eqref{eq:wg} and \eqref{eq:tang} we obtain that
$$
|\,\pa_{\widetilde{q}} \tG_{L\Lb}|\le \frac{C\varepsilon}{r^{2-\eps}}.
$$
Based on this we obtain
\begin{equation*}
	|\g|\, \tpi^{\a\b}=\frac 1{4} L^\a L^\b \left(2\pa_{\widetilde{q}} \tG_{CD}
	\pa_{\widetilde{q}} \tG^{CD}\! - (\de^{CD}\pa_{\widetilde{q}}\tG_{CD})^2\right) + O(\frac
	1{r^{3-2\eps}(1\!+\!|\widetilde{q}|)})= \frac 12  L^\a L^\b |\pa_{\widetilde{q}}
	\hat\G|^2 +O(\frac 1{r^{3-2\eps}(1\!+\!|\widetilde{q}|)}).
\end{equation*}
Here $\hat \tG_{CD}$ is a tensor
$$
\hat \tG_{CD}=\tG_{CD}-\frac 12\de_{CD} (\tG_{AA}+\tG_{BB}).
$$
Remembering that $\tG_{\a\b}=\sqrt{|\g|} \g_{\a\b}$ we obtain
$$
\pa_{\widetilde{q}}\hat \tG_{AB}=\pa_{\widetilde{q}}
(\sqrt{|\g|})\hat{\tilde\ga}_{AB}+\sqrt{|\g|}\,  \pa_{\widetilde{q}}(\hat
{\tilde\ga}).
$$
Here
$$
\hat{\tilde\ga}_{CD}=\g_{CD}-\frac 12 \de_{CD} (\g_{AA}+\g_{BB})
=\hat\ga
+O(\frac{\ln^2 r}{r^2}),\qquad \hat\ga_{CD}=h^1_{CD}
-\frac 12 \de_{CD}
(h^1_{AA}+h^1_{BB}).
$$
In the above we used \eqref{eq:tilg}. Therefore we obtain
\begin{equation}
|\,\g| \tpi^{\a\b}= \frac 12  L^\a L^\b |\pa_{\widetilde{q}} \hat\ga|^2
+O(\frac
1{r^{3-2\eps}(1+|\,\widetilde{q}|)}).\tag*{\qedhere}
\end{equation}
\end{proof}

\subsection{The ADM mass}
It remains to show that the Trautman mass as $\widetilde{q}\to \infty$ tends to the ADM mass
of initial data which in our setting is $M$. It follows from the results in the
previous section that all components of $h^1$ tend to $0$ faster than $r^{-1}$ as
$r\to\infty$ so the limit of the mass as $r\to\infty$ only depend on $h^0_{\mu\nu}=
M\delta_{\mu\nu}/r$. In fact a direct calculation implies the limit is exactly $M$.
That this constant is positive would follow from proving that the Bondi mass tend to
$0$ as $\widetilde{q}\to -\infty$. This in turn would follow from Proposition \ref{prop:totalmassloss}.

\section{The asymptotic Hawking mass}\label{sec:5}
In this section we will use the modified asymptotically Schwarzschild null coordinates $\widehat{y}^p=(v^*, u^*, \widehat{y}^3, \widehat{y}^4)$ as defined in subsection \ref{subsec:2.2}.
\subsection{The definition of the asymptotic Hawking mass and radiated energy}
We define the radius of a surface $S$ to be  $r(S)\!=\!\sqrt{\!\text{Area}(S)\!/4\pi}$.
Let $\Lh$ and $\Lbh$ be the outgoing respectively incoming
null normals to $S$ satisfying
$g(\Lh,\Lbh)\!\!=\!-2$.\!
$\Lh$ and $\Lbh$ are unique up to the transformation
$\Lh\!\!\to \!a \Lh$ and
$\Lbh\!\to  a^{\!-1\!}\Lbh$.
The null second fundamental form and the conjugate null second fundamental form
are defined to be the tensors
$$
\chi(X,Y)\!\!=\!g(\nabla_{\!\!X}\Lh,Y),\qquad\text{respectively}\qquad \underline{\chi}(X,Y)\!\!=\!g(\nabla_{\!\!X} \Lbh,Y),
$$
for any vectors $X,Y$ tangent to $S$ at a point,
where $\nabla_{\!\!X}$ is covariant differentiation.
Under the transformation above $\chi\!\!\to\! a\chi$ and
$\underline{\chi}\!\!\to\! a^{-1} \underline{\chi}$ so
the  Hawking mass of $S$,
$$
M_{\mathcal{H}}(S)\!=r(S)\big(1\!+\int_S \tr \chi \,\tr \underline{\chi} \,dS\!/16\pi
\big),
$$
is invariant. If $\tr \chi \,\tr \underline{\chi}\!<\!0$ we
can fix $\Lh$ and $\Lbh$ by
$\tr\chi\!+\tr \underline{\chi}\!=\!0$. Let $\hat{\chi}$ and $\hat{\underline{\chi}}$
be the traceless parts. The incoming respectively outgoing energy flux through $\!S_{\!}$ are
$$
E(S)\!=\!\int_{S}\! \hat{\chi}^2 dS\!/16\pi,\qquad\text{and}\qquad
\underline{E}(S)\!=\!\int_{S}\! \hat{\underline{\chi}}^2 \!dS_{\!}/16\pi\!.
$$

Owing to Proposition \ref{prop:eikonalintro} the outgoing characteristic surfaces of $g$ is asymptotic the null cones $u^*=t-r^*$ constant, we use the family of spheres $S_{u^*\!, r}=\!\{(t,x); t\!=\!u^*+r^*(r), \, |x|\!\!=\!r\}$ to define the asymptotic Hawking mass and the radiated energy at null infinity as follows
\begin{align*}
	M_{AH}(u^*)=\lim_{r\to \infty}M_{\mathcal{H}}(S_{u^*\!, r})\qquad \text{and}\qquad E_{AH}(u^*)=\lim_{r\to\infty}\underline{E}(	S_{u^*\!,r}),
\end{align*}
with $r(S)^2\slashed{g}$ converging to a round metric
where $\slashed{g}$ is the restriction of $g$ on the spheres $S_{u^*\!, r}$.

In order for the limit of the Hawking mass to exist as well as the energy to be well defined we must have that
$r(S)\,\tr \chi\!\sim 2$ and $r(S)\,\tr \uchi\!\sim -2$, as $r(S)\!\to\!\infty$.
In order for the mass to be defined we also require that the rescaled spheres $S^1\!$ (scaled by $r(S)$ so that the radius is $1$)
converge to a round sphere, i.e.
that the Gaussian curvature $r(S)^2 \!K\!\!\sim \!1$, as $r(S)\!\to\!\infty$.

\subsection{The radiated energy at null infinity}
Assume that $r(S)^2\slashed{g}$ converges to a round metric which we will prove in next subsection. We now prove the radiated energy at infinity is well defined.
\begin{proposition}\label{prop:5.1}
	The rdiated energy at null infinity
	is given by
	\begin{equation}\label{eq:radiatedtotalenergy}
		E_{AH}(u^*)=\!\frac{1}{8\pi}\int_{\mathbb{S}^2}\!n(-u^*,\omega){ dS(\omega)}.
	\end{equation}
\end{proposition}
\begin{proof}
Since $\Lbh\, g(X,Y)\!=\!g(\nabla_{\!\Lbh} X,\!Y)\!+\!g(X,\!\nabla_{\!\Lbh}Y )$ and
$\nabla_{\!\Lbh} X\!=\!\nabla_{\!X} \Lbh\! -\![X,\Lbh]$ and since $\chi$ is symmetric we have
\begin{align}\label{eq:chiasderivative}
	2{\chi}(X,Y)&=\Lh\,g(X,Y)+g(X,[Y,\Lh])+g([X,\Lh],Y),\\
	\label{eq:uchiasderivative}
	2\underline{\chi}(X,Y)&=\Lbh\,g(X,Y)+g(X,[Y,\Lbh])+g([X,\Lbh],Y).
\end{align}

Since $g=m\!+h^0\!+h^1$ this is true for the surface measures
\beq\label{eq:surfacedeterminant}
g_{AB}\!=d^{\,2} \delta_{AB}\!+h^1_{\!AB},\quad\text{and}\quad
\det{(g_{AB})}\!=d^{\,4} \!\! + d^{\,2}\delta^{AB}h^1_{\!AB}\!+O\big((h^1)^2\big),
\eq
where $d\!=\!(1\!+M\!/r)^{\!1/2}\!\!$ and $\{A,B\}$ are orthonormal vector fields on $\BS^2$ associated to the coordinates $x$, i.e.\ $A=A^k\pa_{x^k}$. Since $ r\,\delta^{AB} h^1_{\!AB}\to 0$ it
follows from our estimates that $\sqrt{\det{(g_{AB})}}\!=d^{\,2}\!\!+O(r^{-1-\gamma'})$.
Hence
$dS_{u^*\!, r}\!=\!\!\sqrt{\!\det{(g_{AB})}\!}\,\,r^2 dS(\omega)\!\sim r^2 dS(\omega)$ and
$r(S_{u^*\!, r})\!\sim\! r$.

Let us define $\Lh$ and $\Lbh$ by $\Lh{}^\alpha\!\!=g^{\alpha\beta}\Lh_\beta$ and
$\Lbh{}^\alpha\!\!=g^{\alpha\beta}\Lbh_\beta$, where
$\Lh_i\!=-\Lbh_i\!=a\omega_i$, for $i\!=\!1,2,3$, and $\Lh_0\!=a\tau$,
$\Lbh_0\!=a\utau\,$. The condition that they are outgoing respectively incoming null
normal is then equivalent to
$g^{00}\tau^2\!+2g^{0i}\tau \omega_i+g^{ij}\omega_i\omega_j=0$
respectively
$g^{00}\utau^2\!-2g^{0i}\utau \omega_i+g^{ij}\omega_i\omega_j=0$.
Completion of the squares
gives
$(\tau\!+\tau_{\!1})^2\!\!\!
=\!(\utau\!-\!\tau_{1})^2\!\!\!=\!\tau_0^2\!$, where
$\tau_0\!=\big(\tau_1^2\!-g^{ij}\omega_i\omega_j/g^{00})^{\!1/2}\!\!$
and $\tau_{1}\!=\!g^{0i}\omega_i/g^{00}\!$.
Hence $\tau\!+\utau\!=-2\tau_{0}$
and $\tau-\utau=-2\tau_{1}$.
If we set
$a\!=2b/(\tau\!+\utau)$ we get $g(\Lbh,\Lh)\!
=\!a^2\big( g^{00}\tau\utau\!-g^{0i}(\tau-\utau) \omega_i-g^{ij}\omega_i\omega_j\big)\!
=2b^2 g^{00}\!$. $\Lh+\Lbh$
is normal to the hyperplanes $t$ constant and $\Lh-\Lbh$
is the normal to the  spheres $S_{u^*\!, r}$ in these hyperplanes. If we set $a=(-g^{00}\tau_0^2)^{-1/2}$ we get $g(\hat{L}, \hat{\underline{L}})=-2$.

It follows that $\Lh=d\Ls\!+O(h^1)$ and  $\Lbh=d\Lbs\!+O(h^1)$, where
$\Ls\!=\pa_t\!+\pa_{r^*}$, $\Lbs\!=\pa_t-\pa_{r^*}$ and $d=(1\!+M\!/r)^{\!1/2}\!\!.$
Moreover, this is true also for the derivatives.
We therefore have $\chi\sim \!\chi^*$ and $\underline{\chi}\!\sim  \underline{\chi}^*\!$,
where ${\chi}^*(X,Y)\!=\! d g(\nabla_{\!\!X} \Ls,Y)$ and
$\underline{\chi}^*(X,Y)\!=\! d g(\nabla_{\!\!X} \Lbs,Y)$.
Since $[A,\Lbs]=-{r}^{-1}(dr/dr^*)A=-r^{-1}A+O(r^{-2})$, we obtain
\begin{align*}
	\underline{\chi}^*(A,B)\!=d \Lbs g(A,B)/2- d g(A,B)/r+O(r^{-2}),\quad\text{and}\quad
	{\chi}^*(A,B)\!=d \Ls g(A,B)/2+d g(A,B)/r+O(r^{-2}),
\end{align*}
by \eqref{eq:chiasderivative}-\eqref{eq:uchiasderivative}. With $\tr k\!=\sls{g\,}{}^{\!AB}k_{\!AB}$,
where $\sls{g\,}{}^{\!AB}$ is the inverse of $g_{\!AB}$,
we have
\beqs
\tr\underline{\chi}^*(A,B)=d \,\det{(g(A,B))}^{-1}\Lbs\,\det{(g(A,B))}/2-2d/r+O(r^{-2}).
\eqs
Here we used \eqref{eq:surfacedeterminant} and the identity
$Z \det{\!A}=\det{\!A}\, \tr (A^{-1}ZA)$.
Hence
\begin{align*}
	\tr\underline{\chi}^*(A,B)&=d^{\,-1}\Lbs\,\delta^{AB}h^1_{AB} /\,2-2d/r+O\big(h^1\pa h^1\big)+O(r^{-2}),\\
	\tr{\chi}^*(A,B)&=d^{\,-1}\Ls\,\delta^{AB}h^1_{AB} /\,2+2d/r+O\big(h^1\pa h^1\big)+O(r^{-2}).
\end{align*}
Hence with $\hat{\uchi}^*(A,B)=\uchi^*(A,B)-\tr \uchi^*\, g(A,B)/2$ we have
\begin{align*}
	\hat{\uchi}^*(A,B)=\Lbs \hat{h}^1(A,B)/2+O\big(h^1 \pa h^1\big)+O(r^{-2}),\quad\text{and}\quad
	\hat{\chi}^*(A,B)=\Ls \hat{h}^1(A,B)/2+O\big(h^1 \pa h^1\big)+O(r^{-2}).
\end{align*}
where $\hat{h}^1_{AB}\!={h}^1_{AB}-\delta_{AB}\,\delta^{CD}h^1_{\,CD}/2$.
Since we shown that
$r\,\delta^{CD}\pa_{q^*} h^1_{CD}\!\to \!0$ and that $r \hat{h}^1\!$ has a limit as
$r\!\to \!\infty$ along
the curves $(u^*+r^*(r),r\omega)$ in Remark \ref{rem:limit} it follows that
\beq
r^2|\hat{\uchi}^*|^2\!\sim r^2 (\pa_{q^*} \hat{h}^1)_{AB}\, (\pa_{q^*} \hat{h}^1)^{AB}
\!\!\to 2 n(-u^*\!,\omega).\tag*{\qedhere}
\eq
\end{proof}

\subsection{The convergence to a round metric sphere}
\begin{proposition}
	$r(S)^2\slashed{g}$ converging to a round metric
	where $\slashed{g}$ is the restriction of $g$ on the spheres $S_{u^*\!, r}$
\end{proposition}
\begin{proof}
By  Gauss equation
(\cite{C09}) $\!K\!+\tr \chi\tr\uchi/2 -({\chi},{\uchi})\!/2
\!=\sls{g\,}{}^{\!AC}\!\sls{g\,}{}^{\!BD}\!R_{ABCD}\!$, where $\sls{g\,}{}^{\!AC}\!$ is the inverse
of the restriction of the metric to the sphere $S$.
That $r(S)^2 K\!\sim \! 1$ follows if we show that
$r(S)^2\sls{g\,}{}^{\!AC}\!\sls{g\,}{}^{\!BD}\!R_{ABCD}\!\!
\sim 0$.

We now calculate the curvature components in modified asymptotically Schwarzschild null coordinates:
\[
\widehat{R}(g)_{abc}^{\quad p}=\frac{\pa\widehat{\Gamma}(g)_{~ac}^p}{\pa\widehat{x}^b}-\frac{\pa\widehat{\Gamma}(g)_{~bc}^p}{\pa\widehat{x}^a}+\widehat{\Gamma}(g)_{~ac}^q\widehat{\Gamma}(g)_{~bq}^p-\widehat{\Gamma}(g)_{~bc}^q\widehat{\Gamma}(g)_{~aq}^p,
\]
and
\[
\widehat{R}(g)_{abcd}=\widehat{g}_{dp}\frac{\pa\widehat{\Gamma}(g)_{~ac}^p}{\pa\widehat{x}^b}-\widehat{g}_{dp}\frac{\pa\widehat{\Gamma}(g)_{~bc}^p}{\pa\widehat{x}^a}+\widehat{\Gamma}(g)_{~ac}^q\widehat{\Gamma}(g)_{dbq}-\widehat{\Gamma}(g)_{~bc}^q\widehat{\Gamma}(g)_{daq}.
\]
Since $g\!=\!g^0\!+\!h^1\!$ with $g^0_{\alpha\beta}\!=m_{\alpha\beta}\!+\!M\delta_{\alpha\beta}/r$ and the inverse metric $g^{-1}\!=(g^0)^{-1}\!+h_1\!+O(M^2/r^2)$, we  calculate
\begin{gather*}
	-\widehat{\Gamma}(g^0)^{v^*}_{~ab},\, \widehat{\Gamma}(g^0)^{u^*}_{~ab}=f(r)\widehat{q}_{ab},\qquad\!\!-\widehat{\Gamma}(g^0)_{v^*ab},\, \widehat{\Gamma}(g^0)_{u^*ab}=g(r)\widehat{q}_{ab}/2\!\quad\!\!\text{where}~f(r), g(r)=r+O_1(1)\\
	\widehat{\Gamma}(g^0)_{abc}=(r^2+Mr)\widehat{\Gamma}(m)_{abc},
\qquad\widehat{\Gamma}(g^0)^c_{~ab}=\widehat{\Gamma}(m)^c_{~ab},\qquad \widehat{R}_{abcd}(g^0)=O(r).
\end{gather*}
Next we compute
\begin{gather*}
	\widehat{\Gamma}(g)_{abv^*}=\widehat{\Gamma}(g^0)_{abv^*}+\widehat{\Gamma}(h^1)_{abv^*}=\widehat{\Gamma}(g^0)_{abv^*}+\frac{1}{2}(\frac{\pa\widehat{h}^1_{ab}}{\pa v^*}+\frac{\pa\widehat{h}^1_{av^*}}{\pa \widehat{x}^b}-\frac{\pa\widehat{h}^1_{bv^*}}{\pa \widehat{x}^a})=\widehat{\Gamma}(g^0)_{abv^*}+O_1(1),\\
	\widehat{\Gamma}(g)_{abu^*}=\widehat{\Gamma}(g^0)_{abu^*}+\pa_{u^*}\widehat{h}^1_{ab}/2+O_1(1),\qquad \widehat{\Gamma}(g)_{abc}=\widehat{\Gamma}(g^0)_{abc}+O_1(r).
\end{gather*}
and
\begin{gather*}
	\widehat{\Gamma}(g)_{~ab}^{v^*}=\widehat{g}^{v^*p}\widehat{\Gamma}(g)_{pab}=\widehat{\Gamma}(g^0)^{v^*}_{~ab}+\pa_{u^*}\widehat{h}^1_{ab}+O_1(r^{1-\sigma}),\\
	\widehat{\Gamma}(g)_{~ab}^{u^*}=\widehat{g}^{u^*p}\widehat{\Gamma}(g)_{pab}=\widehat{\Gamma}(g^0)^{u^*}_{~ab}+O_1(1),\qquad
	\widehat{\Gamma}(g)_{~ab}^{c}=\widehat{g}^{cp}\widehat{\Gamma}(g^0)_{pab}=\widehat{\Gamma}(g^0)^{c}_{~ab}+O_1(r^{-1}).
\end{gather*}
Therefore
\begin{gather*} \pa_{\widehat{x}^d}\bigl(\widehat{\Gamma}(g)_{~ab}^{v^*}\bigr)=\pa_{\widehat{x}^d}\bigl(\widehat{\Gamma}(g^0)_{~ab}^{v^*}\bigr)
+\pa_{\widehat{x}^d}\pa_{u^*}\widehat{h}^1_{ab}/2+O_1(r^{1-\sigma}),\\ \pa_{\widehat{x}^d}\bigl(\widehat{\Gamma}(g)_{~ab}^{u^*}\bigr)=\pa_{\widehat{x}^d}\bigl(\widehat{\Gamma}(g^0)_{~ab}^{u^*}\bigr)+O_1(1),\qquad	\pa_{\widehat{x}^d}\bigl(\widehat{\Gamma}(g)_{~ab}^{c}\bigr)=\pa_{\widehat{x}^d}\bigl(\widehat{\Gamma}(g^0)_{~ab}^{c}\bigr)+O_1(r^{-1}).
\end{gather*}
Finally we conclude that
\begin{align*}
	\widehat{R}(g)_{3434}&=\widehat{R}(g^0)_{3434}+\bigl(r\widehat{q}_{44}\pa_{u^*}\widehat{h}^1_{33}+r\widehat{q}_{33}\pa_{u^*}\widehat{h}^1_{44}\bigr)/2-\bigl(r\widehat{q}_{34}\pa_{u^*}\widehat{h}^1_{34}+r\widehat{q}_{34}\pa_{u^*}\widehat{h}^1_{34}\bigr)/2+O(r^{2-\sigma})\\
	&=r^3\det(\widehat{q}_{ab})\pa_{u^*}(\slashed{\tr}h^1)/2+O(r^{2-\sigma})=O(r^{2-\sigma}).
\end{align*}
where we used $\widehat{q}^{\,33}\!=\det(\widehat{q}_{ab})^{-1}\widehat{q}_{44}$, $ \widehat{q}^{\,44}\!=\det(\widehat{q}_{ab})^{-1}\widehat{q}_{33}$, $\widehat{q}^{\,34}\!=-\det(\widehat{q}_{ab})^{-1}\widehat{q}_{34}$ and $r^{-2}\widehat{q}^{\,ab}\widehat{h}^1_{ab}=\slashed{\tr}h^1$ in the second step and Remark \ref{rem:sharpmetricdecay} in the last step. Due to the symmetry of the Riemann curvature tensor we see that $\widehat{R}(g)_{abcd}=O(r^{2-\sigma})$. Since $\slashed{g}^{ab}$ is the inverse of $\widehat{g}_{ab}\!=(1\!+\!M/r)r^2\widehat{q}_{ab}\!+\!\widehat{h}^1_{ab}$ we have $\slashed{g}^{ab}\!=r^{-2}\widehat{q}^{ab}+O(r^{-3})$ and thus $\slashed{g}^{ac}\slashed{g}^{bd}\widehat{R}(g)_{abcd}=O(r^{-2-\sigma})$. This proves that the rescaled spheres converge to a round sphere.
\end{proof}
\subsection{The limit of the Hawking mass along the asymptotically null hypersurfaces}\label{subsec:5.4}
We now establish the existence of the asymptotic Hawking mass
\[
	M_{AH}(u^*)=\lim_{r\to \infty}M_{\mathcal{H}}(S_{u^*\!, r}).
\]
\begin{proposition}\label{thm:limitofHawkingmass}
	\begin{equation}\label{eq:limitofHawkingmass}
		M_{AH}(u^*)=\lim_{r\to\infty} M_{\mathcal{H}}(S_{u^*\!, r})=M-\frac{1}{8\pi}\int_{-u^*}^{\infty}\int_{\mathbb{S}^2}n(\eta, \omega)\,dS(\omega)d\eta.
	\end{equation}
	\end{proposition}
In view of \eqref{eq:radiatedtotalenergy} and the above proposition, we establish the following {\it Bondi mass loss formula}
\begin{theorem}\label{thm:Bondiloss}
\begin{equation}\label{eq:Bondiloss}
	\frac{d}{du^*}M_{AH}(u^*)=-E_{AH}(u^*).
\end{equation}
Moreover we see that $M(u^*)\!=\!M$ as $u^*=t-r^*\!\to\!-\infty$ and by Proposition \ref{prop:totalmassloss} $M(u^*)\!=\!0$ as $u^*\!=\!t-r^*\!\to\!\infty$.
\end{theorem}

\begin{proof}[Proof of Proposition \ref{thm:limitofHawkingmass}]
	In order to obtain the explicit expression for $M_{AH}(u^*)$, we need to refine the expressions for $\tr\chi$, $\tr\underline{\chi}$. Recall that the outgoing  and incoming null normals $\Lh, \Lbh$ to the surfaces $S_{u^*\!,r}$
are expressed in terms of $\tau, \underline{\tau}, a$ as defined in the proof of proposition \ref{prop:5.1}. 
A direct computation implies
\begin{align*}
	\tau&=-1+\frac{M}{r}-\frac{1}{2}(h_1^{00}+h_1^{ij}\omega_i\omega_j)+h_1^{0i}\omega_i+O_1(\frac{\varepsilon}{r^{1+\sigma}})=-1+\frac Mr+O_1(\frac{\varepsilon}{r^{1+\sigma}}),\\
	\underline{\tau}&=-1+\frac{M}{r}-\frac{1}{2}(h_1^{00}+h_1^{ij}\omega_i\omega_j)-h_1^{0i}\omega_i+O_1(\frac{\varepsilon}{r^{1+\sigma}})=-1+\frac Mr-\frac{\widehat{h}_1^{v^*v^*}}{2}+O_1(\frac{\varepsilon}{r^{1+\sigma}}),\\
	a&=1+\frac{M}{2r}-\frac{1}{2}h_1^{ij}\omega_i\omega_j+O_1(\frac{\varepsilon}{r^{1+\sigma}})=1+\frac {M}{2r}+\frac{1}{8}(2\widehat{h}_1^{v^*u^*}-\widehat{h}_1^{v^*v^*})+O_1(\frac{\varepsilon}{r^{1+\sigma}}).
\end{align*}
Suppose $\{A, B\}$ are orthonormal vector fields on $\BS^2$ associated to the coordinates $x$, i.e.\ $A=A^k\pa_{x^k}$,  we have
\begin{align*}
	\hat{L}&=\Big(1+\frac {M}{2r}-\frac{\widehat{h}_1^{v^*v^*}}{8}-\frac{\widehat{h}_1^{v^*u^*}}{4}
+O_1(\frac{\varepsilon}{r^{1+\sigma}})\Big)\Ls+O_1(\frac{\varepsilon}{r^{1+\sigma}})
\Lbs+O_1(\frac{\varepsilon}{r^{1+\sigma}})\pa_A,\\ \hat{\underline{L}}&=\Big(-\frac{\widehat{h}_1^{v^*v^*}}{4}\!
+\!O_1(\frac{\varepsilon}{r^{1+\sigma}})\Big)\Ls+\Big(1+\frac {M}{2r}+\frac{\widehat{h}_1^{v^*v^*}}{8}-\frac{\widehat{h}_1^{u^*v^*}}{4}
+O_1(\frac{\varepsilon}{r^{1+\sigma}})\Big)\Lbs
+\Big(-\widehat{h}_1^{v^*A}\!+\!O_1(\frac{\varepsilon}{r^{1+\sigma}})\Big)\pa_A.
\end{align*}
Recall that
\begin{equation*}
	2{\chi}(A, B)=\Lh\,g(A, B)+g(A, [ B,\Lh])+g([A,\Lh], B),\quad\text{and}\quad
	2\underline{\chi}(A, B)=\Lbh\,g(A, B)+g(A,[B,\Lbh])+g([A,\Lbh], B).
\end{equation*}
With $\tr k\!=\sls{g\,}{}^{\!AB}k_{\!AB}$,
where $\sls{g\,}{}^{\!AB}$ is the inverse of $g_{\!AB}=g(A, B)$, we have
\begin{align*}
	\tr{\chi}&= \,\det{(g(A,B))}^{-1}\hat{L}\,\det{(g(A,B))}/2+\sls{g\,}{}^{\!AB}\bigl(g(A, [ B,\Lh])+g([A,\Lh], B)\bigr)/2,\\
	\tr{\underline{\chi}}&= \,\det{(g(A,B))}^{-1}\hat{\underline{L}}\,\det{(g(A,B))}/2+\sls{g\,}{}^{\!AB}\bigl(g(A, [ B,\hat{\underline{L}}])+g([A,\hat{\underline{L}}], B)\bigr)/2.
\end{align*}
Here we used the identity
$Z \det{\!A}=\det{\!A}\, \tr (A^{-1}ZA)$.
A direct calculation yields
\begin{align*}
	\Ls\det(g_{AB})&=-2Mr^{-2}+O(\varepsilon r^{-2-\sigma}),\quad\pa_A\det(g_{AB})=O(\varepsilon r^{-2-\sigma}),\\ \Lbs\det(g_{AB})&=2Mr^{-2}+\Lbs\slashed{\tr}h^1+\Lbs\det(h^1_{AB})+O(\varepsilon r^{-2-\sigma}),
\end{align*}
Therefore
\[
\hat{L}\det(g_{AB})=-2Mr^{-2}+O(\varepsilon r^{-2-\sigma}),\qquad\text{and}\qquad\hat{\underline{L}}\det(g_{AB})=2Mr^{-2}+\Lbs\slashed{\tr}h^1+\Lbs\det\!h^1+O(\varepsilon r^{-2-\sigma}).
\]
and
\begin{align*}
	\tr{\chi}&= -Mr^{-2}+\sls{g\,}{}^{\!AB}\bigl(g(A, [ B,\Lh])+g([A,\Lh], B)\bigr)/2+O(\varepsilon r^{-2-\sigma}),\\
	\tr{\underline{\chi}}&=Mr^{-2}+\Lbs\slashed{\tr}h^1/2+\Lbs\det(h^1_{AB})/2+\sls{g\,}{}^{\!AB}\bigl(g(A, [ B,\hat{\underline{L}}])+g([A,\hat{\underline{L}}]. B)\bigr)/2+O(\varepsilon r^{-2-\sigma}).
\end{align*}
where we used $\det(g_{AB})^{-1}\!=\!1\!+\!O(\varepsilon r^{-1})$.  It remains to control the commutators terms.
We compute
\begin{align*}
	g(A, [B,\Ls])&=g(A, \frac{dr}{dr^*}\frac{\pa_B}{r})=\Bigl(\frac{1}{r}-\frac{M}{r^2}+O(\frac{M^2}{r^3})\Bigr)g_{AB},\\
	g(A, [B,\Lbs])&=g(A, -\frac{dr}{dr^*}\frac{\pa_B}{r})=-\Bigl(\frac{1}{r}-\frac{M}{r^2}+O(\frac{M^2}{r^3})\Bigr)g_{AB},\\
	g(A, [B,C])&=g(A, (\slashed{\Gamma}_{BC}^D-\slashed{\Gamma}_{CB}^D)D)=(\slashed{\Gamma}_{BC}^D-\slashed{\Gamma}_{CB}^D)g_{AD}.
\end{align*}
Here $\slashed{\Gamma}_{AB}^C=m(\slashed{\nabla}_{\pa_A}\pa_B, D)$ where $\slashed{\nabla}$ is the covariant differentiation on sphere and then $\slashed{\Gamma}$ are the associated frame-Christoffel symbol and homogeneous functions of degree $-1$ with respect to the radial variable $r$. Then
\begin{align*}
	\frac{1}{2}\sls{g\,}{}^{\!AB}\bigl(g(A, [ B,\Lh])+g([A,\Lh], B)\bigr)&=\frac{2}{r}-\frac{M}{r^2}-\frac{\widehat{h}_1^{v^*v^*}}{4r}
-\frac{\widehat{h}_1^{v^*u^*}}{2r}+O(\frac{\varepsilon}{r^{2+\sigma}}),\\
	\frac{1}{2}\sls{g\,}{}^{\!AB}\bigl(g(A, [ B,\Lh])+g([A,\Lh], B)\bigr)&=-\frac{2}{r}+\frac{M}{r^2}-\frac{3\widehat{h}_1^{v^*v^*}}{4r}
+\frac{\widehat{h}_1^{v^*u^*}}{2r}-\pa_C\widehat{h}_1^{v^*C}-\widehat{h}_1^{v^*C}(\slashed{\Gamma}_{DC}^D
-\slashed{\Gamma}_{CD}^D)
+O(\frac{\varepsilon}{r^{2+\sigma}})\\ &=-\frac{2}{r}+\frac{M}{r^2}-\frac{3\widehat{h}_1^{v^*v^*}}{4r}
+\frac{\widehat{h}_1^{v^*u^*}}{2r}\!-\!\slashed{\nabla}_C\widehat{h}_1^{v^*C}\!+\!O(\frac{\varepsilon}{r^{2+\sigma}}).
\end{align*}
Here we used the fact that $\slashed{\Gamma}_{CD}^D=m(\slashed{\nabla}_{\pa_C}\pa_D, D)=\pa_C(m(D, D))/2=0$. Hence
\begin{align*}
	\tr\chi&=\frac{2}{r}-\frac{2M}{r^2}-\frac{\widehat{h}_1^{v^*v^*}}{4r}-\frac{\widehat{h}_1^{v^*u^*}}{2r}+O(\frac{\varepsilon}{r^{2+\sigma}}),\\
	\tr\underline{\chi}&=-\frac{2}{r}+\frac{2M}{r^2}-\frac{3\widehat{h}_1^{v^*v^*}}{4r}+\frac{\widehat{h}_1^{v^*u^*}}{2r}\!-\!\slashed{\nabla}_C\widehat{h}_1^{v^*C}+\frac{\Lbs\slashed{\tr} h^1}{2}+\frac{\Lbs\det(h^1_{AB})}{2}+O(\frac{\varepsilon}{r^{2+\sigma}}).
\end{align*}
Using \eqref{eq:wccLb} we obtain
\[
\tr\underline{\chi}=-\frac{2}{r}+\frac{2M}{r^2}+\frac{\pa_{\Ls}\widehat{h}_1^{v^*v^*}}{2}+\frac{\widehat{h}_1^{v^*v^*}}{4r}-\frac{\widehat{h}_1^{v^*u^*}}{2r}+O(\frac{\varepsilon}{r^{2+\sigma}}).
\]
Here we used these facts $\slashed{\nabla}_C\widehat{h}_1^{v^*C}=\widehat{\nabla}_c\widehat{h}_1^{v^*c}$, $\pa_{\Lbs}(\widehat{h}_1^{ac}r^2\widehat{q}_{cb}\widehat{h}_1^{bd}r^2\widehat{q}_{da})=\pa_{\Lbs}(h^1_{AB}h^{1AB})+O_2(\varepsilon r^{-2-\sigma})$ and $h^1_{AB}h^{1AB}+2\det\!h^1=(\slashed{\tr}h^1)^2$. Hence
\[
\tr\chi\tr\underline{\chi}=-\frac{4}{r^2}+\frac{8M}{r^3}+\frac{\pa_{\Ls}\widehat{h}_1^{v^*v^*}}{r}+\frac{\widehat{h}_1^{v^*v^*}}{r^2}+O(\frac{\varepsilon}{r^{3+\sigma}}).
\]
Since $dS_{u^*\!, r}=\sqrt{\det(g_{AB})}r^2dS(\omega)=(1+M/r+O(\varepsilon r^{-1-\sigma}))r^2dS(\omega)$ and $r(S_{u^*\!, r})=r+O(\varepsilon)$, It follows that
\begin{multline*}
	M_{\mathcal{H}}(S_{u^*\!, r})=r(S_{u^*\!, r})\big(1\!+\int_{S_{u^*\!, r}} \tr \chi \,\tr \underline{\chi} \,dS_{u^*\!, r}\!/16\pi \big)\!\\=\bigl(r+O(\varepsilon)\bigr)\Bigl(\frac{M}{r}+\frac{1}{16\pi}\int_{S_{u^*\!, r}} \!\!\!\!\! r\pa_{\Ls}\widehat{h}_1^{v^*v^*}\!\!+\widehat{h}_1^{v^*v^*}\!\!dS(\omega)
+O(\frac{\varepsilon}{r^{1+\sigma}})\Bigr)=M-\frac{1}{8\pi}\int_{-u^*}^{\infty}n(\eta, \omega)\,d\eta+O(\frac{\varepsilon}{r^{\sigma}}),
\end{multline*}
where we used the asymptotics result for the metric component $\widehat{h}^{v^*v^*}_1$ in Remark \ref{rem:limit}. Therefore
\[
M_{AH}(u^*)=\lim_{r\to\infty} M_{\mathcal{H}}(S_{u^*\!, r})=M-\frac{1}{8\pi}\int_{-u^*}^{\infty}\int_{\mathbb{S}^2}n(\eta, \omega)\,dS(\omega)d\eta.
\tag*{\qedhere}
\]
\end{proof}
\begin{remark}
According to the proof of Proposition \ref{thm:limitofHawkingmass}, we find that the existence of the limit $M_{AH}(u^*)$ of the Hawking mass  along the asymptotically null hypersurfaces does not require the null infinity to extend all the way back to the spatial infinity. Suppose the null infinity could be extended back to the spatial infinity, the past limit $\lim_{u^*\to-\infty} M_{AH}(u^*)$ equals to the ADM mass.
\end{remark}

\section{Bondi-Sachs coordinates}
In this section, our goal is to construct the Bondi-Sachs coordinates $\overline{y}^p=(u, \overline{r}, \overline{y}^3, \overline{y}^4)$ under which we denote the metric by $\overline{g}$. The Bondi-Sachs coordinates  $\overline{y}^p =(u,\overline{r},\overline{y}^3, \overline{y}^4)$ are based on a family of outgoing
null hypersurfaces $\overline{y}^1=u=const$. The two angular coordinates $\overline{y}^a$, $(a,b,c,...=3,4)$, are constant along the null rays, i.e. $g^{\alpha\beta}\pa_\beta u\pa_\alpha \overline{y}^a\!\!=\!0$.  The coordinate $\overline{y}^2 =\overline{r}$, which varies along the null rays,
is chosen to be an areal coordinate such that
$\det [\overline{g}_{ab}] = \overline{r}^4 \mathfrak{q}$, where $\mathfrak{q}(\overline{y}^a)$ is the determinant of the unit sphere metric $\overline{q}_{ab}$
associated with the angular coordinates $\overline{y}^a$. In these coordinates, the metric takes the Bondi-Sachs form
\[
\overline{g}_{pq}d\overline{y}^p d\overline{y}^q=-\frac{V}{\overline{r}}e^{2\beta} du^2-2 e^{2\beta}dud\overline{r} +\overline{r}^2h_{ab}\Big(d\overline{y}^a-U^adu\Big)\Big(d\overline{y}^b-U^bdu\Big).
\]
We have already constructed $u$ coordinate in section \ref{sec:3}, it remains to construct the angular coordinates $\overline{y}^a$ and areal coordinate $\overline{r}$.

\subsection{Construction of angular coordinates}\label{subsec:6.1}
Since $\{u=const\}$ are null hypersurfaces, for any point $P$, it must be at some null geodesic. After reparametrization, we see that $P$ must be at some $X(s)=(s, u^*(s), \widehat{y}^3(s), \widehat{y}^4(s))=X(s; u, \overline{y}^3, \overline{y}^4)$ where we use the notation $X(s; u, \overline{y}^3, \overline{y}^4)$ to emphasize that  the integral curve $X(s)$ of the vector field $\widehat{g}^{pq}\pa_{\widehat{y}^q}u\pa_{\widehat{y}^p}/\widehat{g}^{v^*q}\pa_{\widehat{y}^q}u$ satisfies that $(u^*(s), \widehat{y}^3(s), \widehat{y}^4(s))\to(u, \overline{y}^3, \overline{y}^4)$ as $s\to\infty$. Therefore, for any point $P\in X(s; u, \overline{y}^3, \overline{y}^4)$, we define $(\overline{y}^3, \overline{y}^4)$ to be the angular coordinates in Bondi-Sachs coordinates.
Using the estimate
\[
\frac{\widehat{g}^{aq}\pa_{\widehat{y}^q}u}{\widehat{g}^{v^*q}\pa_{\widehat{y}^q}u}=\frac{\widehat{h}^{u^*a}+O(\varepsilon r^{-3-\sigma})}{-1+O(\varepsilon r^{-1})}=-\widehat{h}^{u^*a}+O(\frac{\varepsilon}{r^{3+\sigma}}),
\]
and integrating along the curve $X(s; u, \overline{y}^3, \overline{y}^4)$ yield
\[
\mathring{y}^a:=\overline{y}^a-\widehat{y}^a=\int_{v^*}^\infty\!\frac{\widehat{g}^{aq}\pa_{\widehat{y}^q}u}{\widehat{g}^{v^*q}\pa_{\widehat{y}^q}u}\,ds=O(\frac{\varepsilon}{r^{1+\sigma}}).
\]
We now have new angular coordinates $\overline{y}^a(v^*, u^*, \widehat{y}^3, \widehat{y}^4)=\widehat{y}^a+\mathring{y}^a(v^*, u^*, \widehat{y}^3, \widehat{y}^4)$ which we use in Bondi-Sachs coordinate system  and then we derive the system for the derivatives of $\overline{y}^a$. According to the construction,
\begin{equation}\label{eq:angulareqn}
	g^{\alpha\beta}\pa_\alpha u\pa_\beta\overline{y}^a=0.
\end{equation}

Differentiating \eqref{eq:angulareqn} gives
\beq\label{eq:diffangulareqn}
g^{\alpha\beta} \partial_\alpha u\, \partial_\beta Z\overline{y}^a
= -g_Z^{\alpha\beta}\partial_\alpha u\, \partial_\beta \overline{y}^a-g^{\alpha\beta} \partial_\alpha Zu\, \partial_\beta \overline{y}^a,
\eq
where the Lie derivative $g_Z^{\alpha\beta}=\mathcal{L}_Z g^{\alpha\beta}$ is given by
\beq\label{eq:liederivative2}
g_Z^{\alpha\beta}\partial_\alpha u\,\partial_\beta w=
(Zg^{\alpha\beta}) \partial_\alpha u\,\partial_\beta w
+g^{\alpha\beta} \partial_\alpha u\, [Z,\partial_\beta] w+
g^{\alpha\beta} [Z,\partial_\alpha] u\,\partial_\beta w.
\eq
Hence with the notation $g_Z(U,V)=g_Z^{\alpha\beta}U_\alpha V_\beta$ and using the facts $\pa Z \us=0,~g_{0Z}(U,V)\!=\!0$ for $Z\in{\mathcal Z}=\{ \Omega_{ij},\partial_t\}$,  \eqref{eq:angulareqn} respectively \eqref{eq:diffangulareqn} become
\begin{align}
	\partial_{\widetilde{L}}  \overline{y}^a&=0,\label{eq:angularsystemone}\\
	\partial_{\widetilde{L}} Z \mathring{y}^a&=-h_{1Z}(\partial u,\partial \overline{y}^a)-g^{\alpha\beta} \partial_\alpha Z\mathring{u}\, \partial_\beta \overline{y}^a-\partial_{\widetilde{L}} Z \widehat{y}^a.\label{eq:angularsystemtwo}
\end{align}

In order to estimate this system we need the following lemmas.
\begin{lemma}
	If $|\Omega \mathring{y}^a|\leq1/20$ we have
	\begin{equation}\label{eq:vderiofangular}
		|\pa_{v^*}\mathring{y}^a|\leq\frac{C_0\varepsilon}{r^{1+\sigma}}
		|\pa_{t}\mathring{y}^a|+\frac{C_0\varepsilon}{r^{2+\sigma}}.	\end{equation}
\end{lemma}

\begin{proof}
	Using \eqref{eq:estimateofu} and \eqref{eq:highbadderi} we know that
	\begin{equation}\label{eq:formalderiofu}
		\pa_{\widehat{y}^p}u=\Bigl(O(\frac{\varepsilon}{r^{1+\sigma}}),~ 1+\frac{\widehat{h}_1^{v^*u^*}}{2}+O(\frac{\varepsilon}{r^{1+\sigma}}), ~O(\frac{\varepsilon}{r^{\sigma}}),~ O(\frac{\varepsilon}{r^{\sigma}})\Bigr),
	\end{equation}
	and
	\begin{equation}\label{eq:formalderiofangular}
		\pa_{\widehat{y}^p}\overline{y}^a=\Bigl(\pa_{v^*}\mathring{y}^a,~\pa_{u^*}\mathring{y}^a, ~\delta_b^a+\pa_{\widehat{y}^b}\mathring{y}^a\Bigr).
	\end{equation}
	Since $g^{\alpha\beta}\pa_\alpha u\pa_\beta\overline{y}^a\!=\!\widehat{g}^{pq}\pa_{\widehat{y}^p}u\pa_{\widehat{y}^q}\overline{y}^a=0$, with the estimates of $\widehat{g}$ and the assumption $|\Omega{\mathring{y}^a}|\leq1/20$ we obtain
	\[
	\bigl(-2+O(\frac{\varepsilon}{r})\bigr)\pa_{v^*}\mathring{y}^a+O(\frac{\varepsilon}{r^{1+\sigma}})\pa_{t}\mathring{y}^a+O(\frac{\varepsilon}{r^{2+\sigma}})=0.
	\]
	This finishes the proof of the lemma.
\end{proof}

\begin{lemma}\label{lem:honet}
	If $|\Omega{\mathring{y}^a}|\leq1/20$ we have
	\begin{equation}\label{eq:tderiofangular}
		|\pa_{\widetilde{L}}\pa_t\mathring{y}^a|\leq \frac{C_0\varepsilon}{r}||\pa_{v^*}\mathring{y}^a|
+\frac{C_0\varepsilon}{r^{1+\sigma}}|\pa_t\mathring{y}^a|
+\frac{C_0\varepsilon}{r^{2+\sigma}}.
	\end{equation}
\end{lemma}
\begin{proof}
	If $Z=\pa_t$, we have $\partial_{\widetilde{L}} \pa_t \widehat{y}^a=0$ and by \eqref{eq:formalderiofu} and \eqref{eq:formalderiofangular}
	\begin{equation*}
		h_{1Z}(\pa u, \pa \overline{y}^a)=(\pa_th_1^{\alpha\beta})\pa_\alpha u\pa_\beta\overline{y}^a=(\pa_t\widehat{h}_1^{pq})\pa_{\widehat{y}^p} u\pa_{\widehat{y}^q}\overline{y}^a =O(\frac{\varepsilon}{r})\pa_{v^*}\mathring{y}^a+O(\frac{\varepsilon}{r^{1+\sigma}})\pa_{t}\mathring{y}^a+O(\frac{\varepsilon}{r^{2+\sigma}}).
	\end{equation*}
	By Proposition \ref{prop:eikonalintro} and \eqref{eq:highbadderi}
	\[ \pa_{\widehat{y}^p}\pa_t\mathring{u}=\Bigl(O(\frac{\varepsilon}{r^{1+\sigma}}),~ \frac{\pa_t(\widehat{h}_1^{v^*u^*})}{2}+O(\frac{\varepsilon}{r^{1+\sigma}}), ~O(\frac{\varepsilon}{r^{\sigma}}),~ O(\frac{\varepsilon}{r^{\sigma}})\Bigr),
	\]
	Hence by \eqref{eq:formalderiofangular}
	\begin{align*}
		g^{\alpha\beta} \partial_\alpha \pa_t\mathring{u}\, \partial_\beta \overline{y}^a=\widehat{g}^{pq}\pa_{\widehat{y}^p}\pa_t\mathring{u}\pa_{\widehat{y}^q}\overline{y}^a=O(\frac{\varepsilon}{r})\pa_{v^*}\mathring{y}^a+O(\frac{\varepsilon}{r^{1+\sigma}})\pa_{t}\mathring{y}^a+O(\frac{\varepsilon}{r^{2+\sigma}}).
	\end{align*}
	Then this lemma follows from \eqref{eq:angularsystemtwo} with $Z=\pa_t$ and the assumption $|\Omega{\mathring{y}^a}|\leq1/20$.
\end{proof}

\begin{lemma}\label{lem:honeomega}
	If $|\Omega{\mathring{y}^a}|\leq1/20$ we have
	\begin{equation}\label{eq:aderiofangular}
		|\pa_{\widetilde{L}}\Omega \mathring{y}^a|\leq \frac{C_0\varepsilon}{r}||\pa_{v^*}\mathring{y}^a|+\frac{C_0\varepsilon}{r^{1+\sigma}}|\pa_t\mathring{y}^a|+\frac{C_0\varepsilon|\Omega \mathring{y}^a|}{r^2}+\frac{C_0\varepsilon}{r^{2+\sigma}}.
	\end{equation}
\end{lemma}
\begin{proof}
	The estimates for$(\Omega h_1^{\alpha\beta})\pa_\alpha u\pa_\beta\overline{y}^a$ and $g^{\alpha\beta} \partial_\alpha \Omega\mathring{u}\, \partial_\beta \overline{y}^a$ are similar to those in the proof of Lemma \ref{lem:honet}. Therefore we are left with $h_1^{\alpha\beta} \partial_\alpha u\, [Z,\partial_\beta] \overline{y}^a$, $h_1^{\alpha\beta} [Z,\partial_\alpha] u\,\partial_\beta \overline{y}^a$ and $\partial_{\widetilde{L}} \Omega \widehat{y}^a$. By Lemma \ref{lem:quadraticterms}, if $\Omega\!=\!x^i\pa_j-x^j\pa_i$
	\[
	h_1^{\alpha\beta} [\pa_\beta, \Omega] u
	=k^{\alpha \Omega/r}\pa_r u
	+\big(k^{\alpha i}
	\overline{\pa}_{j}-k^{\alpha j} \overline{\pa}_i\big)u,
	\]
	with
	$h_1^{\alpha \Omega/r}\!\!=k^{\alpha i} \omega_{j}-k^{\alpha j}
	\omega_i$. By \eqref{eq:formalderiofu} and  \eqref{eq:formalderiofangular} we conclude that
	\[
	|h_1^{\alpha\beta} \partial_\alpha u\, [Z,\partial_\beta] \overline{y}^a|+|h_1^{\alpha\beta} [Z,\partial_\alpha] u\,\partial_\beta \overline{y}^a|=O(\frac{\varepsilon}{r})\pa_{v^*}\mathring{y}^a+O(\frac{\varepsilon}{r^{1+\sigma}})\pa_{t}\mathring{y}^a+O(\frac{\varepsilon}{r^2})|\Omega \mathring{y}^a|.
	\]
	In view of \eqref{eq:tildaL}, we have $\partial_{\widetilde{L}} \Omega \widehat{y}^a=O(\varepsilon r^{-2-\sigma})\pa_{\widehat{y}^c}\Omega\widehat{y}^a=O(\varepsilon r^{-2-\sigma})$.
\end{proof}

\begin{proposition}\label{prop:deriofangular}
	If $\varepsilon>0$ is sufficiently small we have for $r>t/2$ and $r>2$ with constants $C_1=2C_0C_\sigma$ for some universal constant $C_\sigma$
	\begin{align}
		|\pa_{v^*}\mathring{y}^a|&\leq \frac{2C_0\varepsilon}{r^{2+\sigma}},\label{eq:vderi}\\
		|\pa_{t}\mathring{y}^a|&\leq\frac{C_1\varepsilon}{r^{1+\sigma}},\label{eq:tderi}\\
		|\Omega \mathring{y}^a|&\leq\frac{C_1\varepsilon}{r^{1+\sigma}}.\label{eq:aderi}
	\end{align}
\end{proposition}
\begin{proof}
	We can prove this by assuming these estimates are true  and show that they imply better estimates if $\varepsilon$ is sufficiently small. First from \eqref{eq:vderiofangular} and the assumed bound \eqref{eq:tderi}, we prove \eqref{eq:vderi} with $2C_0$ replaced by $3C_0/2$
	if $\varepsilon$ is sufficiently small such that $2C_0^2C_\sigma\varepsilon\!\leq\! 1/2$. From the construction of $\overline{y}^a$ we see that $\pa_t \mathring{y}^a, \Omega \mathring{y}^a\!\to\!0$ as $v^*\!\to\!\infty$. If we integrate \eqref{eq:angularsystemtwo} with $Z\!=\!\pa_t$ and use the assumed bound \eqref{eq:vderi} and \eqref{eq:tderi} we obtain
	\[
	|\pa_t\mathring{y}^a|\leq \frac{\varepsilon C_\sigma(2\varepsilon C_0^2C_0+\varepsilon C_0C_1+C_0)}{r^{1+\sigma}}\leq\frac{3C_1\varepsilon/4}{r^{1+\sigma}},
	\]
	if $\varepsilon$ is sufficiently small such that $2C_0^2\varepsilon+\varepsilon C_1\leq1/2$ which proves \eqref{eq:tderi}. If we integrate \eqref{eq:angularsystemtwo} with $Z=\Omega$ and use the assumed bound \eqref{eq:vderi}, \eqref{eq:tderi} and \eqref{eq:aderi} we obtain
	\[
	|\Omega \mathring{y}^a|\leq \frac{\varepsilon C_\sigma(2\varepsilon C_0^2C_0+2\varepsilon C_0C_1+C_0)}{r^{1+\sigma}}\leq\frac{3C_1\varepsilon/4}{r^{1+\sigma}},
	\]
	if $\varepsilon$ is sufficiently small such that $2C_0^2\varepsilon+2\varepsilon C_1\leq1/2$ which proves \eqref{eq:aderi}.
\end{proof}
We now turn to higher order derivatives of $\mathring{y}^a$.
\begin{proposition}\label{prop:highderiofangular}
	We have
	\begin{equation}
		{\sum}_{|I|\leq2}|Z^{*I}\mathring{y}^a|=O(\frac{\varepsilon}{r^{1+\sigma}}).
	\end{equation}
\end{proposition}
\begin{proof}
Following the proof of Proposition \ref{prop:eikonalintro} we commutate the vector fields $X\in\mathcal{X}=\{S^*=t\pa_t+x^{*i}\pa_{x^{*i}}, ~\Omega_{ij}, ~\pa_t\}$ through the equation \eqref{eq:angularsystemone}. Let
$\widetilde{X}\!\!=\!X\!-\!\delta_{X \!\Ss}\!$ and
$\widetilde{\mathcal{L}}_{\!X}\!\!=\!\mathcal{L}_{\!X}\!\!+\!2\delta_{X\!\Ss}$,
where $\delta_{X\! \Ss}\!\!=\!1$ if $X\!\!=\Ss\!\!$, and $=\!0$ otherwise.
Then $X \big(k(\pa u,\pa v\big)\!=\!( \widetilde{\mathcal{L}}_X k) (\pa u,\pa v)
\!+k(\pa\widetilde{X} u,\pa v)\!+k(\pa u,\pa\widetilde{X}v)\!$ and $\pa \widetilde{X}  \us\!\!=\!0$.
Since $g(\pa u, \pa\overline{y}^a)=0$ we get $\pa_{\widetilde{L}}\widetilde{X}\widetilde{Z}\overline{y}^a=-H(g, u,\overline{y}^a)$ where
\begin{align*}
	H(g, u,\overline{y}^a)&=\widetilde{\mathcal{L}}_X\widetilde{\mathcal{L}}_Zg(\pa u, \pa\overline{y}^a)+\widetilde{\mathcal{L}}_Xg(\pa u, \pa \widetilde{Z}\overline{y}^a)+\widetilde{\mathcal{L}}_Xg(\pa \widetilde{Z}u, \pa \overline{y}^a)+\widetilde{\mathcal{L}}_Zg(\pa u, \pa \widetilde{X}\overline{y}^a)\\
	&\quad+\widetilde{\mathcal{L}}_Zg(\pa \widetilde{X}u, \pa \overline{y}^a)+g(\pa \widetilde{X}\widetilde{Z}u, \pa\overline{y}^a)+g(\pa \widetilde{Z}u, \pa \widetilde{X}\overline{y}^a)+g(\pa \widetilde{X}u, \pa \widetilde{Z}\overline{y}^a).
\end{align*}
Here $\widetilde{\mathcal{L}}_X g
\!=\!\widetilde{\mathcal{L}}_X g_0\!+\widetilde{\mathcal{L}}_X h_1$, where
$\widetilde{\mathcal{L}}_{\pa_t} g_0\!=\!\widetilde{\mathcal{L}}_\Omega g_0\!=\!0$ and
$\widetilde{\mathcal{L}}_{\Ss} g_0\!=\!\kappa_3 g_0\!-2(\kappa_1\!-\kappa_2)\overline{g}_0$.
Here $\kappa_1\!\sim \! M\! \ln{r}\!/r$, $\kappa_2\!\sim \! \kappa_3\!\sim \!M\!/r$ and $\overline{g\overline{}^{}}_0(\pa u,\pa v)\!={g}_0^{ij}\pas_{\!i } u\,\pas_{\!\!j}v$.
Since $g_0(\pa \widetilde{X}^Iu, \pa\overline{y}^a)=O(\varepsilon r^{-2-\sigma})$ and $\overline{g}_0(\pa \widetilde{X}^Iu, \pa\overline{y}^a)=O(\varepsilon r^{-2-\sigma})$ for $|I|\leq 1$, we have $\widetilde{\mathcal{L}}_Xg_0(\pa \widetilde{X}^Iu, \pa\overline{y}^a)=O(\varepsilon r^{-2-\sigma})$ for $|I|\leq1$. Moreover, $\widetilde{\mathcal{L}}_{X} \overline{g}_0(\pa u\!,\pa \overline{y}^a)
=X(\overline{g}_0(\pa u\!, \pa\overline{y}^a))-\overline{g}_0(\pa \widetilde{X} u\!,\pa \overline{y}^a)
-\overline{g}_0(\pa u\!,\pa \widetilde{X} \overline{y}^a)$. It follows that
$|\widetilde{\mathcal{L}}_X^I g_0(\pa u\!,\pa \overline{y}^a)|=O(\varepsilon r^{-2-\sigma})|\Omega\widetilde{X}\overline{y}^a|+O(\varepsilon r^{-2-\sigma})$, for $|I|\!\leq \!2$. Hence
\begin{equation*}
	|\widetilde{\mathcal{L}}_X\widetilde{\mathcal{L}}_Zg(\pa u, \pa\overline{y}^a)+\widetilde{\mathcal{L}}_Xg(\pa \widetilde{Z}u, \pa \overline{y}^a)+\widetilde{\mathcal{L}}_Zg(\pa \widetilde{X}u, \pa \overline{y}^a)+g(\pa \widetilde{X}\widetilde{Z}u, \pa\overline{y}^a)|
	\!=O(\frac{\varepsilon} {r^{2+\sigma}})\Bigl(\!|\Omega\widetilde{X}\overline{y}^a\!|
+|\Omega\widetilde{Z}\overline{y}^a\!|\!\Bigr)+O(\frac{\varepsilon}{r^{2+\sigma}}).
\end{equation*}
Then it remains to control the terms containing second order derivatives of $\overline{y}^a$
\begin{multline*}
	|\widetilde{\mathcal{L}}_Xg(\pa u, \pa \widetilde{Z}\overline{y}^a)+\widetilde{\mathcal{L}}_Zg(\pa u, \pa \widetilde{X}\overline{y}^a)+g(\pa \widetilde{Z}u, \pa \widetilde{X}\overline{y}^a)+g(\pa \widetilde{X}u, \pa \widetilde{Z}\overline{y}^a)|\\
	=O(\frac{\varepsilon}{r})\Bigl(|\pa_{\Ls} \widetilde{X}\overline{y}^a|+|\pa_{\Ls}\widetilde{Z}\overline{y}^a|\Bigr)+O(\frac{\varepsilon}{r^{1+\sigma}})\Bigl(|\pa_{t} \widetilde{X}\overline{y}^a|+|\pa_{t}\widetilde{Z}\overline{y}^a|\Bigr)
	+O(\frac{\varepsilon}{r^{2+\sigma}})\Bigl(|\Omega \widetilde{X}\overline{y}^a|+|\Omega \widetilde{Z}\overline{y}^a|\Bigr)+O(\frac{\varepsilon}{r^{2+\sigma}}).
\end{multline*}
We have $\pa_{\widetilde{L}}\widetilde{X}\overline{y}^a=-\widetilde{\mathcal{L}}_Xg(\pa u, \pa\overline{y}^a)-g(\pa\widetilde{X}u, \pa\overline{y}^a)$, so $|\pa_{\widetilde{L}}\widetilde{X}\overline{y}^a|=O(\varepsilon r^{-2-\sigma})$.
By \eqref{eq:tildaL}
\[
|\pa_{\Ls}\widetilde{X}\overline{y}^a|\les|\pa_{\widetilde{L}}\widetilde{X}\overline{y}^a|+O(\frac{\varepsilon}{r^{1+\sigma}})|\pa_{t}\widetilde{X}\overline{y}^a|
+O(\frac{\varepsilon}{r^{2+\sigma}})|\Omega \widetilde{X}\overline{y}^a|.
\]
Therefore we conclude
\[
|\pa_{\widetilde{L}}\widetilde{X}\widetilde{Z}\mathring{y}^a|=O(\frac{\varepsilon}{r^{1+\sigma}})\Bigl(|\pa_{t}\widetilde{X}\tilde{y}^a|+|\pa_{t}\widetilde{Z}\tilde{y}^a|\Bigr)+O(\frac{\varepsilon}{r^{2+\sigma}})\Bigl(|\Omega \widetilde{X}\tilde{y}^a|+|\Omega \widetilde{Z}\tilde{y}^a|\Bigr)+O(\frac{\varepsilon}{r^{2+\sigma}}).
\]
where we used the facts that $\Omega\widetilde{X}\widehat{y}^a=O(1)$, $\pa_t\widetilde{X}\widehat{y}^a=0$ and $\pa_{\widetilde{L}}\widetilde{X}\widetilde{Z}\widehat{y}^a=O(\varepsilon r^{-2-\sigma})$. We repeat the proof of Proposition \ref{prop:deriofangular} and the conclusion follows.
\end{proof}
We now refine $\pa_t Z^I\mathring{y}^a$ with $\lvert I\rvert\leq1$.
\begin{proposition}\label{prop:badderihigh}
	If $Z\in\{\pa_t, \Omega_{ij}\}$ and $|I|\leq 1$	we have
	\begin{equation}\label{eq:badderihigh}
		\pa_tZ^I\mathring{y}^a=\frac12 Z^I(\widehat{h}_1^{v^*a})+\frac{r^*}{2}Z^I(\widehat{\slashed{\nabla}}_c\widehat{h}_1^{ac})+O(\frac{\varepsilon}{r^{2+\sigma}}).
	\end{equation}
\end{proposition}
\begin{proof}
We analyze \eqref{eq:angularsystemtwo} with $Z=\pa_t$ to find that
\[
\pa_{\widetilde{L}}\pa_t \mathring{y}^a=2\pa_{v^*}\pa_t\mathring{u}\pa_{t}\mathring{y}^a-\pa_t(\widehat{h}_1^{u^*a})-\frac{1}{r^2}\widehat{q}^{\,ab}\pa_{\widehat{y}^b}\pa_t\mathring{u}+O(\frac{\varepsilon}{r^{3+\sigma}}).
\]
Using \eqref{eq:wcc:Y} and \eqref{eq:highbadderi} we obtain
\[
\pa_{\widetilde{L}}\pa_t \mathring{y}^a-2\pa_{v^*}\pa_t\mathring{u}\pa_{t}\mathring{y}^a=\frac 12\pa_{\Ls}(\widehat{h}_1^{v^*a})+\frac{2\widehat{h}^{v^*a}_1}{r}+\widehat{\slashed{\nabla}}_c\widehat{h}_1^{ac}+O(\frac{\varepsilon}{r^{3+\sigma}}).
\]
Along an integral curve $(v^*(s), u^*(s), \widehat{y}^a(s))$ of the vector field $\widetilde{L}$, we have the following equation with $H=\int_{s}^\infty\!2\pa_{v^*}\pa_t\mathring{u}\,d\eta=O(\frac{1}{r^{\sigma}})$
\begin{equation*}
	\frac{d}{ds}\Big(e^H\pa_t \mathring{y}^a\Big)=e^H\Big(\frac 12\pa_{\Ls}(\widehat{h}_1^{v^*a})+\frac{2\widehat{h}^{v^*a}_1}{r}+\widehat{\slashed{\nabla}}_c\widehat{h}_1^{ac}\Big)
	O(\frac{\varepsilon}{r^{3+\sigma}})
	=\frac 12\pa_{\Ls}(\widehat{h}_1^{v^*a})+\frac{2\widehat{h}^{v^*a}_1}{r}+\widehat{\slashed{\nabla}}_c\widehat{h}_1^{ac}+O(\frac{\varepsilon}{r^{3+\sigma}}).
\end{equation*}
Using the asymptotics results in Remark \ref{rem:limit} and integrating backward along the integral curve we conclude
\begin{equation}\label{eq:badderiofangular}
	\pa_t\mathring{y}^a=\frac12 \widehat{h}_1^{v^*a}+\frac{r^*}{2}\widehat{\slashed{\nabla}}_c\widehat{h}_1^{ac}+O(\frac{\varepsilon}{r^{2+\sigma}}).
\end{equation}
Once we have \eqref{eq:badderiofangular}, Proposition \ref{prop:eikonalintro}, Remark \ref{rem:Ithree} and Proposition \ref{prop:highderiofangular} at our disposal, we are able to express $\pa_t Z\mathring{y}^a$ with $Z\in\{\pa_t, \Omega_{ij}\}$ more precisely. In fact we have the following equations
\begin{align*}
	\pa_{\widetilde{L}}\pa_t^2\mathring{y}^a&=-h_{1\pa_t}(\pa u, \pa\pa_t \mathring{y}^a)-g(\pa\pa_t\mathring{u}, \pa\pa_t \mathring{y}^a)+\pa_t\bigl(\pa_{\widetilde{L}}\pa_t \mathring{y}^a\bigr)\\
	&=4\pa_{v^*}\pa_t\mathring{u}\pa_{t}^2\mathring{y}^a+\frac 12\pa_{\Ls}\pa_t(\widehat{h}_1^{v^*a})+\frac{2\pa_t(\widehat{h}^{v^*a}_1)}{r}
+\pa_t(\widehat{\slashed{\nabla}}_c\widehat{h}_1^{ac})
+O(\frac{\varepsilon}{r^{3+\sigma}}),\\
	\pa_{\widetilde{L}}\pa_t\Omega \mathring{y}^a\!\!&=-h_{1\Omega}(\pa u, \pa\pa_t \mathring{y}^a)\!-g(\pa\Omega\tilde{u}, \pa\pa_t \mathring{y}^a)\!+\Omega\bigl(\pa_{\widetilde{L}}\pa_t \mathring{y}^a\bigr)\!
	+\frac 12\pa_{\Ls}\Omega(\widehat{h}_1^{v^*\!a})+
\frac{2\Omega(\widehat{h}^{v^*\!a}_1)\!}{r}+\Omega(\widehat{\slashed{\nabla}}_{\!\!c}
\widehat{h}_1^{ac})+O(\frac{\varepsilon}{r^{3+\sigma}}).
\end{align*}
where we used \eqref{eq:highbadderi}, \eqref{eq:wcc_L} and \eqref{eq:wcc:Y}. Then we repeat the proof of \eqref{eq:badderiofangular}.
\end{proof}

\subsection{Construction of areal coordinate}
We now construct the areal coordinate such that
$\det [\overline{g}_{ab}] = \overline{r}^4 \mathfrak{q}$, where $\mathfrak{q}(\overline{y}^a)$ is the determinant of the unit sphere metric $\overline{q}_{ab}$
associated with the angular coordinates $\overline{y}^a$.  Since $\overline{g}$ takes the Bondi-Sachs form, we have $\overline{g}_{ac}\overline{g}^{cb}=\delta_a^b$ and thus $\det[\overline{g}_{ab}]=1/\det[\overline{g}^{ab}]$. Then we define
\begin{equation}\label{defofr} \overline{r}=\big({\det[\overline{g}^{ab}]\det[\overline{q}_{ab}]}\big)^{-1/4}
=\big({\det[\overline{q}_{ac}\overline{g}^{cb}]}\big)^{-1/4}.
\end{equation}
By Proposition \ref{prop:highderiofangular} and \ref{prop:badderihigh} we have $\pa_{\widehat{y}^p} \overline{y}^a=\bigl(O_1(r^{-2-\sigma}), (\widehat{h}_1^{v^*a}+r^*\widehat{\nabla}_c\widehat{h}_1^{ca}+O_1(r^{-2-\sigma}))/2, \delta_c^a+\pa_{\widehat{y}^c}\mathring{y}^a\bigr)$ and
\begin{align*}	\overline{g}^{ab}&=\widehat{g}^{pq}\frac{\pa\overline{y}^a}{\pa\widehat{y}^p}\frac{\pa\overline{y}^b}{\pa\widehat{y}^q}=\frac{1}{r^2}\Bigl((1-\frac{M}{r})\widehat{q}^{\,ab}+r^2\widehat{h}_1^{ab}+\widehat{q}^{\,ac}\frac{\pa \mathring{y}^b}{\pa \widehat{y}^c}+\widehat{q}^{\,bc}\frac{\pa \mathring{y}^a}{\pa \widehat{y}^c}+O_1(\frac{\varepsilon}{r^{2+\sigma}})\Bigr),\\ \overline{q}_{ab}&=\widehat{q}_{ab}+\mathring{y}^c\pa_{\widehat{y}^c}(\widehat{q}_{ab})+O_2(\frac{\varepsilon}{r^{2+2\sigma}}).
\end{align*}
Then we find
\begin{align*}
	\overline{q}_{ac}\overline{g}^{cb}&=\frac{1}{r^2}\Bigl((1-\frac Mr)I_2+r^2\widehat{q}_{ac}\widehat{h}_1^{cb}+\widehat{q}_{ac}\widehat{q}^{\,cd}\frac{\pa \mathring{y}^b}{\pa \widehat{y}^d}+\widehat{q}_{ac}\widehat{q}^{\,bd}\frac{\pa \mathring{y}^c}{\pa \widehat{y}^d}+\widehat{q}^{\,cb}\mathring{y}^d\pa_{\widehat{y}^d}(\widehat{q}_{ac})+O_1(\frac{\varepsilon}{r^{2+\sigma}})\Bigr)\\
	&=\frac{1-M/r}{r^2}\Bigl(I_2+r^2\widehat{q}_{ac}\widehat{h}_1^{cb}+\widehat{q}_{ac}\widehat{q}^{\,cd}\frac{\pa \mathring{y}^b}{\pa \widehat{y}^d}+\widehat{q}_{ac}\widehat{q}^{\,bd}\frac{\pa \mathring{y}^c}{\pa \widehat{y}^d}+\widehat{q}^{\,cb}\mathring{y}^d\pa_{\widehat{y}^d}(\widehat{q}_{ac})+O_1(\frac{\varepsilon}{r^{2+\sigma}})\Bigr).
\end{align*}
Therefore
\begin{equation*} \det[\overline{q}_{ac}\overline{g}^{cb}]=\frac{(1\!-\!M{}_{\!}/{}_{\!}r)^2\!}{r^4}
\Bigl(1\!+r^2\widehat{q}_{ab}\widehat{h}_1^{ab}\!+2\widehat{q}_{ac}\widehat{q}^{\,cd}\frac{\pa \mathring{y}^a}{\pa \widehat{y}^d}+\widehat{q}^{\,ab}\mathring{y}^d
\pa_{\widehat{y}^d}(\widehat{q}_{ab})-\frac 12r^2\widehat{q}_{ac}\widehat{h}_1^{cb}r^2\widehat{q}_{bd}\widehat{h}_1^{da}\!
+O_1(\frac{\varepsilon}{r^{2+\sigma}})\Bigr),
\end{equation*}
and
\begin{equation}\label{eq:expforr}
\overline{r}\!=r+\!\frac{M\!}{2}+\frac{3M^2\!\!\!\!}{8r}+r\Bigl(\!-\frac 14r^2\widehat{q}_{ab}\widehat{h}_1^{ab}\!-\frac 12\widehat{q}_{ac}\widehat{q}^{\,cd}\frac{\pa \mathring{y}^a}{\pa \widehat{y}^d}-\frac 14\widehat{q}^{\,ab}\mathring{y}^d\pa_{\widehat{y}^d}(\widehat{q}_{ab})\!\Bigr)\\
		+r\Big(\frac 18r^2\widehat{q}_{ac}\widehat{h}_1^{cb}r^2\widehat{q}_{bd}\widehat{h}_1^{da}\!
+O_1(\frac{\varepsilon}{r^{2+\sigma}})\!\Big).
\end{equation}
From the expression of $\overline{r}$ \eqref{eq:expforr}, it is clear that
\begin{equation}\label{eq:estimateofderiofr}
	\pa_{\Ls}\overline{r}=1-\frac{M}{r}+O(\frac{\varepsilon}{r^{1+\sigma}})\qquad\text{and}\qquad\pa_{\widehat{y}^a}\overline{r}=O(\frac{\varepsilon}{r^\sigma}).
\end{equation}
As for $\pa_{\Lbs}\overline{r}$, we need more delicate analysis. First given that $\widehat{q}^{\,ab}\pa_{\widehat{y}^d}(\widehat{q}_{ab})=2\widehat{\Gamma}_{dc}^c$ and using \eqref{eq:badderihigh} we get
\begin{align*}
	&-\frac 12\widehat{q}_{ac}\widehat{q}^{\,cd}\pa_{\Lbs}\pa_{\widehat{y}^d}\mathring{y}^a-\frac 14\widehat{q}^{\,ab}\pa_{\Lbs}(\mathring{y}^d)\pa_{\widehat{y}^d}(\widehat{q}_{ab})=-\frac 12(\pa_{\Lbs}\pa_{\widehat{y}^d}\mathring{y}^d+\widehat{\Gamma}_{dc}^c\pa_{\Lbs}\mathring{y}^d)\\
	&\qquad=-\frac 12\pa_{\Lbs}(\widehat{\slashed{\nabla}}_d\mathring{y}^d)=-\frac12\Bigl(\widehat{\slashed{\nabla}}_a\widehat{h}^{v^*a}_1+r^*\widehat{\slashed{\nabla}}_a\widehat{\slashed{\nabla}}_b\widehat{h}_1^{ab}\Bigr)+O(
	\frac{\varepsilon}{r^{2+\sigma}}).
\end{align*}
In order to handle the term $\pa_{\Lbs}(r^2\widehat{q}_{ab}\widehat{h}_1^{ab})$, we need the following lemma
\begin{lemma}
	We have
	\begin{equation}
		\pa_{\Lbs}(r^2\widehat{q}_{ab}\widehat{h}_1^{ab})=-\pa_{\Lbs}(\slashed{\tr}h^1)+\pa_{\Lbs}(r^2\widehat{q}_{ac}\widehat{h}_1^{cb}r^2\widehat{q}_{bd}\widehat{h}_1^{da})+O(\frac{\varepsilon}{r^{2+\sigma}}).
	\end{equation}
\end{lemma}
\begin{proof}
	As we can see in the proof of Lemma \ref{lem:deriofLLb}, we have
	\[
	h_1^{\alpha\beta}=-m^{\alpha\mu}h^1_{\mu\nu}m^{\nu\beta}+m^{\alpha\alpha'}(\frac Mr\delta_{\alpha'\mu}+h^1_{\alpha'\mu})m^{\mu\nu}(\frac Mr\delta_{\nu\beta'}+h^1_{\nu\beta'})m^{\beta'\beta}
	+O_1(\frac{\varepsilon}{r^{2+\sigma}}).
	\]
	Then we repeat the calculation in the proof of Lemma \ref{lem:deriofLLb} and obtain
	\begin{align*}
		\tr h^1&=m^{\alpha\beta}h^1_{\alpha\beta}=-(1+\frac{2M}{r})m_{\alpha\beta}h_1^{\alpha\beta}+\frac{4M}{r^2}+\frac{4M}{r}h^1_{00}+m_{\alpha\beta}m_{\mu\nu}h_1^{\alpha\mu}h_1^{\beta\nu}+O(\frac{\varepsilon}{r^{2+\sigma}})\\
		&=-\bigl(1\!+\!\frac {2M\!}{r}\bigr)\bigl(-h_{1L\Lb}+r^2\widehat{q}_{ab}\widehat{h}_1^{ab}\bigr)
+\frac{4M\!}{r^2}+\frac{M\!}{r}h^1_{\Lb\Lb}+\frac{2M\!}{r}h^1_{L\Lb}
+r^2\widehat{q}_{ac}\widehat{h}_1^{cb}r^2\widehat{q}_{bd}\widehat{h}_1^{da}\!
+\frac{1}{2}(h_{1\Ls\Lbs}\!)^2\!+O_1(\frac{\varepsilon}{r^{2+\sigma}}),\\
		h^1_{L\Lb}&=-h_{1L\Lb}-\frac 12(h^1_{L\Lb})^2-\frac{2M}{r^2}-\frac Mrh^1_{\Lb\Lb}+O_1(\frac{\varepsilon}{r^{2+\sigma}}).	
	\end{align*}
	where we used the notations $h^1_{UV}=h^1_{\alpha\beta}U^\alpha V\beta$ and $h_{1UV}=h_1^{\alpha\beta}U_\alpha V_\beta$. In view of the facts that $h^1_{UV}=-h_{1UV}+O_1(r^{-1-\sigma})$ if $U_\alpha=m_{\alpha\beta}U^\beta$ and $V_\alpha=m_{\alpha\beta}V^\beta$ and $h_{1U^{*}V^{*}}=h_{1UV}+O_1(r^{-1-\sigma})$, we conclude that
	\[
	\slashed{\tr}h^1=\tr h^1+h^1_{L\Lb}=-r^2\hat{q}_{ab}\hat{h}_1^{ab}+\frac{2M}{r^2}+r^2\widehat{q}_{ac}\widehat{h}_1^{cb}r^2\widehat{q}_{bd}\widehat{h}_1^{da}+O_1(\frac{\varepsilon}{r^{2+\sigma}}).	
	\]
	Applying $\Lbs$ proves the lemma.
\end{proof}
Therefore by \eqref{eq:wccLb} we conclude that
\begin{equation}\label{eq:badderiofr}
	\pa_{\Lb^*}\overline{r}=-1+\frac Mr+\frac{r\pa_{\Ls}(\widehat{h}_1^{v^*v^*})}{4}+\frac{\widehat{h}_1^{v^*v^*}}{2}-\frac{\widehat{h}_1^{v^*u^*}}{2}-\frac{r^2\widehat{\slashed{\nabla}}_a\widehat{\slashed{\nabla}}_b\widehat{h}_1^{ab}}{2}+O(\frac{\varepsilon}{r^{1+\sigma}}).
\end{equation}

\subsection{The Jacobian}
We now give the Jacobian of the mapping from the coordinates $\widehat{y}^p$ to the Bondi-Sachs coordinates $\overline{y}^p$. According to Propositions \ref{prop:eikonalintro}, \ref{prop:badderiofu}, \ref{prop:highderiofangular}, \ref{prop:badderihigh}, and identities \eqref{eq:estimateofderiofr}, \eqref{eq:badderiofr}
\begin{equation}\label{eq:jacobian}
	\begin{split}
		\pa_{\widehat{y}^p}u&=\Bigl(O(\frac{\varepsilon}{r^{1+\sigma}}), ~1+\frac{\widehat{h}_1^{v^*u^*}}{2}+O(\frac{\varepsilon}{r^{1+\sigma}}), ~O(\frac{\varepsilon}{r^{\sigma}}), ~O(\frac{\varepsilon}{r^{\sigma}})\Bigr),\\
		\pa_{\widehat{y}^p}\overline{r}&=\Bigl(\frac 12-\frac{M}{2r}+O(\frac{\varepsilon}{r^{1+\sigma}}), ~\frac{1}{2}\pa_{\Lbs}\overline{r}, ~O(\frac{\varepsilon}{r^{\sigma}}), ~O(\frac{\varepsilon}{r^{\sigma}})\Bigr),\\
		\pa_{\widehat{y}^p}\overline{y}^3&=\Bigl(O(\frac{\varepsilon}{r^{2+\sigma}}), ~\frac12 \widehat{h}_1^{v^*2}+\frac{r^*}{2}\widehat{\slashed{\nabla}}_c\widehat{h}_1^{c2}+O(\frac{\varepsilon}{r^{2+\sigma}}), ~1+O(\frac{\varepsilon}{r^{1+\sigma}}), ~O(\frac{\varepsilon}{r^{1+\sigma}})\Bigr),\\
		\pa_{\widehat{y}^p}\overline{y}^4&=\Bigl(O(\frac{\varepsilon}{r^{2+\sigma}}), ~\frac12 \widehat{h}_1^{v^*3}+\frac{r^*}{2}\widehat{\slashed{\nabla}}_c\widehat{h}_1^{c3}+O(\frac{\varepsilon}{r^{2+\sigma}}), ~O(\frac{\varepsilon}{r^{1+\sigma}}), ~1+O(\frac{\varepsilon}{r^{1+\sigma}})\Bigr),\\
	\end{split}
\end{equation}
where \[
\pa_{\Lbs}\overline{r}=-1+\frac Mr+\frac{r\pa_{\Ls}(\widehat{h}_1^{v^*v^*})}{4}+\frac{\widehat{h}_1^{v^*v^*}}{2}-\frac{\widehat{h}_1^{v^*u^*}}{2}-\frac{r^2\widehat{\slashed{\nabla}}_a\widehat{\slashed{\nabla}}_b\widehat{h}_1^{ab}}{2}+O(\frac{\varepsilon}{r^{1+\sigma}}).
\]
Then we have
\begin{equation}
	\begin{split}\label{eq:BSinversemetric}
		\overline{g}^{pq}\pa_{\overline{y}^p}\pa_{\overline{y}^q}\!&=\!-2\Bigl(1\!+\!O(\frac{\varepsilon}{\overline{r}^{1+\sigma}})\Bigr)\pa_u\pa_{\overline{r}}+\Bigl(1\!-\!\frac Mr\!-\!\frac{r\pa_{\Ls}(\widehat{h}_1^{v^*v^*})}{4}\!-\!\frac{\widehat{h}_1^{v^*v^*}}{4}\!+\!\frac{r^2\widehat{\slashed{\nabla}}\!_c\widehat{\slashed{\nabla}}\!_d\widehat{h}_1^{cd}}{2}\!+\!O(\frac{\varepsilon}{\overline{r}^{1+\sigma}})\Bigr)\pa_{\overline{r}}^2\\
		&\quad+2\Bigl(-\frac 12r\widehat{\slashed{\nabla}}_c\widehat{h}_1^{ca}+O(\frac{\varepsilon}{\overline{r}^{2+\sigma}})\Bigr)\pa_{\overline{r}}\pa_{\overline{y}^a}+\frac{1}{r^2}\Bigl((1-\frac Mr)\widehat{q}^{\,ab}+r^2\widehat{h}_1^{ab}+O(\frac{\varepsilon}{\overline{r}^{1+\sigma}})\Bigr)\pa_{\overline{y}^a}\pa_{\overline{y}^b}.
	\end{split}
\end{equation}

\section{The Bondi Mass}
\subsection{The Bondi-Sachs metric}
Now we can establish the following expression for the Bondi-Sachs metric
\begin{prop}\label{prop:BSmetric}
	We have
	\begin{equation}\label{eq:BSmetric}
		\begin{split}
			\overline{g}_{pq}d\overline{y}^p d\overline{y}^q&=-\Bigl(1\!-\!\frac Mr\!-\!\frac{r\pa_{\Ls}(\widehat{h}_1^{v^*v^*})}{4}\!-\!\frac{\widehat{h}_1^{v^*v^*}}{4}\!+\!\frac{r^2\widehat{\slashed{\nabla}}\!_c\widehat{\slashed{\nabla}}\!_d\widehat{h}_1^{cd}}{2}\!+\!O(\frac{\varepsilon}{\overline{r}^{1+\sigma}})\Bigr)du^2-2\Bigl(1\!+\!O(\frac{\varepsilon}{\overline{r}^{1+\sigma}})\Bigr)dud\overline{r}\\
			&\quad+r^2\Bigl((1+\frac Mr)\widehat{q}_{ab}-r^2\widehat{q}_{ac}\widehat{h}_1^{cd}\widehat{q}_{db}
+O(\frac{\varepsilon}{\overline{r}^{1+\sigma}})\Bigr)\Bigl(d\overline{y}^a
-U^adu\Bigr)\Bigl(d\overline{y}^b-U^bdu\Bigr),
		\end{split}
	\end{equation}
	where $U^a=-\frac{r}{2}\widehat{\slashed{\nabla}}_c\widehat{h}_1^{ca}+O(\frac{\varepsilon}{\overline{r}^{2+\sigma}})$.
\end{prop}
\begin{proof}
	Taking matrix inversion to \eqref{eq:BSinversemetric} yields \eqref{eq:BSmetric}.
\end{proof}

\begin{remark}
The existence of a compactification  with smooth future null infinity $\mathscr{I}^+$, i.e., the conformally rescaled metric can extend smoothly across $\mathscr{I}^+$, is a delicate issue as it is sensitive to the choice of conformal factor and the smooth structure near $\mathscr{I}^+$ \cite{W84}. The Bondi-Sachs coordinates constructed above allows us to obtain $C^{1,\delta}$ regularity of $\mathscr{I}^+$ where $0<\delta<1$, which is consistent with the result in \cite{HV20}. More specifically, we let coordinates $(u, \overline{\rho}=\overline{r}^{-1}, \overline{y}^3,\overline{y}^4)$ be a smooth coordinate system near $\mathscr{I}^+$ and choose $\overline{\rho}^2$ as  the conformal factor. Then we find that
	\[
	\overline{\rho}^2\overline{g}(\partial_u, \partial_u)\in C^{3,\delta},\quad \overline{\rho}^2\overline{g}(\partial_u, \partial_{\overline{y}^a})\in C^{2,\delta}\quad\mbox{and}\quad 	\overline{\rho}^2\overline{g}(\partial_u, \partial _{\overline{\rho}}),~\overline{\rho}^2\overline{g}(\partial_{\overline{y}^a}, \partial_{\overline{y}^b})\in C^{1,\delta}
	\] which implies $\overline{\rho}^2\overline{g}\in C^{1,\delta}$ and thus $\mathscr{I}^+$ is of the class $C^{1,\delta}$. We also note that the work by Christodoulou \cite{C02} strongly suggests that the conformally compactification is generically at most of class $C^{1,\alpha}$ with $ \alpha<1$.
\end{remark}
 Then by Proposition \ref{prop:BSmetric} we see that in the Bondi-Sachs coordinates $\overline{y}^\alpha=(u, \overline{r}, \overline{y}^3,\overline{y}^4)$, the metric takes the following Bondi-Sachs form
\[
\overline{g}_{pq}d\overline{y}^p d\overline{y}^q=-{V}\overline{r}^{-1}e^{2\beta} du^2-2 e^{2\beta}dud\overline{r} +\overline{r}^2h_{ab}(d\overline{y}^a-U^adu)(d\overline{y}^b-U^bdu),
\]
where
\begin{equation} V\!\!=\overline{r}-M-\frac{\overline{r}^2\pa_{\Ls}(\widehat{h}_1^{v^*\!v^*\!})\!}{4}-
\frac{\overline{r}\widehat{h}_1^{v^*\!v^*}\!\!\!\!\!}{4}+\frac{\overline{r}^3
\widehat{\slashed{\nabla}}_{\!c}\widehat{\slashed{\nabla}}_{\!d}\widehat{h}_1^{cd}\!\!}{2}
+O(\frac{1}{\overline{r}^{\sigma}}),\quad\text{and}\quad	h_{ab}\!=\overline{q}_{ab}\!-r^2\overline{q}_{ac}\widehat{h}_1^{cd}
\overline{q}_{db}+O(\frac{1}{\overline{r}^{1+\sigma}}).\label{eq:expofVexpofhab}
\end{equation}
Here we used the fact that $r=\overline{r}-M/2+O(\overline{r}^{-\sigma})$ which is implied by the definition of $\overline{r}$ \eqref{eq:expforr}.

The mass aspect $M_A$ and news tensor $N_{ab}$ are defined as follows
\begin{align*}
	M_A(u, \overline{y}^a)&:=-\lim_{\overline{r}\to\infty}\big(V(u, \overline{r}, \overline{y}^a)-\overline{r}\big),\\
	N_{ab}(u, \overline{y}^c)&:=\frac 12\pa_uC_{ab}(u, \overline{y}^c)\quad\text{where}\quad C_{ab}(u, \overline{y}^c):=\lim_{\overline{r}\to\infty}\overline{r}(h_{ab}(u, \overline{r}, \overline{y}^c)-\overline{q}_{ab}(\overline{y}^c)).
\end{align*}
The Bondi mass $M_B$ and radiated energy $E_B$ are defined by
\[
M_B(u)=\frac{1}{4\pi}\int_{\BS^2}M_A(u,\overline{y}^a)d\overline{S}(\overline{y}^a)\quad\text{and}\quad
E_B(u)=\frac{1}{4\pi}\int_{\BS^2}\!|N|^2\,d\overline{S}(\overline{y}^a),
\]
where $d\overline{S}(\overline{y}^{a\!})\!\!=\!\!\sqrt{\!\mathfrak{q}(\overline{y}^{a\!})}d\overline{y}^3d\overline{y}^4$ is the volume form of the unit sphere metric $\overline{q}_{ab}$ and $\lvert \! N\!\rvert^2\!\!=\!\overline{q}^{ac}\overline{q}^{bd}\!N_{ab}N_{cd}$.
\subsection{Bondi mass loss law}\label{subsec:7.2}
We will prove the existence of $M_B(u)$ and $E_B(u)$ and the Bondi mass loss law.
\begin{theorem}
	Let $M_A, N_{ab}, M_B, E_B$ be defined as above, then the Bondi mass is given by
	\begin{equation}
	M_B(u)=\frac{1}{4\pi}\int_{\BS^2}M_A(u,\overline{y}^a)d\overline{q}=M-\frac{1}{16\pi}\int_{\BS^2}\!\int_{-u}^\infty\!n(\eta,\overline{y}^a)\,d\eta d\overline{S}(\overline{y}^a),
	\end{equation}
	and the radiated energy is equal to
	\begin{equation} E_B(u)=\frac{1}{8\pi}\int_{\BS^2}\!n(-u,\overline{y}^a)\,d\overline{S}(\overline{y}^a).
	\end{equation}
where $n(q^*\!\!, \widehat{y}^a)\!=\!\widehat{q}_{ac}\widehat{q}_{bd}V^{ab}\widehat{V}^{cd}\!/2$ with $\widehat{V}^{ab}\!\!=\!\pa_{q^*}\widehat{H}_{1\infty}^{ab}$ defined in \eqref{eq:NewsfunctionM}.They satisfy the Bondi mass loss law
	\begin{equation}
		\frac{d}{du}M_B(u)=-E(u).
	\end{equation}
	Moreover, $M_B(u)\!\to \!M$ as $u\!\to\!\!-\infty$ where $M$ is the ADM mass and $M_B(u)\!\to\!0$ as $u\to\infty$.
\end{theorem}
\begin{proof}
	By Remark \ref{rem:limit} we write  with $q^*=-u^*=r^*-t$
	\begin{align*}
		\widehat{h}_1^{ab}(v^*\!\!, -q^*\!, \widehat{y}^a)&=\frac{\widehat{H}_{1\infty}^{ab}(q^*, \widehat{y}^a)}{r^{*3}}+O(\frac{1}{{r^*}^{3+\sigma}}),\\
		\widehat{h}_1^{v^*v^*\!}(v^*\!\!, -q^*\!, \widehat{y}^a)&=-\frac{2}{r^*}\!\!\int_{q^*}^\infty\!\!\!\ln\Bigl(
\frac{v^*\!\!+\eta}{u^*\!\!+\eta}\Bigr)n(\eta, \widehat{y}^a)\,d\eta+\frac{\widehat{H}_{1\infty}^{v^*v^*}\!(q^*\!, \widehat{y}^a)}{r^*}+\frac{2M\!}{r^*}\Bigl(1\!-\!\chi^e(q^*)\Bigr)+O(\frac{1}{{r^*}^{1+\sigma}}),
	\end{align*}
	 where $\chi^e(s)\!=\!1$ when $s\!\geq\! 2$, $\chi^e(s)\!=\!0$ when $s\!\leq\! 1$. Plugging these expressions for $\widehat{h}^{v^*\!v^*}_1\!\!$ and $\widehat{h}_1^{ab}$ into \eqref{eq:expofVexpofhab} gives
	\begin{equation*}
		V\!\!=\overline{r}-M+\frac 12\int_{q^*}^\infty\!\!\!\! n(\eta,\widehat{y}^a)\,d\eta+ \frac{\widehat{\slashed{\nabla}}_{\!c}\widehat{\slashed{\nabla}}_{\!d}\widehat{H}^{cd}(q^*, \widehat{y}^a)}{2}+O(\frac{1}{\overline{r}^{\sigma}}),\qquad\text{and}\qquad h_{ab}=\overline{q}_{ab}-\frac{\overline{q}_{ac}\widehat{H}^{cd}
\overline{q}_{db}}{\overline{r}}+O(\frac{1}{\overline{r}^{1+\sigma}}).
	\end{equation*}
	Therefore we conclude
	\begin{multline*}
		M_A(u,\overline{y}^a)=M-\frac 12\lim_{\overline{r}\to\infty}\Bigl(\int_{\mathring{u}-u}^\infty\!n(\eta,\overline{y}^a-\mathring{y}^a)\,d\eta+\widehat{\slashed{\nabla}}_c\widehat{\slashed{\nabla}}_d\widehat{H}^{cd}(\mathring{u}-u, \overline{y}^a-\mathring{y}^a)\Bigr)\\
		=M-\frac 12\lim_{\overline{r}\to\infty}\Bigl(\int_{\mathring{u}-u}^\infty\!\!\!\!\!\!
n(\eta,\overline{y}^a\!-\mathring{y}^a)\,d\eta
+\overline{\slashed{\nabla}}_{\!c}\overline{\slashed{\nabla}}_{\!d}\widehat{H}^{cd}(\mathring{u}-u, \overline{y}^a\!-\mathring{y}^a)\Big)
		=M-\frac 12\int_{-u}^\infty\!\!\!\!\!\!
n(\eta,\overline{y}^a)\,d\eta
-\frac{\overline{\slashed{\nabla}}_{\!c}\overline{\slashed{\nabla}}_{\!d}\widehat{H}^{cd}(-u, \overline{y}^a)}{2}.
	\end{multline*}
	where we used the fact that $\overline{\slashed{\nabla}}_c\overline{\slashed{\nabla}}_dH^{cd}=\widehat{\slashed{\nabla}}_c\widehat{\slashed{\nabla}}_d\widehat{H}^{cd}+O(\overline{r}^{-\sigma})$. Then the Bondi mass $M_B(u)$ is given by
	\begin{equation}\label{eq:Bondimass} M_B(u)=\frac{1}{4\pi}\int_{\BS^2}M_A(u,\overline{y}^a)d\overline{q}=M-\frac{1}{16\pi}\int_{\BS^2}\!\int_{-u}^\infty\!n(\eta,\overline{y}^a)\,d\eta d\overline{S}(\overline{y}^a),
	\end{equation}
	where we used the fact that the integral of the spherical divergence $\overline{\slashed{\nabla}}_c\overline{\slashed{\nabla}}_d\widehat{H}^{cd}$ over the sphere is $0$.
	Since
	\begin{equation*}
		C_{ab}(u, \overline{y}^a)=-\lim_{\overline{r}\to\infty}\overline{q}_{ac}\widehat{H}^{cd}(\mathring{u}-u, \overline{y}^a-\mathring{y}^a)\overline{q}_{db}=-\overline{q}_{ac}\widehat{H}^{cd}(-u, \overline{y}^a)\overline{q}_{db},
	\end{equation*}
	with $\widehat{V}^{ab}(q^*, \widehat{y}^a)=\pa_{q^*}\widehat{H}^{ab}(q^*,\widehat{y}^a)$ defined in \eqref{eq:NewsfunctionM} we obtain
	\begin{equation*}
		N_{ab}(u, \overline{y}^a)=\frac{1}{2}\overline{q}_{ac}\overline{q}_{db}	\widehat{V}^{cd}(-u, \overline{y}^a).
	\end{equation*}
Therefore with \eqref{eq:NewsfunctionM} the radiated energy flux is equal to
	\begin{equation}\label{eq:radiatedenergy}
		\begin{split}
		E_B(u)&=\frac{1}{4\pi}\int_{\BS^2}\!|N|^2\,d\overline{S}(\overline{y}^a)=\frac{1}{4\pi}\int_{\BS^2}\!\overline{q}^{ac}\overline{q}^{bd}N_{cd}N_{ab}\,d\overline{S}(\overline{y}^a)\\
		&=\frac{1}{16\pi}\int_{\BS^2}\!\overline{q}_{ac}\overline{q}_{bd}\widehat{V}^{ab}\widehat{V}^{cd}(-u,\overline{y}^a)\,d\overline{S}(\overline{y}^a)=\frac{1}{8\pi}\int_{\BS^2}\!n(-u,\overline{y}^a)\,d\overline{S}(\overline{y}^a).
		\end{split}
	\end{equation}
	By \eqref{eq:Bondimass} and \eqref{eq:radiatedenergy} we establish the mass loss formula
	\begin{equation}
		\frac{d}{du}M_B(u)=-E_B(u).
	\end{equation}
	Moreover, since $n(\eta, \overline{y}^a)$ is integrable in $\eta$, we see that $M_B(u)\to M$ as $u\to-\infty$. By Proposition \ref{prop:totalmassloss} we know that $\int_{-\infty}^\infty\!E_B(u)\,du=M$ and thus $M_B(u)\to 0$ as $u\to\infty$.
\end{proof}

\bibliographystyle{plain}
\bibliography{references}
\end{document}